\documentclass{CSML}

\def\dOi{11(4:21)2015}
\lmcsheading%
{\dOi}
{1--53}
{}
{}
{Dec.~15, 2014}
{Dec.~31, 2015}
{}

\ACMCCS{[{\bf Theory of computation}]: Models of
  computation---Concurrency---Distributed computing models;
  Logic---Type theory; Design and analysis of algorithms---Distributed
  algorithms---Self-organization; Semantics and reasoning---Program
  constructs---Control primitives; Semantics and reasoning---Program
  semantics; Semantics and reasoning---Program reasoning---Program
  analysis; [{\bf Software and its engineering}]: Software notations
  and tools ---General programming languages---Language types---
  Distributed programming languages\,/\,Functional languages; Software
  notations and tools---Context specific languages; [{\bf Computing
    methodologies}]: Artificial intelligence---Distributed artificial
  intelligence---Cooperation and coordination}

\usepackage{hyperref}

\newcommand{\FORGET}[1]{}

\usepackage{tabularx}
\usepackage{booktabs}
\usepackage{amsmath,amssymb}
\usepackage{graphicx}
\usepackage{epsfig}
\usepackage{color}
\usepackage{verbatim}
\usepackage[dvipsnames]{xcolor}
\usepackage{float}
\usepackage{mathptmx}
\usepackage{datetime}
\usepackage{listings}
\usepackage{subfigure}
\usepackage{stmaryrd}
\usepackage{latexsym}
\usepackage{amsmath,amsfonts,amstext,amssymb,wasysym}
\usepackage{url}
\usepackage{mathpartir}
\usepackage{enumitem}
\usepackage{todonotes}
\usepackage{fancyvrb}
\usepackage{xspace}

\usepackage{tikz}

\usepackage{float}
\usepackage{afterpage}

\newfloat{listing}{tbp}{lol}
\floatname{listing}{Listing}
\floatstyle{boxed}
\restylefloat{listing}

\newif\ifpgf\pgftrue  

\lstdefinestyle{mystyle}{flexiblecolumns=true,showstringspaces=false,keepspaces=true,basewidth={0em,0em},,basicstyle=\sffamily,commentstyle=\itshape,stringstyle=\sffamily}

\newcommand{\BNFcce}{{\bf ::=}}
\newcommand{\BNFmid}{\;\bigr\rvert\;}

\newcommand{\PROGRAM}{\textbf{P}}
\newcommand{\FUNCTION}{\texttt{D}}
\newcommand{\e}{\texttt{e}}
\newcommand{\dname}{\texttt{d}}
\newcommand{\fname}{\texttt{f}}
\newcommand{\xname}{\texttt{x}}

\newcommand{\oname}{\texttt{b}}

\newcommand{\boolVar}{\texttt{c}\xspace}

\newcommand{\defK}{\texttt{def}}
\newcommand{\isK}{\texttt{is}}

\newcommand{\defKK}[4]{\defK \; #1 \;#2 (#3) \;\isK\; #4}

\newcommand{\skiptransition}{\\[15pt]}

\newcommand{\ruleNameSize}[1]{{\scriptsize #1}}

\newcommand{\senv}[1]{\texttt{\##1}}

\newcommand{\deftt}[3]{\texttt{def #1(#2) is #3}}

\newcommand{\spreadtt}[2]{\texttt{{\char '173} #1 : #2 {\char '175}}}
\newcommand{\startt}{\texttt{@}}

\newcommand{\Topo}{\tau}
\newcommand{\Sens}{\Sigma}
\newcommand{\Envi}{E}
\newcommand{\EnviS}[2]{#1,#2}
\newcommand{\Cfg}{N}
\newcommand{\Field}{F}
\newcommand{\SystS}[2]{\langle #1;#2\rangle}
\newcommand{\devset}{I}

\newcommand{\nettran}[3]{#1\xrightarrow{#2} #3}
\newcommand{\mapupdate}[2]{#1[#2]}
\newcommand{\netArrIdStar}{\stackrel{\overline{\deviceId}}{\netArrStar}}

\newcommand{\netArrIdFairStarN}[1]{\netArrIdStar_{#1}}

\newcommand{\wfn}[1]{\textit{WFE}(#1)}

\newcommand{\surfaceTyping}[3]{
  \begin{array}{l@{\;}c}
    \stackrel{~}{{\tiny \textrm{[#1]}}} & #2 \\ \hline 
    \multicolumn{2}{c}{#3}
  \end{array}
}
\newcommand{\nullsurfaceTyping}[2]{
  \surfaceTyping{#1}{}{#2}
}

\newcommand{\surExpTypJud}[3]{#1\vdash #2 : #3}
\newcommand{\surFunTypJud}[3]{#1\vdash #2 : #3}

\newcommand{\SurTypEnv}{\mathcal{T}}

\newcommand{\typeofNAME}{\textit{type}}
\newcommand{\typeof}[1]{\typeofNAME(#1)}

\newcommand{\anyvalue}{\texttt{v}}
\newcommand{\anyvaluealt}{\texttt{u}}
\newcommand{\anyvaluebis}{\texttt{w}}
\newcommand{\anyvalueInNC}[2]{\anyvalue_{#1 (\textrm{in }#2)}}

\newcommand{\groundvalue}{\texttt{g}}

\newcommand{\snsname}{s}

\newcommand{\main}{\texttt{main}}

\newcommand{\starK}{\texttt{@}}
\newcommand{\progK}[2]{#1(\starK,#2)}
\newcommand{\spreadThree}[3]{\{#1:\progK{#2}{#3}\}}

\newcommand{\lengthOf}[1]{\textit{length}(#1)}

\newcommand{\signature}{\textit{t-sig}}
\newcommand{\signatureOf}[1]{\signature(#1)}

\newcommand{\anytype}{\texttt{T}}

\newcommand{\lsempar}{[\![}
\newcommand{\rsempar}{]\!]}
\newcommand{\semOf}[1]{\lsempar{#1}\rsempar}
\newcommand{\lowerBound}{\textstyle{\bigwedge}}

\newcommand{\wfEq}{=}
\newcommand{\wfEqOf}[1]{\wfEq_{#1}}

\newcommand{\wfLt}{<}
\newcommand{\wfLtOf}[1]{\wfLt_{#1}}
\newcommand{\wfLe}{\le}
\newcommand{\wfLeOf}[1]{\wfLe_{#1}}

\newcommand{\keywfEq}{=^1}
\newcommand{\keywfEqOf}[1]{\keywfEq_{#1}}

\newcommand{\keywfLt}{<^1}
\newcommand{\keywfLtOf}[1]{\keywfLt_{#1}}
\newcommand{\keywfLe}{\le^1}
\newcommand{\keywfLeOf}[1]{\keywfLe_{#1}}

\newcommand{\dom}{\textit{dom}}
\newcommand{\domOf}[1]{\dom(#1)}

\newcommand{\noetherianOf}[1]{\noetherianOf(#1)}

\newcommand{\netframe}{\Envi}

\newcommand{\netframeName}{\textit{environment}}
\newcommand{\netframeOf}[1]{\netframeName(#1)}

\newcommand{\deviceIdSet}{\textbf{I}}
\newcommand{\nodeSubSet}{\textbf{D}}

\newcommand{\stableNodeSet}{\textbf{S}}
\newcommand{\sourceNodeSubSet}{\textbf{M}}
\newcommand{\network}{\Cfg}

\newcommand{\snsFun}{\sigma}
\newcommand{\snsFunFor}[1]{\snsFun_{#1}}

\newcommand{\frontier}{\textit{frontier}}
\newcommand{\frontierOfIn}[2]{\frontier_{#1}(#2)}

\newcommand{\netopsemRule}[3]{\surfaceTyping{#1}{#2}{#3}}

\newcommand{\vtree}{\theta}
\newcommand{\vtreealt}{\eta}
\newcommand{\vtreeInNC}[2]{\theta_{#1( \textrm{in }#2)}}

\newcommand{\senstate}{\sigma}
\newcommand{\genmap}[2]{#1\rhd#2}

\newcommand{\act}{\ell}
\newcommand{\envact}{\epsilon}
\newcommand{\netArr}{\rightarrow}

\newcommand{\netArrWithAct}[1]{\stackrel{#1}{\netArr}}

\newcommand{\netArrStar}{\Longrightarrow}
\newcommand{\netArrActStar}{\stackrel{\overline{\act}}{\netArrStar}}

\newcommand{\netArrPresStar}{\netArrStar}

\newcommand{\fair}{\textit{fair}}

\newcommand{\netArrFairStarN}[1]{\netArrStar_{#1}}

\newcommand{\actDevInOf}[1]{\envact}

\newcommand{\actDevOutOf}[1]{\envact}

\newcommand{\actEdgeInOf}[1]{\envact}

\newcommand{\actEdgeOutOf}[1]{\envact}

\newcommand{\actSnsUpdOf}[1]{\envact}

\newcommand{\emain}{\e_{\main}}

\newcommand{\opFunOf}[1]{\semOf{#1}}
\newcommand{\opApply}[2]{#1(#2)}

\newcommand{\bsopsem}[4]{#1;#2\vdash #3\Downarrow #4}
\newcommand{\deviceId}{\iota}

\newcommand{\premise}{\mathbf{\pi}}
\newcommand{\premiseNum}[1]{\premise_{#1}}
\newcommand{\premiseNumOf}[2]{\premiseNum{#1}(#2)}
\newcommand{\conclusion}{\mathbf{\rho}}
\newcommand{\conclusionOf}[1]{\conclusion(#1)}
\newcommand{\vroot}{\mathbf{\rho}}
\newcommand{\vrootOf}[1]{\vroot(#1)}
\newcommand{\substitution}[2]{#1:=#2}
\newcommand{\applySubstitution}[2]{#1[#2]}

\newcommand{\bsWFVT}[3]{\textit{WTVT}(#1,#2,#3)}
\newcommand{\bsWSVT}[3]{\textit{WSVT}(#1,#2,#3)}

\newcommand{\canonicalTopNAME}{\texttt{pt}}
\newcommand{\canonicalTopFor}[1]{\canonicalTopNAME[#1]}
\newcommand{\canonicalTop}[2]{\canonicalTopFor{#1}(#2)}

\newcommand{\ctdNAME}{\texttt{sd}}

\newcommand{\ctdOfFor}[2]{\ctdNAME[#1,#2]}

\newcommand{\sumOrNAME}{\texttt{sum\_or}}
\newcommand{\addToFstNAME}{\texttt{add\_to\_1st}}
\newcommand{\ptRealBoolNAME}{\texttt{\canonicalTopNAME\_\topNumK\_\trueK}}
\newcommand{\ptRealRealNAME}{\texttt{\canonicalTopNAME\_\topNumK\_\topNumK}}
\newcommand{\spSumOrNAME}{\texttt{\ctdNAME\_sum\_or}}
\newcommand{\spAddToFstNAME}{\texttt{\ctdNAME\_add\_to\_1st}}

\newcommand{\progressiveSortSigNAME}{progressive}
\newcommand{\ProgressiveSortSigNAME}{Progressive}
\newcommand{\progressiveSortSigOf}[1]{\textit{\progressiveSortSigNAME}(#1)}

\newcommand{\pureFunOfNAME}{\textit{fun}}
\newcommand{\pureFunOf}[1]{\pureFunOfNAME(#1)}
\newcommand{\myldots}{...}

\newcommand{\diffusionText}{diffusion}
\newcommand{\diffusionsText}{diffusions}
\newcommand{\DiffusionText}{Diffusion}
\newcommand{\DiffusionsText}{Diffusions}

\newcommand{\stabilisingText}{stabilising}

\newcommand{\prestabilisingText}{prestabilising}

\newcommand{\possiblyprestabilisingText}{possibly prestabilising}
\newcommand{\certainlyprestabilisingText}{certainly prestabilising}

\newcommand{\groundType}{\texttt{G}}
\newcommand{\boolType}{\texttt{bool}}
\newcommand{\numNAME}{real}
\newcommand{\numType}{\texttt{\numNAME}}
\newcommand{\mkpairType}[2]{\texttt{<}#1,#2\texttt{>}}
\newcommand{\numNAMEs}{reals}

\newcommand{\mkpair}[2]{\texttt{<}#1,#2\texttt{>}}
\newcommand{\fstK}{\texttt{fst}}
\newcommand{\sndK}{\texttt{snd}}
\newcommand{\leftKey}{\textit{key}}
\newcommand{\leftKeyOf}[1]{\leftKey(#1)}

\newcommand{\notK}{\texttt{not}}

\newcommand{\orK}{\texttt{or}}
\newcommand{\trueK}{\texttt{TRUE}}
\newcommand{\falseK}{\texttt{FALSE}}
\newcommand{\topNumK}{\texttt{POSINF}}
\newcommand{\botNumK}{\texttt{NEGINF}}
\newcommand{\condExpr}[3]{#1\texttt{?}#2\texttt{:}#3}

\newcommand{\progressionOpText}{\diffusionText}
\newcommand{\progressionOp}{\textit{\progressionOpText}}
\newcommand{\progressionOpOf}[1]{\progressionOp(#1)}

\newcommand{\deterministicText}{deterministic}
\newcommand{\deterministic}{\textit{\deterministicText}}
\newcommand{\deterministicOf}[1]{\deterministic(#1)}

\newcommand{\minimalsNAME}{\textit{minimals}}
\newcommand{\minimalsOf}[1]{\minimalsNAME(#1)}

\newcommand{\stabilisingOpText}{stabilising}
\newcommand{\stabilisingOp}{\textit{\stabilisingOpText}}
\newcommand{\stabilisingOpOf}[2]{\stabilisingOp(#1,#2)}

\newcommand{\rsignatures}{\textit{s-sigs}}
\newcommand{\rsignaturesOf}[1]{\rsignatures(#1)}

\newcommand{\stabilisingSignatures}{\textit{stb-s-sigs}}
\newcommand{\stabilisingSignaturesOf}[1]{\stabilisingSignatures(#1)}

\newcommand{\asignatures}{\textit{a-s-sigs}}
\newcommand{\asignaturesOf}[1]{\asignatures(#1)}

\newcommand{\Ppsignatures}{\textit{!-s-sigs}}
\newcommand{\PpsignaturesOf}[1]{\Ppsignatures(#1)}

\newcommand{\psignatures}[1]{#1\textit{-s-sigs}}
\newcommand{\psignaturesOf}[2]{\psignatures{#1}(#2)}

\newcommand{\mostSpecific}{\textit{ms}}
\newcommand{\mostSpecificOf}[2]{\mostSpecific(#1,#2)}

\newcommand{\rtype}{\texttt{S}}

\newcommand{\rsub}{\le}
\newcommand{\rsubof}[2]{#1\rsub#2}

\newcommand{\rtypeofNAME}{\textit{sort}}
\newcommand{\rtypeof}[1]{\rtypeofNAME(#1)}

\newcommand{\atypeofNAME}{\textit{a-sort}}
\newcommand{\atypeof}[1]{\atypeofNAME(#1)}

\newcommand{\RefTypEnv}{\mathcal{S}}

\newcommand{\refofNAME}{\textbf{sorts}}
\newcommand{\refof}[1]{\refofNAME(#1)}

\newcommand{\refsigofNAME}{\textbf{sort-signatures}}
\newcommand{\refsigof}[1]{\refsigofNAME(#1)}

\newcommand{\progrefsigofNAME}{\textbf{progressive-sort-signatures}}
\newcommand{\progrefsigof}[1]{\progrefsigofNAME(#1)}

\newcommand{\refSupOfNAME}{\textit{sup}}
\newcommand{\refSupOf}[2]{\refSupOfNAME_{\rsub}(#1,#2)}

\newcommand{\topValOfNAME}{\top}
\newcommand{\topValOf}[1]{\topValOfNAME_{#1}}

\newcommand{\trueType}{\texttt{true}}
\newcommand{\falseType}{\texttt{false}}

\newcommand{\nnumType}{\texttt{nr}}
\newcommand{\znnumType}{\texttt{znr}}
\newcommand{\znumType}{\texttt{zr}}
\newcommand{\npnumType}{\texttt{npr}}
\newcommand{\zpnumType}{\texttt{zpr}}
\newcommand{\pnumType}{\texttt{pr}}

\newcommand{\rsubwrhof}[2]{#1\subwrhLe#2}
\newcommand{\rsemsubwrhof}[2]{#1\semsubwrhLe#2}

\newcommand{\ReverseHoare}{\textrm{progressive}}

\newcommand{\subwrhLe}{\le^\ReverseHoare}

\newcommand{\semsubwrhLe}{\subseteq^{\ReverseHoare}}

\newcommand{\powerSet}[1]{\mathcal{P}(#1)}
\newcommand{\setInPowerSet}{\textbf{S}}

\newcommand{\rsubstabof}[2]{#1\substabLe#2}

\newcommand{\STABILISING}{\textrm{stabilising}}

\newcommand{\substabLe}{\le^\STABILISING}

\newcommand{\stabsigofNAME}{\textbf{stabilising-sort-signatures}}
\newcommand{\stabsigof}[1]{\stabsigofNAME(#1)}

\newcommand{\prestabsigofNAME}{\textbf{prestabilising-sort-signatures}}

\newcommand{\psigofNAME}{\pAnn\textrm{-}\prestabsigofNAME}
\newcommand{\psigof}[1]{\psigofNAME(#1)}

\newcommand{\PpsigofNAME}{\PpAnn\textrm{-}\prestabsigofNAME}
\newcommand{\Ppsigof}[1]{\PpsigofNAME(#1)}

\newcommand{\annsigofNAME}{\textbf{annotated-sort-signatures}}
\newcommand{\annsigof}[1]{\annsigofNAME(#1)}

\newcommand{\pannsigofNAME}{\pAnn\textrm{-}\annsigofNAME}
\newcommand{\pannsigof}[1]{\pannsigofNAME(#1)}

\newcommand{\PpannsigofNAME}{\PpAnn\textrm{-}\annsigofNAME}
\newcommand{\Ppannsigof}[1]{\PpannsigofNAME(#1)}

\newcommand{\WpannsigofNAME}{\WpAnn\textrm{-}\annsigofNAME}
\newcommand{\Wpannsigof}[1]{\WpannsigofNAME(#1)}

\newcommand{\SB}[1]{\texttt{[$#1$]}}

\newcommand{\pAnn}{\pi}
\newcommand{\PpAnn}{\texttt{!}}
\newcommand{\WpAnn}{\texttt{?}}

\newcommand{\pAnnSB}{\texttt{[$\pi$]}}
\newcommand{\PpAnnSB}{\texttt{[!]}}
\newcommand{\WpAnnSB}{\texttt{[?]}}

\newcommand{\inputAnnSB}{\texttt{[?]}}

\newcommand{\pAnnScript}[1]{\pAnn#1}

\newcommand{\pAnnScriptSB}[1]{\texttt{[$\pAnnScript{#1}$]}}

\newcommand{\annErasure}[1]{\vert#1\vert}

\newcommand{\artype}{\texttt{A}}
\newcommand{\AnnTypEnv}{\mathcal{A}}
\newcommand{\annExpTypJud}[3]{#1\vdash #2 : #3}
\newcommand{\annFunTypJud}[3]{#1\vdash #2 : #3}

\newcommand{\pOf}[2]{\pAnn\mbox{-}\textit{\prestabilisingText}(#1,#2)}
\newcommand{\PpOf}[2]{\PpAnn\mbox{-}\textit{\stabilisingText}(#1,#2)}
\newcommand{\WpOf}[2]{\WpAnn\mbox{-}\textit{\stabilisingText}(#1,#2)}
\newcommand{\pOfSCRIPT}[3]{\pAnn#1\mbox{-}\textit{\prestabilisingText}(#2,#3)}

\begin{document}

\title[Type-based Self-stabilisation for Computational Fields]{Type-based Self-stabilisation for Computational Fields\rsuper*}

\author[F.\ Damiani]{Ferruccio Damiani\rsuper a}	
\address{{\lsuper a}University of Torino, Italy}	
\email{ferruccio.damiani@unito.it}  
\thanks{{\lsuper a}Ferruccio Damiani has been partially supported by  project HyVar (\textit{www.hyvar-project.eu}) which has received funding from the European Union's Horizon 2020 research and
innovation programme under grant agreement No 644298, by ICT COST Action IC1402 ARVI (\textit{www.cost-arvi.eu}), by ICT COST
Action IC1201 BETTY (\textit{www.behavioural-types.eu}), by the Italian PRIN 2010/2011 project CINA (\textit{sysma.imtlucca.it/cina}), and by Ateneo/CSP project SALT (\textit{salt.di.unito.it}).}	

\author[M.~Viroli]{Mirko Viroli\rsuper b}	
\address{{\lsuper b}University of Bologna, Italy}	
\email{mirko.viroli@unibo.it}  
\thanks{{\lsuper b}Mirko Viroli has been partially supported by the EU
FP7 project ``SAPERE - Self-aware Pervasive Service
Ecosystems'' under contract No. 256873 and by the Italian PRIN
2010/2011 project CINA (\textit{sysma.imtlucca.it/cina}).}




\titlecomment{{\lsuper*}A preliminary version of some of the material presented in this paper has appeared in COORDINATION 2015.}


\begin{abstract}
Emerging network scenarios require the development of solid large-scale situated systems. Unfortunately, the diffusion/aggregation computational processes therein often introduce a source of complexity that hampers predictability of the overall system behaviour. Computational fields have been introduced to help engineering such systems: they are spatially distributed data structures designed to adapt their shape to the topology of the underlying (mobile) network and to the events occurring in it, with notable applications to pervasive computing, sensor networks, and mobile robots. To assure behavioural correctness, namely, correspondence of micro-level specification (single device behaviour) with macro-level behaviour (resulting global spatial pattern), we investigate the issue of self-stabilisation for computational fields.
We present a tiny, expressive, and type-sound calculus of computational
fields, and define sufficient conditions for self-stabilisation, defined as the ability to react to changes in the environment
finding a new stable state in finite time. A type-based approach is used to provide a correct checking procedure for self-stabilisation.
\end{abstract}

\maketitle

\sloppy
\section{Introduction}\label{sec-intro}

\emph{Computational fields} \cite{maclennanFC,tokoro} (sometimes simply \emph{fields} in the following) are an abstraction traditionally used to enact self-organisation mechanisms in contexts including swarm robotics \cite{ProtoSwarm}, sensor networks \cite{proto06a}, pervasive computing \cite{tota}, task assignment \cite{weyns-field}, and traffic control \cite{ClaesHW11}.
They are distributed data structures originated from pointwise events raised in some specific device (i.e., a sensor), and propagating in a whole network region until forming a spatio-temporal data structure upon which distributed and coordinated computation can take place.
Middleware/platforms supporting this notion include TOTA \cite{tota}, Proto \cite{mitproto}, and SAPERE \cite{Zetal-PCS2011,VCO-SAC2009}.
The most paradigmatic example of computational field is the so-called \emph{gradient} \cite{crf,tota,MVFFZ-MONET2012}, mapping each node of the network to the minimum distance from the source node where the gradient has been injected.
Gradients are key to get awareness of physical/logical distances, to project a single-device event into a whole network region, and to find the direction towards certain locations of a network, e.g., for routing purposes.
Several pieces of work have been developed that investigate coordination models supporting fields \cite{tota,spatialcoord-coord2012}, introduce advanced gradient-based spatial patterns \cite{MPV-SASO2012}, study universality and expressiveness \cite{BVD-SCW14}, and develop catalogues of self-organisation mechanisms where gradients play a crucial role \cite{FDMVA-NACO2013}.

As with most self-organisation approaches, a key issue is to try to fill the gap between the system micro-level (the single-node computation and interaction behaviour) and the system macro-level (the shape of the globally established spatio-temporal structure), namely, ensuring that the programmed code results in the expected global-level behaviour.
However, the issue of formally tackling the problem is basically yet unexplored in the context of spatial computing, coordination, and process calculi---some exceptions are \cite{crf,giavitto-autonomic}, which however apply in rather ad-hoc cases.
We note instead that studying this issue will likely shed light on which language constructs are best suited for developing well-engineered self-organisation mechanisms based on computational fields, and to consolidate existing patterns or develop new ones.

In this paper we follow this direction and devise an expressive calculus 
to specify the propagation process of those computational fields for which we can identify a precise mapping between system micro- and macro-level.
The key constructs of the calculus are three: sensor fields (considered as an environmental input), pointwise functional composition of fields, and a form of \emph{spreading} that tightly couples information diffusion and re-aggregation.
The spreading construct is constrained so as to enforce a special ``stabilising-diffusion condition'' that we identified, by which we derive self-stabilisation~\cite{dolev}, that is, the ability of the system running computational fields to reach a stable distributed state in spite of perturbations (changes of network topology and of local sensed data) from which it recovers in finite time.
A consequence of our results is that the ultimate (and stable) state of an even complex computational field can be fully-predicted once the environment state is known (network topology and sensors state).
Still, checking that a field specification satisfies such stabilising-diffusion condition is subtle, since it involves the ability of reasoning about the 
relation holding between inputs and outputs of functions used to propagate information across nodes.
Hence, as an additional contribution, we introduce a type-based approach
that provides
a correct checking procedure for the stabilising-diffusion condition.


The remainder of this paper is organised as follows: Section~\ref{sec-fields} illustrates the proposed linguistic constructs by means of examples; Section \ref{sec-calculus} presents the calculus and formalises the self-stabilisation property; Section~\ref{sec-stabilisingProgression} introduces the stabilising-diffusion condition to constrain spreading in order to ensure self-stabilisation; Section~\ref{sec-properties-stabilising} proves  that the stabilising-diffusion condition guarantees self-stabilisation; Section \ref{sec-calculus-extended} extends the calculus with the pair data structure and provides further examples; Sections~\ref{sec-stabilisation-checking},~\ref{sec-ensuringSelfStabilisation} and~\ref{sec-checkingStabilisingDiffusion} incrementally present a type-based approach for checking the stabilising-diffusion condition and prove that the approach is sound;
Section~\ref{sec-related} discusses related work; and finally Section~\ref{sec-conclusion}  concludes and discusses directions for future work.
The appendices contain the proof of the main results.
A preliminary version of some of the material  presented in this paper appeared in~\cite{ViroliDamiani:COORDINATION-2014}. 


\section{Computational Fields}\label{sec-fields}

From an abstract viewpoint, a computational field is simply a map from nodes of a network to some kind of value.
They are used as a valuable abstraction to engineer self-organisation into networks of situated devices.
Namely, out of local interactions (devices communicating with a small neighbourhood), global and coherent patterns (the computational fields themselves) establish that are robust to changes of environmental conditions.
Such an adaptive behaviour is key in developing system coordination in dynamic and unpredictable environments \cite{OV-KER2011}.

Self-organisation and computational fields are known to build on top of three basic mechanisms \cite{FDMVA-NACO2013}: diffusion (devices broadcast information to their neighbours), aggregation (multiple information can be reduced back into a single sum-up value), and evaporation/decay (a cleanup mechanism is used to reactively adapt to changes).
These mechansisms have been used to synthetise a rather vast set of distributed algorithms \cite{autonomicommunications,ZV-JPCC2011,FDMVA-NACO2013,SpatialIGI2013}.

For instance, these mechanisms are precisely those used to create adaptive and stable \emph{gradients}, which are building blocks of more advanced patterns \cite{FDMVA-NACO2013,MPV-SASO2012}.
A gradient is used to reify in any node some information about the path towards the nearest gradient source. It can be computed by the following process: value $0$ is held in the gradient source; each node executes asynchronous computation rounds in which \emph{(i)} messages from neighbours are gathered and \emph{aggregated} in a minimum value, \emph{(ii)} this is increased by one and is \emph{diffused} to all neighbours, and \emph{(iii)} the same value is stored locally, to replace the old one which \emph{decays}.
This continuous ``spreading process'' stabilises to a so called \emph{hop-count gradient}, storing distance to the nearest source in any node, and automatically repairing in finite time to changes in the environment (changes of topology, position and number of sources).

\subsection{Basic Ingredients}\label{sec-basicIngredients}

Based on these ideas, and framing them so as to isolate those cases where the spreading process actually stabilises, we propose a core calculus to express  computational fields.
Its syntax is reported in Figure~\ref{fig:source:syntax}.
\begin{figure}[!t]{
\centerline{\framebox[\textwidth]{$
\begin{array}{lclr}
 \e & \BNFcce &  \xname \;
     \; \BNFmid \; \snsname
    \; \BNFmid \;  \groundvalue
     \; \BNFmid \;   \condExpr{\e_0}{\e_1}{\e_2}
    \; \BNFmid \;    \fname(\e_1,\myldots,\e_n)
    \;\BNFmid \;
 \spreadThree{\e}{\fname}{\e_1,\myldots,\e_n}
    &   {\mbox{\footnotesize expression}} 
\\[2pt]
     \fname & \BNFcce &  \dname \;\BNFmid\; \oname & {\mbox{\footnotesize function name}}\\[2pt]
     \FUNCTION & \BNFcce & \defKK{\anytype}{\dname}{\anytype_1\;\xname_1,\myldots,\anytype_n\;\xname_n}{\e}
  &   {\mbox{\footnotesize function definition}}\\
 \end{array}
 $}}}
 \caption{Syntax of expressions and function definitions}
\label{fig:source:syntax}
\end{figure}
Our language is typed and (following the general approach used in other languages for spatial computing \cite{VDB-FOCLASA-CIC2013,mitproto}, which the one we propose here can be considered as a core fragment) functional. 

Types $\anytype$  are monomorphic. For simplicity, only  \emph{ground types} $\groundType$ (like \texttt{\numNAME} and \texttt{bool}) are modeled---in Section~\ref{sec-calculus-extended} we will point out that the properties of the calculus are indeed parametric in the set of modeled types  (in particular, we will consider and extension of the calculus that models pairs).
We write $\semOf{\anytype}$ to denote the set of the values of type $\anytype$. We assume that each type $\anytype$ is equipped with a \emph{total order} $\wfLeOf{\anytype}$ over $\semOf{\anytype}$ that is \emph{noetherian}~\cite{Huet:1980:CRA:322217.322230}, i.e.,  there are no infinite ascending chains of  values  $\anyvalue_0\wfLtOf{\anytype}\anyvalue_1\wfLtOf{\anytype}\anyvalue_2\wfLtOf{\anytype}\cdots$.
This implies that  $\semOf{\anytype}$ has a maximum element, that we denote by $\topValOf{\anytype}$.
Each ground type usually comes with a natural ordering (for $\boolType$ we consider $\falseK\wfLeOf{\boolType}\trueK$) which is total and noetherian---though in principle ad-hoc ordering relations could be used in a deployed specification language.

An expression can be a \emph{variable} $\xname$, a \emph{sensor} $\snsname$, a \emph{ground-value} $\groundvalue$, 
a \emph{conditional} $\condExpr{\e_0}{\e_1}{\e_2}$, a \emph{function application} $\fname(\e_1,\ldots,\e_n)$ or a \emph{spreading} $\spreadThree{\e}{\fname}{\e_1,\myldots,\e_n}$.
Variables are the formal parameters of a function. 

Sensors are sources of input produced by the environment, available in each device (in examples, we shall use for them literals starting with symbol ``\senv{}'').
For instance, in a urban scenario we may want to use a crowd sensor \senv{crowd} yielding non-negative real numbers, to represent the perception of crowd level available in each deployed sensor over time \cite{MVFFZ-MONET2012}.
 
Values $\anyvalue$ coincide with  \emph{ground values}  $\groundvalue$ (i.e., values of ground type), like
\numNAMEs\
(e.g., \texttt{1}, \texttt{5.7},... and \topNumK\ and \botNumK\ meaning the maximum and the minimum \numNAME, respectively) and  booleans (\texttt{TRUE}
and \texttt{FALSE}).

A function can be either a built-in function $\oname$ or a user-defined function $\dname$.
Built-in functions include usual mathematical/logical ones, used either
in prefix or infix notation, e.g., to form expressions like
\texttt{2*\senv{crowd}} and \texttt{or(TRUE,FALSE)}.
User-defined functions are declared by a \emph{function definition} $\defKK{\anytype}{\dname}{\anytype_1\;\xname_1,\myldots,\anytype_n\;\xname_n}{\e}$---cyclic
definitions are prohibited, and the 0-ary function \texttt{main} is
the program entry point. 
As a first example of user-defined function consider the following function \texttt{restrict}:  
\[
\deftt{\numNAME\ restrict}{\numNAME\ i, bool \boolVar}{\boolVar\ ? i :
\topNumK}.
\]
%
%
It takes two arguments \texttt{i} and \boolVar, and
yields the former if \boolVar is true, or \topNumK{}
otherwise---as we shall see, because of our semantics \topNumK{}
plays a role similar to an undefined value.  

A \emph{pure function}
$\fname$  is either a built-in function $\oname$, or a user-defined function $\dname$ whose call graph (including $\dname$ itself) does not contain functions  with spreading expressions or sensors in their body. We write $\opFunOf{\fname}$ to denote the (trivial) semantics of a pure-function $\fname$, which is a computable functions that maps a tuple of elements from $\semOf{\anytype_1}$,\myldots, $\semOf{\anytype_n}$ to  $\semOf{\anytype}$, where $\anytype_1,\myldots,\anytype_n$ and $\anytype$ are the types of the arguments and the type of the result of $\fname$, respectively.

As in \cite{VDB-FOCLASA-CIC2013,mitproto}, expressions in our language have a twofold interpretation.
When focussing on the \emph{local} device behaviour, they represent values computed in a node at a given time.
When reasoning about the \emph{global} outcome of a specification instead, they represent whole computational fields: $1$ is the immutable field holding $1$ in each device, \senv{crowd} is the (evolving) crowd field, and so on.

The key construct of the proposed language is \emph{spreading},
denoted by syntax
\mbox{$\spreadThree{\e}{\fname}{\e_1,\myldots,\e_n}$}, where
\texttt{e} is called \emph{source} expression, and
$\progK{\fname}{\e_1,\myldots,\e_n}$ is called \emph{\diffusionText}
expression. In a \diffusionText\ expression the function 
$\fname$, which we call \emph{\diffusionText}, must be a pure function whose return type  and first argument type are the same. The symbol $\starK$
plays the role of a formal argument, hence the \diffusionText\
expression can be seen as the body of an anonymous, unary
function.
Viewed locally to a node, expression
$\e=\spreadThree{\e_0}{\fname}{\e_1,\myldots,\e_n}$ is evaluated at a given time
to value $\anyvalue$ as follows:
\begin{enumerate}
\item expressions $\e_0,\e_1,\myldots,\e_n$ are evaluated to
    values $\anyvalue_0,\anyvalue_1,\myldots,\anyvalue_n$;
\item the current values $\anyvaluebis_1,\myldots,\anyvaluebis_m$ of $\e$ in neighbours are
    gathered;
\item for each $\anyvaluebis_j$ in them, the \diffusionText\ function is applied as $\fname(\anyvaluebis_j,\anyvalue_1,\myldots,\anyvalue_n)$, giving value $\anyvaluebis'_j$;
\item the final result $\anyvalue$ is the minimum value among $\{\anyvalue_0,\anyvaluebis'_1,\myldots,\anyvaluebis'_m$\}: this value is made available to other nodes.
\end{enumerate}
Note that $\anyvalue\wfLeOf{\anytype}\anyvalue_0$, and if the
device is isolated then $\anyvalue=\anyvalue_0$.
 \FORGET{\begin{enumerate}
\item the current values of $\e$ in neighbours are
    gathered, and the minimum value is computed, call it
    $\anyvaluebis$;
\item expressions $\e_0,\e_1,\myldots,\e_n$ are evaluated to values $\anyvalue_0,\anyvalue_1,\myldots,\anyvalue_n$;
\item $\anyvalue$ is computed as $\textit{min}(\anyvalue_0,\fname(\anyvaluebis,\anyvalue_1,\myldots,\anyvalue_n))$, and is made available to other nodes.
\end{enumerate}
Note that $\anyvalue\wfLeOf{\anytype}\anyvalue_0$, and if the
device is isolated then $\anyvalue=\anyvalue_0$, as we assume
$\anyvaluebis$ is undefined.}
Viewed globally,
\mbox{$\spreadThree{\e_0}{\fname}{\e_1,\myldots,\e_n}$}
represents a field initially equal to $\e_0$; as time passes
some field values can decrease due to smaller values being
received from neighbours (after applying the diffusion
function).

The hop-count gradient created out of a \senv{src} sensor is hence simply defined as 
\[
\spreadtt{\senv{src}}{\startt{} + 1}
\]
assuming \senv{src} holds what we call a zero-field, namely, it is $0$ on source nodes and \topNumK{} everywhere else.
In this case \senv{src} is the source expression, and $\fname$ is unary successor function.

\subsection{Examples}\label{sec-examples}

As a reference scenario to ground the discussion, we can consider crowd steering in pervasive environments \cite{MVFFZ-MONET2012}: computational fields run on top of a myriad of small devices spread in the environment (including smartphones), and are used to guide people in complex environments (buildings, cities) towards points of interest (POIs) across appropriate paths.
There, a smartphone can perceive neighbour values of a gradient spread from a POI, and give directions towards smallest values so as to steer its owner and make him/her quickly descend the gradient \cite{tota}.
Starting from the hop-count gradient, various kinds of behaviour useful in crowd steering can be programmed, based on the definitions reported in Figure \ref{f:code}.
%

{\footnotesize \begin{figure}[t]
\[\begin{array}{|l|}
\hline
\deftt{\numNAME\ grad}{\numNAME\ i}{\spreadtt{i}{\startt{} + \senv{dist}}} \\[1mm]
\deftt{\numNAME\ restrict}{\numNAME\ i, bool \boolVar}{\boolVar\ ? i : \topNumK}\\[1mm]
\deftt{\numNAME\ restrictSum}{\numNAME\ x, \numNAME\ y, bool \boolVar}{restrict(x + y, \boolVar)}\\[1mm]
\deftt{\numNAME\ gradobs}{\numNAME\ i, bool \boolVar}{\spreadtt{i}{restrictSum(\startt{},\senv{dist}, \boolVar)}}\\[1mm]
\deftt{float gradbound}{\numNAME\ i, \numNAME\ z}{gradobs(i, grad(i) < z))}\\[1mm]
\hline
\end{array}\]
\caption{Definitions of examples}\label{f:code}\end{figure}}

The first function in Figure \ref{f:code} defines a more powerful gradient construct, called \texttt{grad}, which can be used to generalise over the hop-by-hop notion of distance: sensor \senv{dist} is assumed to exist that reifies an application-specific notion of distance as a positive number.
It can be $1$ everywhere to model hop-count gradient, or can vary from device to device to take into consideration contextual information.
For instance, it can be the output of a crowd sensor, leading to greater distances when/where crowded areas are perceived, so as to dynamically compute routes penalising crowded areas as in \cite{MVFFZ-MONET2012}.
In this case, note that diffusion function $\fname$ maps
$(\anyvalue_1,\anyvalue_2)$ to $\anyvalue_1+\anyvalue_2$.
Figure \ref{f:grad} (left) shows a pictorial representation, assuming devices are uniformly spread in a 2D environment: considering that an agent or data items move in the direction descending the values of a field, a gradient looks like a sort of uniform attractor towards the source, i.e., to the nearest source node.
It should be noted that when deployed in articulated environments, the gradient would stretch and dilate to accommodate the static/dynamic shape of environment, computing optimal routes.

\begin{figure}[t]
\begin{center}
\includegraphics[width=0.2\textwidth]{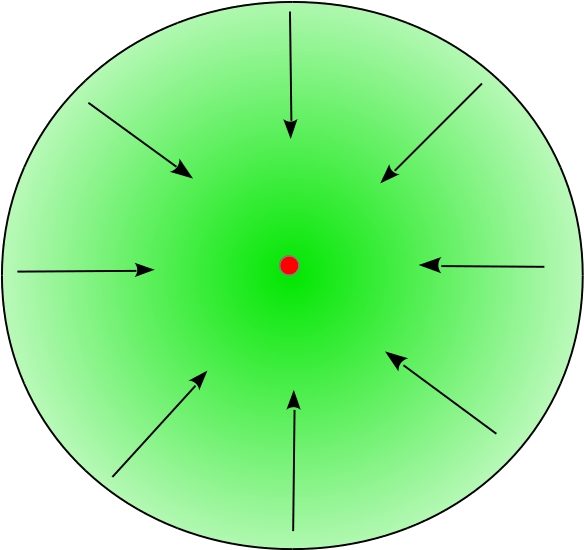}~~
\hspace{0.1\textwidth}
\includegraphics[width=0.2\textwidth]{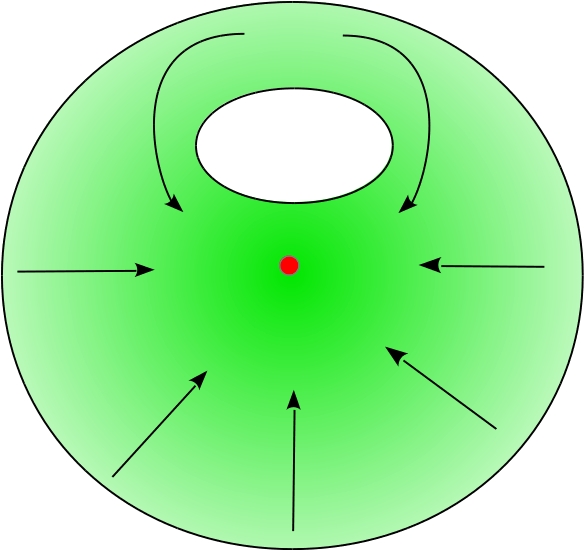}~~
\hspace{0.1\textwidth}
\includegraphics[width=0.22\textwidth]{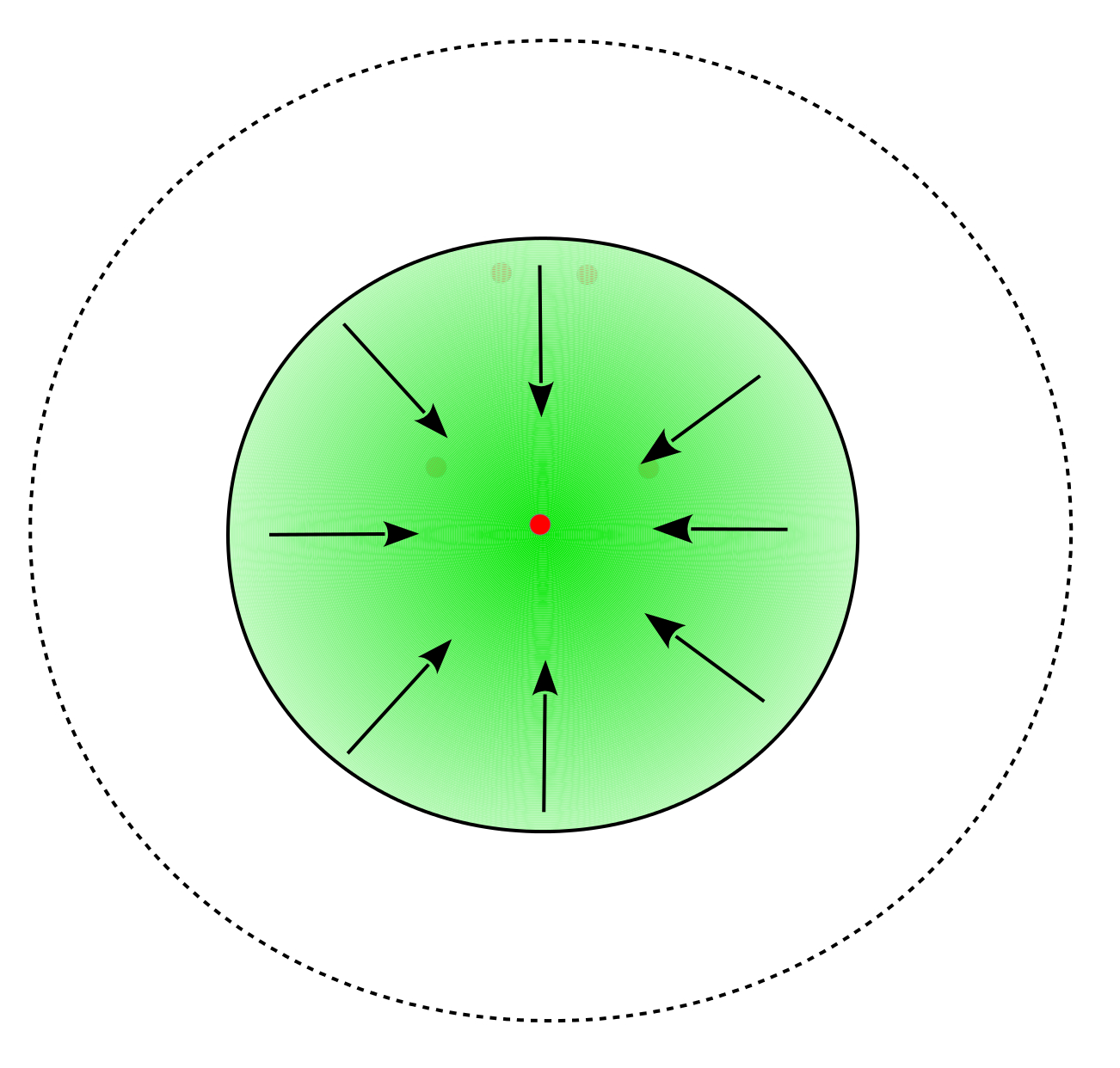}~~
\caption{Pictorial representation of  hop-count gradient field (left), a gradient circumventing ``crowd'' obstacles field (center), and a gradient with bounded distance (right)}\label{f:grad}
\end{center}
\end{figure}

By suitably changing the \diffusionText\ function, it is also
possible to block the diffusion process of gradients, as shown
in function \texttt{gradobs}: there, by
restriction we turn the gradient value to \topNumK{} in nodes
where the ``obstacle'' boolean field \boolVar holds \texttt{FALSE}.
This can be used to completely circumvent obstacle areas, as
shown in Figure~\ref{f:grad} (center).
Note that we here refer to a ``blocking'' behaviour, since sending a \topNumK{} value has no effect on the target because of the semantics of spreading; hence, an optimised implementation could simply avoid sending a \topNumK{} at all, so as not to flood the entire network.
This pattern is useful whenever steering people in environments with prohibited areas---e.g. road construction in a urban scenario.

Finally, by a different blocking mechanism we can limit the propagation distance of a gradient, as shown by function \texttt{gradbound}
and Figure \ref{f:grad} (right): the second argument \texttt{z} imposes a numerical bound to the distance, which is applied by exploiting the functions \texttt{gradobs} and \texttt{grad}.

In section~\ref{sec-calculus-extended} we will exploit the pair data structure to program more advanced examples of
behaviour useful in crowd steering.

\section{The Calculus of Self-Stabilising Computational Fields}\label{sec-calculus}

After informally introducing the proposed calculus in previous section, we now provide a formal account of it and precisely state the self-stabilisation property.
Namely, we formalise and illustrate by means of examples the type system (in Section~\ref{sec-typing}),  the operational semantics  (in Section~\ref{sec-operational-semantics}),  and the self-stabilisation property (in Section~\ref{sec:self-stabilisation}).

\subsection{Type checking}\label{sec-typing}

The syntax of the calculus is reported in Figure~\ref{fig:source:syntax}. As a standard syntactic notation in calculi for object-oriented and functional
languages~\cite{FJ}, we use the overbar notation to denote
metavariables over lists, e.g., we let $\overline \e$ range
over lists of expressions, written $\e_1\;\e_2\;\myldots\;\e_n$, and similarly for $\overline \xname$, $\overline \anytype$ and so on.
We write  $\signatureOf{\fname}$ to denote the type-signature $\anytype(\overline{\anytype})$ of $\fname$ (which specifies the type $\anytype$ of the result and the types $\overline{\anytype}=\anytype_1,\myldots,\anytype_n$ of the $n\ge 0$ arguments of $\fname$). We assume that the mapping $\signatureOf{\cdot}$ associates a type-signature to each built-in function and, for user-defined functions,  returns the type-signature specified in the function definition.
%
%

A program $\PROGRAM$ in our language is a mapping from function
names to function definitions, enjoying the following
\emph{sanity conditions}: \emph{(i)} $\PROGRAM(\dname)
= \defKK{\dname}{\cdots}{\cdots}{\cdots}$ for every
$\dname\in\domOf{\PROGRAM}$; \emph{(ii)} for every
function name $\dname$ appearing anywhere in $\PROGRAM$, we
have $\dname\in\domOf{\PROGRAM}$; \emph{(iii)} there
are no cycles in the function call graph (i.e., there are no
recursive functions in the program); and \emph{(iv)}
$\main\in\domOf{\PROGRAM}$ and it has zero arguments.
A program that does not contain the $\main$ function is called a \emph{library}.

The type system we provide aims to guarantee that no run-time error may arise during evaluation: its typing rules are given in Figure~\ref{fig:SurfaceTyping}.
 \begin{figure}[!t]{
 \framebox[1\textwidth]{
 $\begin{array}{l}
\textbf{Expression type checking:} \hfill
  \boxed{\surExpTypJud{\SurTypEnv}{\e}{\anytype}}
  \\
\begin{array}{c}
\nullsurfaceTyping{T-VAR}{
\surExpTypJud{\SurTypEnv,\xname:\anytype}{\xname}{\anytype}
}
\qquad\qquad
\nullsurfaceTyping{T-SNS}{
\surExpTypJud{\SurTypEnv}{\snsname}{\typeof{\snsname}}
}
\qquad\qquad
\nullsurfaceTyping{T-GVAL}{
\surExpTypJud{\SurTypEnv}{\groundvalue}{\typeof{\groundvalue}}
}
\skiptransition
\surfaceTyping{T-COND}{  \quad
\surExpTypJud{\SurTypEnv}{\e_0}{\boolType}
\quad
\surExpTypJud{\SurTypEnv}{\e_1}{\anytype}
\quad
\surExpTypJud{\SurTypEnv}{\e_2}{\anytype}
}{ \surExpTypJud{\SurTypEnv}{ \condExpr{\e_0}{\e_1}{\e_2}}{\anytype} }
\skiptransition
\surfaceTyping{T-FUN}{  \quad
\anytype(\overline{\anytype})=\signatureOf{\fname} \quad \surExpTypJud{\SurTypEnv}{\overline{\e}}{\overline{\anytype}}
}{ \surExpTypJud{\SurTypEnv}{\fname(\overline{\e})}{\anytype} }
\qquad
\surfaceTyping{T-SPR}{ \quad
\progressionOpOf{\fname}
\quad
\surExpTypJud{\SurTypEnv}{\fname(\e,\overline{\e})}{\anytype} }{
\surExpTypJud{\SurTypEnv}{\spreadThree{\e}{\fname}{\overline{\e}}}{\anytype} }
\skiptransition
%
%
%
\end{array}
\\
\textbf{User-defined function type checking:} \hfill
  \boxed{\surFunTypJud{}{\FUNCTION}{\anytype(\overline{\anytype})}}
  \\
\qquad\begin{array}{c}
\surfaceTyping{T-DEF}{ \quad
\surExpTypJud{\overline{\xname}:\overline{\anytype}}{\e}{\anytype}
}{ \surFunTypJud{}{\defK \; \anytype \;\dname
(\overline{\anytype\;\xname}) = \e}{\anytype(\overline{\anytype})} }
\end{array}
\end{array}$}
} 
\caption{Type-checking rules for expressions and function definitions} \label{fig:SurfaceTyping}
\end{figure}
\emph{Type environments}, ranged over by $\SurTypEnv$ and written $\overline{\xname}:\overline{\anytype}$, contain type assumptions for program variables.
The type-checking judgement for expressions is of the form $\surExpTypJud{\SurTypEnv}{\e}{\anytype}$, to be read: $\e$ has type $\anytype$ under the type assumptions
$\SurTypEnv$ for the program variables occurring in $\e$.
As a standard syntax in type systems~\cite{FJ}, given
$\overline{\xname}=\xname_1,\myldots,\xname_n$,
$\overline{\anytype}=\anytype_1,\myldots,\anytype_n$ and
$\overline{\e}=\e_1,\myldots,\e_n$ ($n\ge 0$), we write
$\overline{\xname}:\overline{\anytype}$ as short for
$\xname_1:\anytype_1,\myldots,\xname_n:\anytype_n$, and
$\surExpTypJud{\SurTypEnv}{\overline{\e}}{\overline{\anytype}}$
as short for $\surExpTypJud{\SurTypEnv}{\e_1}{\anytype_1}$
$\cdots$ $\surExpTypJud{\SurTypEnv}{\e_n}{\anytype_n}$.

Type checking of variables, sensors, ground values, 
conditionals, and function applications are almost standard. In particular, values and sensors and built-in functions are given a type by construction: the mapping $\typeof{\cdot}$ associates  a sort to each ground value and to each sensor, while rule \ruleNameSize{[T-FUN]} exploits the mapping $\signatureOf{\cdot}$.

\begin{exa}\label{exa:source-instance}
Figure~\ref{fig:source-instance} illustrates the ground types,  sensors, and built-in functions used in the examples introduced throughout the paper.
\end{exa}

 \begin{figure}[!t]{
 \framebox[1\textwidth]{
 $\begin{array}{@{\hspace{-6cm}}l}
\textbf{Ground types:}
  \\
\begin{array}{lclr}
\groundType & = & \boolType   \; \BNFmid \;   \numType 
\end{array}
\\
\\
\textbf{Sensor types:} 
  \\
\begin{array}{lcl}
\typeof{\texttt{\#src}} & = & \numType  
\\
\typeof{\texttt{\#dist}} & = & \numType  
\end{array}
\\
\\
\textbf{Built-in function type-signatures:} 
  \\
\begin{array}{lclr}
\signatureOf{\notK} & = & \boolType(\boolType)   & 
\\
\signatureOf{\orK} & = &  \boolType(\boolType,\boolType)   & 
\\
\signatureOf{-} & = & \numType(\numType)   & 
\\
\signatureOf{+} & = &  \numType(\numType,\numType)   & 
\\
\signatureOf{=} & = &  \boolType(\numType,\numType)   & 
\\
\signatureOf{<} & = &  \boolType(\numType,\numType)   & 
\end{array}
\end{array}$}
} 
\caption{Ground types, types for sensors, and type-signatures for built-in functions used in the examples} \label{fig:source-instance}
\end{figure}
The only ad-hoc type checking is provided for spreading expressions $\spreadThree{\e}{\fname}{\overline \e}$: they are given the type of $\fname(\e,\overline{\e})$, though the function $\fname$ must be a \emph{\diffusionText}, according to the following definition.

\begin{defi}[\DiffusionText]\label{def:Progression}
A type signature $\anytype(\overline{\anytype})$ with $\overline{\anytype}=\anytype_1,\myldots,\anytype_n$ $(n\ge 1)$ is a \emph{diffusion type signature} (notation $\progressionOpOf{\anytype(\overline{\anytype})}$) if $\anytype=\anytype_1$.
A pure function $\fname$ is a  \emph{\diffusionText} (notation $\progressionOpOf{\fname}$) if its type signature $\signatureOf{\fname} $  is a diffusion type signature.
\end{defi}

\begin{exa}\label{exa:diffusion}
Consider the functions defined in Figure~\ref{f:code}. The following two predicates hold:
\begin{itemize}
\item
 $\progressionOpOf{\texttt{+}}$, and
\item
 $\progressionOpOf{\texttt{restrictSum}}$
\end{itemize}
\end{exa}

Function type checking, represented by judgement
$\surFunTypJud{\SurTypEnv}{\FUNCTION}{\anytype(\overline{\anytype})}$, is standard.
In the following we always consider a \emph{well-typed}
program (or library) $\PROGRAM$, to mean that all the function declarations
in $\PROGRAM$ type check.
Note that no choice may be done when building a derivation for a given  type-checking judgment, so the type-checking rules straightforwardly describe a type-checking algorithm.

\begin{exa}
The library in Figure~\ref{f:code}  type checks by using the ground types, sensors, and type-signatures for built-in functions  in Figure~\ref{fig:source-instance}.
\end{exa}

\subsection{Operational Semantics}\label{sec-operational-semantics}

In this section we fomalise the operational semantics of the calculus. 
As for the field calculus~\cite{VDB-FOCLASA-CIC2013} and the Proto language~\cite{mitproto},
devices undergo computation
in rounds.
In each round, a device sleeps for some time, wakes up,
gathers information about messages received from neighbours while sleeping, evaluates the program, and finally broadcasts a message to
neighbours with information about the outcome of evaluation and 
goes back to sleep.
The scheduling of such rounds across the network is fair and
non-synchronous.
The structure of the network may change over time: when a device is sleeping its neighborhood may change and the device itself  may disappear (switch-off) and subsequently appear (switch-on). 
We 
first focus on single-device computations (in Section~\ref{sec-device-comp}) and then on whole
network evolution (in Section~\ref{sec-network-comp}).

\subsubsection{Device Computation}\label{sec-device-comp}

In the following, we let meta-variable $\deviceId$ 
range over the denumerable set $\deviceIdSet$ of \emph{device
identifiers}, meta-variable $\devset$ over finite sets of such devices, meta-variables $\anyvaluealt$, $\anyvalue$ and $\anyvaluebis$ over values.
Given a finite nonempty set $V\subseteq\semOf{\anytype}$ we denote by $\lowerBound V$ its minimum
element, and write  $\anyvalue\wedge\anyvalue'$ as short for $\lowerBound\{\anyvalue,\anyvalue'\}$.

In order to simplify the notation, we shall assume a fixed program $\PROGRAM$ and write $\e_{\main}$ to denote the body of the \texttt{main} function.
We say that ``device $\deviceId$ \emph{fires}'', to mean that expression $\e_{\main}$ is evaluated on device $\deviceId$.
The result of the evaluation is a \emph{value-tree}, which is an ordered tree of values, tracking the value of any evaluated subexpression.
Intuitively, such an evaluation is performed against the most recently received
value-trees of current neighbours and the current value of sensors, and
produces as result a new value-tree that is broadcasted to current  neighbours for their
firing.
Note that considering  simply a value (instead of a value-tree) as the outcome of the evaluation $\e_{\main}$ on a device $\deviceId$ would not be enough,
since the evaluation of each spreading expression $\e$ occurring in $\e_{\main}$ requires the values (at the root of their sub-value-trees) produced by the most recent evaluation of $\e$ on neighbours of $\deviceId$  (c.f.\ Sect.~\ref{sec-fields}).\footnote{Any implementation
might massively compress the value-tree, storing only enough
information for tracking the values of spreading expressions.}

The syntax of value-trees is given in Figure~\ref{fig:deviceSemantics}, together with the definition of the auxiliary functions $\conclusionOf{\cdot}$ and $\premiseNumOf{i}{\cdot}$ for extracting the root value and the $i$-th subtree of a value-tree, respectively---also the extension of these functions to sequences of value-environments $\overline{\vtree}$ is defined. We sometimes abuse the notation writing a value-tree with just the root as $\anyvalue$ instead of $\anyvalue()$.
The state of sensors $\senstate$ is a map from sensor names to values, modelling the inputs received from the external world.
This is written $\genmap{\overline\snsname}{\overline\anyvalue}$ as an abuse of notation to mean $\genmap{\snsname_1}{\anyvalue_1},\myldots,\genmap{\snsname_n}{\anyvalue_n}$.
We shall assume that it is complete (it has a mapping for any sensor used in the program), and correct (each sensor $\snsname$ has a type written $\typeof{\snsname}$, and is mapped to a value of that type).
For this map, and for the others to come, we shall use the following notations: $\senstate(\snsname)$ is used to extract the value that $\snsname$ is mapped to, $\mapupdate{\senstate}{\senstate'}$ is the map obtained by updating $\senstate$ with all the associations $\genmap{\snsname}{\anyvalue}$ of $\senstate'$ which do not escape the domain of $\senstate$ (namely, only those such that $\senstate$ is defined for $\snsname$).

The computation that takes place on a single device is formalised  by the big-step operational semantics rules given in
Figure~\ref{fig:deviceSemantics}. The derived judgements are of the form $\bsopsem{\senstate}{\overline\vtree}{\e}{\vtree}$,
to be read ``expression $\e$ evaluates to  value-tree $\vtree$
on sensor state $\senstate$ and w.r.t.\ the value-trees
$\overline\vtree$'', where:
\begin{itemize}
\item $\senstate$ is the current sensor-value map, modelling the inputs received from the external world;
\item $\overline\vtree$ is the list of the value-trees produced by the most recent evaluation of $\e$ on the
    current device's neighbours;
\item $\e$ is the closed expression to be evaluated;
\item the value-tree $\vtree$  represents the values
    computed for all the expressions encountered during the
    evaluation of $\e$--- in particular $\vrootOf{\vtree}$
    is the local value of field expression $\e$.
\end{itemize}

\begin{figure}[!t]{
 \framebox[1\textwidth]{
 $\begin{array}{l}
\textbf{Value-trees and sensor-value maps:}\\
\begin{array}{rcl@{\hspace{8.5cm}}r}
%
%
\vtree,\vtreealt & \BNFcce &  \anyvalue(\overline{\vtree})    &   {\mbox{\footnotesize value-tree}} \\[2pt]
\senstate & \BNFcce &  \genmap{\overline\snsname}{\overline\anyvalue}    &   {\mbox{\footnotesize sensor-value map}} \\
%
\end{array}\\
\hline\\[-8pt]
\textbf{Auxiliary functions:}\\
\begin{array}{l@{\hspace{1.5cm}}l}
\conclusionOf{\anyvalue(\overline{\vtree})}  =   \anyvalue
                         &    \premiseNumOf{i}{\anyvalue(\vtree_1,\myldots,\vtree_n)}  =   \vtree_i
\\
\conclusionOf{\vtree_1,\myldots,\vtree_n}  =   \conclusionOf{\vtree_1},\myldots,\conclusionOf{\vtree_n}
                         &
\premiseNumOf{i}{\vtree_1,\myldots,\vtree_n}  = \premiseNumOf{i}{\vtree_1},\myldots,\premiseNumOf{i}{\vtree_n}
\\
\end{array}\\
\hline\\[-8pt]
\textbf{Rules for expression evaluation:} \hfill
  \boxed{\bsopsem{\senstate}{\overline\vtree}{\e}{\vtree}}
  \\
\begin{array}{c}
\nullsurfaceTyping{E-SNS}{
\bsopsem{\senstate}{\overline\vtree}{\snsname}{\senstate(\snsname)}
}
\qquad\qquad
\nullsurfaceTyping{E-VAL}{
\bsopsem{\senstate}{\overline\vtree}{\anyvalue}{\anyvalue}
}
\skiptransition
\surfaceTyping{E-COND}{\\[-1.5mm]
  \bsopsem{\senstate}{\premiseNumOf{1}{\overline\vtree}}{\e_1}{\vtreealt_1}
    \quad
    \cdots
    \quad
    \bsopsem{\senstate}{\premiseNumOf{3}{\overline\vtree}}{\e_3}{\vtreealt_3}
    \qquad
    \anyvalue=\left\{\begin{array}{ll}
     \vrootOf{\vtreealt_2} & \mbox{if $\vrootOf{\vtreealt_1}=\trueK$}
     \\
     \vrootOf{\vtreealt_3} & \mbox{if $\vrootOf{\vtreealt_1}=\falseK$}
    \end{array}\right.
 }
{
\bsopsem{\senstate}{\overline\vtree}{ \condExpr{\e_1}{\e_2}{\e_3}}{\anyvalue(\vtreealt_1,\vtreealt_2,\vtreealt_3)}
}

\skiptransition
\surfaceTyping{E-BLT}{\\
  \bsopsem{\senstate}{\premiseNumOf{1}{\overline\vtree}}{\e_1}{\vtreealt_1}
  \quad
    \cdots
    \quad
    \bsopsem{\senstate}{\premiseNumOf{n}{\overline\vtree}}{\e_n}{\vtreealt_n}
    \qquad
    \anyvalue=\opApply{\opFunOf{\oname}}{\vrootOf{\vtreealt_1},\myldots,\vrootOf{\vtreealt_n}}
 }
{
\bsopsem{\senstate}{\overline\vtree}{\oname(\e_1,\myldots,\e_n)}{\anyvalue(\vtreealt_1,\myldots,\vtreealt_n)}
}
\skiptransition
\surfaceTyping{E-DEF}{
          \\
    \defK \; \anytype \;\dname(\anytype_1\;\xname_1,\myldots,\anytype_n\;\xname_n) = \e
    \qquad
  \bsopsem{\senstate}{\premiseNumOf{1}{\overline\vtree}}{\e_1}{\vtree'_1}
  \quad
    \cdots
    \quad
    \bsopsem{\senstate}{\premiseNumOf{n}{\overline\vtree}}{\e_n}{\vtree'_n}
 \\
     \bsopsem{\senstate}{\premiseNumOf{n+1}{\overline\vtree}}{\applySubstitution{\e}{\substitution{\xname_1}{\vrootOf{\vtree'_1}},\myldots,\substitution{\xname_n}{\vrootOf{\vtree'_n}}}}{\anyvalue(\overline\vtreealt)}
 }{
\bsopsem{\senstate}{\overline\vtree}{\dname(\e_1,\myldots,\e_n)}{\anyvalue(\vtree'_1,\myldots,\vtree'_n,\anyvalue(\overline\vtreealt))}
}
\skiptransition
\surfaceTyping{E-SPR}{
\\
\bsopsem{\senstate}{\premiseNumOf{0}{\overline\vtree}}{\e_0}{\vtreealt_0}
  \quad
    \cdots
    \quad
    \bsopsem{\senstate}{\premiseNumOf{n}{\overline\vtree}}{\e_n}{\vtreealt_n} \\
    \conclusionOf{\vtreealt_0,\myldots,\vtreealt_n} = \anyvalue_0 \myldots \anyvalue_n
    \qquad
    \conclusionOf{\overline\vtree} = \anyvaluebis_1 \myldots \anyvaluebis_m
    \\
    \bsopsem{\senstate}{\emptyset}{\fname(\anyvaluebis_1,\anyvalue_1,\myldots,\anyvalue_n)}{\anyvaluealt_1(\cdots)}
    \quad
    \cdots
    \quad
    \bsopsem{\senstate}{\emptyset}{\fname(\anyvaluebis_m,\anyvalue_1,\myldots,\anyvalue_n)}{\anyvaluealt_m(\cdots)}
    }
{
\bsopsem{\senstate}{\overline\vtree}{\spreadThree{\e_0}{\fname}{\e_1,\myldots,\e_n}}{\lowerBound\{\anyvalue_0,\anyvaluealt_1,\myldots,\anyvaluealt_m\}(\vtreealt_0,\vtreealt_1,\myldots,\vtreealt_n)}
}
\end{array}
\end{array}$}
} 
 \caption{Big-step operational semantics for expression evaluation} \label{fig:deviceSemantics}
\end{figure}

The rules of the operational semantics are \emph{syntax
directed}, namely, the rule used for deriving a judgement
$\bsopsem{\senstate}{\overline\vtree}{\e}{\vtree}$ is univocally
determined by $\e$ (cf.\ Figure~\ref{fig:deviceSemantics}).
Therefore, the shape of the value-tree $\vtree$ is univocally
determined by $\e$, and the whole value-tree is univocally
determined by $\senstate$, $\overline\vtree$, and $\e$.

The rules of the operational semantics are almost standard,
with the exception that rules \mbox{\ruleNameSize{[E-COND]}}, \ruleNameSize{[E-BLT]},
\ruleNameSize{[E-DEF]} and \ruleNameSize{[E-SPR]} use the
 auxiliary function $\premiseNumOf{i}{\cdot}$ to ensure that, in
the judgements in the premise of the rule, the value-tree
environment is aligned with the expression to be evaluated. Note that the semantics of conditional expressions prescribes that both the branches of the conditional are evaluated.\footnote{Our calculus does not model the \emph{domain restriction} construct in~\cite{VDB-FOCLASA-CIC2013,mitproto}.}

The most important rule is \ruleNameSize{[E-SPR]} which handles spreading expressions formalising the description provided in Section~\ref{sec-basicIngredients}.
It first recursively evaluates expressions $\e_i$ to value-trees $\vtreealt_i$ (after proper alignment of value-tree environment by operator $\premiseNumOf{i}{.}$) with top-level values $\anyvalue_i$.
Then it gets from neighbours their values $\anyvaluebis_j$ for the spreading expression, and for each of them $\fname$ is evaluated giving top-level result $\anyvaluealt_j$.
The resulting value is then obtained by the minimum among $\anyvalue_0$ and the values $\anyvaluealt_j$ (which equates to $\anyvalue_0$ if there are currently no neighbours). Note that, since  in
a spreading expression
$\spreadThree{\e_0}{\fname}{\e_1,\myldots,\e_n}$ the function $\fname$ must be a diffusion and diffusions are pure
functions  (cf. Section~\ref{sec-typing}), only the root of the
value-tree produced  by the evaluation of the application of $\fname$ to the values of $\e_1$, $\ldots$, $\e_n$ must be stored (c.f.\ the conclusion of rule
\ruleNameSize{[E-SPR]}). We will in the following provide a network semantics taking care of associating to each device the set of neighbour trees against which it performs a computation round, namely, connecting this operational semantics to the actual network topology.

\begin{exa}[About device semantics]\label{exa:DeviceSemantics:hop-count gradient}
Consider the program $\PROGRAM$:
\[
\texttt{def \numNAME\ main() is
\spreadtt{\senv{src}}{\startt{} + \senv{dist}}}, 
\]
where
\senv{src} and \senv{dist} are sensors of type \numType, and \texttt{+} is the built-in
sum operator which has type-signature \texttt{\numNAME(\numNAME,\numNAME)}.

The evaluation of  $\e_{\main}=\spreadtt{\senv{src}}{\startt{} + \senv{dist}}$ on a device $\deviceId_1$ when
\begin{itemize}
\item
the current sensor-value map  for $\deviceId_1$ is $\senstate_1$ such that  $\senstate_1(\senv{src})=0$ and  $\senstate_1(\senv{dist})=1$, and
\item
$\deviceId_1$ has currently no neighbours,
\end{itemize}
(expressed by the judgement $\bsopsem{\senstate_1}{\emptyset}{\e_{\main}}{\vtree_1}$) yields  the value-tree $\vtree_1=0(0,1)$ by rule \mbox{\ruleNameSize{[E-SPR]}}, since: $n=1$; the evaluation of 
$\e_0=\senv{src}$ yields $\vtreealt_0=0()$ (by rule \mbox{\ruleNameSize{[E-SNS]}});  the evaluation of 
$\e_1=\senv{dist}$ yields $\vtreealt_1=1()$ (by rule \mbox{\ruleNameSize{[E-SNS]}});  $m=0$; and $\lowerBound\{0\}=0$.

Similarly, the evaluation of  $\e_{\main}$ on a device $\deviceId_2$ when
\begin{itemize}
\item
the current sensor-value map  for $\deviceId_2$ is $\senstate_2$ such that  $\senstate_2(\senv{src})=8$ and  $\senstate_2(\senv{dist})=1$, and
\item
$\deviceId_2$ has currently no neighbours,
\end{itemize}
(expressed by the judgement $\bsopsem{\senstate_2}{\emptyset}{\e_{\main}}{\vtree_2}$) yields  the value-tree $\vtree_2=8(8,1)$ by rule \mbox{\ruleNameSize{[E-SPR]}}, since: $n=1$; the evaluation of 
$\e_0=\senv{src}$ yields $\vtreealt_0=8()$ (by rule \mbox{\ruleNameSize{[E-SNS]}});  the evaluation of 
$\e_1=\senv{dist}$ yields $\vtreealt_1=1()$ (by rule \mbox{\ruleNameSize{[E-SNS]}});  $m=0$; and $\lowerBound\{8\}=8$.

Then, the evaluation of  $\e_{\main}$ on a device $\deviceId_3$ when
\begin{itemize}
\item
the current sensor-value map  for $\deviceId_3$ is $\senstate_3$ such that  $\senstate_3(\senv{src})=4$ and  $\senstate_3(\senv{dist})=1$, and
\item
$\deviceId_3$ has  neighbours $\deviceId_1$ and $\deviceId_2$,
\end{itemize}
(expressed by the judgement $\bsopsem{\senstate_3}{\vtree_1\vtree_2}{\e_{\main}}{\vtree_3}$) yields  the value-tree $\vtree_3=1(4,1)$ by rule \mbox{\ruleNameSize{[E-SPR]}}, since: $n=1$; the evaluation of 
$\e_0=\senv{src}$ yields $\vtreealt_0=4()$ (by rule \mbox{\ruleNameSize{[E-SNS]}});  the evaluation of 
$\e_1=\senv{dist}$ yields $\vtreealt_1=1()$ (by rule \mbox{\ruleNameSize{[E-SNS]}});  $m=2$; the evaluation of 
$0+1$ yields $1(0(),1())$ (by rule \mbox{\ruleNameSize{[E-BLT]}});  the evaluation of 
$8+1$ yields $9(0(),8())$ (by rule \mbox{\ruleNameSize{[E-BLT]}});
and $\lowerBound\{4,1,9\}=1$.
\end{exa}

\subsubsection{Network Evolution}\label{sec-network-comp}

\begin{figure}[!t]{
 \framebox[1\textwidth]{
 $\begin{array}{l}
 \textbf{System configurations and action labels:}\\
\begin{array}{lcl@{\hspace{9cm}}r}
\Field & \BNFcce &  \genmap{\overline\deviceId}{\overline\vtree}    &   {\mbox{\footnotesize computational field}} \\
\Topo & \BNFcce &  \genmap{\overline\deviceId}{\overline\devset}    &   {\mbox{\footnotesize topology}} \\
\Sens & \BNFcce &  \genmap{\overline\deviceId}{\overline\senstate}    &   {\mbox{\footnotesize sensors-map}} \\
\Envi & \BNFcce &  \EnviS{\Topo}{\Sens}    &   {\mbox{\footnotesize environment}} \\
\Cfg & \BNFcce &  \SystS{\Envi}{\Field}    &   {\mbox{\footnotesize network configuration}} \\
\act & \BNFcce &  \deviceId \;\BNFmid\; \envact    &   {\mbox{\footnotesize action label}} \\
\end{array}\\
\hline\\[-8pt]
\textbf{Environment well-formedness:}\\
\begin{array}{l}
\wfn{\EnviS{\Topo}{\Sens}} \textrm{~~holds if $\Topo,\Sens$ have same domain, and $\Topo$'s values do not escape it.}
\\
\end{array}\\
\hline\\[-8pt]
\textbf{Transition rules for network evolution:} \hfill
  \boxed{\nettran{\Cfg}{\act}{\Cfg}}
  \\[0.2cm]
\vspace{0.5cm}
\begin{array}{c}
\netopsemRule{N-FIR}
                 {\qquad \Envi=\EnviS{\Topo}{\Sens} \qquad \Topo(\deviceId)=\overline\deviceId \qquad \bsopsem{\Sens(\deviceId)}{\Field(\overline\deviceId)}{\emain}{\vtree}}
                 {\nettran{\SystS{\Envi}{\Field}}{\deviceId}{\SystS{\Envi}{\mapupdate{\Field}{\genmap{\deviceId}{\vtree}}}}
                 }
\skiptransition
\netopsemRule{N-ENV}
                 {\\ \wfn{\Envi'}\qquad \Envi'=\EnviS{\Topo}{\genmap{\deviceId_1}{\senstate_1},\myldots,\genmap{\deviceId_n}{\senstate_n}}\\
                  \bsopsem{\senstate_1}{\emptyset}{\emain}{\vtree_1}\quad\cdots\quad\bsopsem{\senstate_n}{\emptyset}{\emain}{\vtree_n} \qquad
                  \Field_0=\genmap{\deviceId_1}{\vtree_1},\myldots,\genmap{\deviceId_n}{\vtree_n}
                 }
                 {\nettran{\SystS{\Envi}{\Field}}{\envact}{\SystS{\Envi'}{\mapupdate{\Field_0}{\Field}}}
                 }\\[-10pt]
\end{array}\\
%
%
%
\end{array}$}
} 
 \caption{Small-step operational semantics for network evolution} \label{fig:networkSemantics}
\end{figure}

We now provide an operational semantics for the evolution of whole networks, namely,
for modelling the distributed evolution of computational fields over time.
Figure \ref{fig:networkSemantics} (top) defines key syntactic elements to this end.
$\Field$ models the overall computational field (state), as a map from device identifiers to value-trees.
$\Topo$ models \emph{network topology}, namely, a directed neighbouring graph, as a map from device identifiers to set of identifiers.
$\Sens$ models \emph{sensor (distributed) state}, as a map from device identifiers to (local) sensors (i.e., sensor name/value maps).
Then, $\Envi$ (a couple of topology and sensor state) models the system's environment.
So, a whole network configuration $\Cfg$ is a couple of a field and environment.

We define network operational semantics in terms of small-steps transitions of the kind $\nettran{\Cfg}{\act}{\Cfg'}$, where $\act$ is either a device identifier in case it represents its firing, or label $\envact$ to model any environment change.
This is formalised by the two rules in Figure \ref{fig:networkSemantics} (bottom).
Rule \ruleNameSize{[N-FIR]} models a network evolution due to a computation round (firing) at device $\deviceId$: it reconstructs the proper local environment, taking local sensors ($\Sens(\deviceId)$) and accessing the value-trees of $\deviceId$'s neighbours;\footnote{The operational semantics abstracts from the details of message broadcast from/to neighbours: the most recent value-trees received by a device $\deviceId$  from its neighbours while it was sleeping are identified with the value-trees associated to the neighbours of the device $\deviceId$ when it fires.} then by the single device semantics we obtain the device's value-tree $\vtree$, which is used to update system configuration.
Rule \ruleNameSize{[N-ENV]} models a network evolution due to change of the environment $\Envi$ to an arbitrarily new well-formed environment $\Envi'$---note that this encompasses both neighborhood change and addition/removal of devices. Let $\deviceId_1,\myldots,\deviceId_n$ be the domain of $\Envi'$. We first construct a field $\Field_0$ associating to all the devices of $\Envi'$ the default value-trees $\vtree_1,\myldots,\vtree_n$ obtained by making devices perform an evaluation with no neighbours and sensors as of $\Envi'$. Then, we adapt the existing field $\Field$ to the new set of devices: $\mapupdate{\Field_0}{\Field}$ automatically handles removal of devices, map of new devices to their default value-tree, and retention of existing value-trees in the other devices.

\begin{exa}[About network evolution]\label{exa:NetworkEvolution:hop-count gradient}
Consider a network of devices running the program $\PROGRAM$ of Example~\ref{exa:DeviceSemantics:hop-count gradient}.
The initial situation, when (the functionality associated to) the program $\PROGRAM$ is switched-off on all the devices, is modelled by the
empty network configuration $\SystS{\EnviS{\genmap{\emptyset}{\emptyset}}{\genmap{\emptyset}{\emptyset}}}{\genmap{\emptyset}{\emptyset}}$.

The network evolution representing the fact that the environment evolves because device $\deviceId_3$ switches-on  when 
its sensor  $\senv{src}$ perceives value $0$ and  its sensor $\senv{dist}$ perceives value $1$ (and gets initialised to the  value-tree obtained by firing with respect to the empty set of neighbourhs),  is modelled (according to rule \ruleNameSize{[N-ENV]})  by the reduction step 
$\nettran{\SystS{\EnviS{\genmap{\emptyset}{\emptyset}}{\genmap{\emptyset}{\emptyset}}}{\genmap{\emptyset}{\emptyset}}}{\envact}{\SystS{\EnviS{\genmap{\deviceId_3}{\emptyset}}{\genmap{\deviceId_3}{\senstate_1}}}{\genmap{\deviceId_3}{\vtree_1}}}$, where the sensor-value mapping $\senstate_1$ and the value-tree $\vtree_1$ are those introduced in Example~\ref{exa:DeviceSemantics:hop-count gradient}.

Then, the network evolution representing the fact that the environment evolves is as follows:
\begin{itemize}
\item
the device $\deviceId_1$ switches-on  when 
its sensor  $\senv{src}$ perceives value $0$ and  its sensor $\senv{dist}$ perceives value $1$,
and  has only $\deviceId_3$ as neighbour;
 \item
the device $\deviceId_2$ switches-on  when 
its sensor  $\senv{src}$ perceives value $4$ and  its sensor $\senv{dist}$ perceives value $1$, and has no neighbours, and
 \item
on device $\deviceId_3$ sensor  $\senv{src}$ perceives value $8$, sensor $\senv{dist}$ perceives value $1$, and
$\deviceId_3$ has $\deviceId_1$ and  $\deviceId_2$ as neighbours,
\end{itemize}
is modelled (according to rule \ruleNameSize{[N-ENV]})  by the reduction step
$\nettran{\SystS{\EnviS{\genmap{\deviceId_3}{\emptyset}}{\genmap{\deviceId_3}{\senstate_1}}}{\genmap{\deviceId_3}{\vtree_1}}}{\envact}{\SystS{\EnviS{\Topo}{\Sens}}{\Field}}$, where
\[
\Topo=\genmap{\deviceId_1}{\{\deviceId_3\}},\genmap{\deviceId_2}{\emptyset},\genmap{\deviceId_3}{\{\deviceId_1,\deviceId_2\}},
\quad
\Sens=\genmap{\deviceId_1}{\senstate_1},\genmap{\deviceId_2}{\senstate_2},\genmap{\deviceId_3}{\senstate_3}, 
\quad
\Field=\genmap{\deviceId_1}{\vtree_1},\genmap{\deviceId_2}{\vtree_2},\genmap{\deviceId_3}{\vtree_1}
\]
and the sensor-value mapping $\senstate_2$, $\senstate_3$ and the value-tree $\vtree_2$ are those introduced in Example~\ref{exa:DeviceSemantics:hop-count gradient}.

Finally, the network evolution representing the fact that device $\deviceId_3$ fires 
is modelled (according to rule \ruleNameSize{[N-FIR]})  by the reduction step
$\nettran{\SystS{\EnviS{\Topo}{\Sens}}{\Field}}{\deviceId_3}{\SystS{\EnviS{\Topo}{\Sens}}{\genmap{\deviceId_1}{\vtree_1},\genmap{\deviceId_2}{\vtree_2},\genmap{\deviceId_3}{\vtree_3}}}$, where the value-tree $\vtree_2$ is that introduced in Example~\ref{exa:DeviceSemantics:hop-count gradient}.
\end{exa}

\subsection{The  Self-stabilisation Property}\label{sec:self-stabilisation}

Upon this semantics, we introduce the following definitions and notations ending with the self-stabilisation property.
\begin{description}
 \item[Initiality] 
The empty network configuration $\SystS{\EnviS{\genmap{\emptyset}{\emptyset}}{\genmap{\emptyset}{\emptyset}}}{\genmap{\emptyset}{\emptyset}}$ is said \emph{initial}.
 \item[Reachability] 
Write $\Cfg\netArrActStar\Cfg'$ as short for
$\Cfg\netArrWithAct{\act_1}\Cfg_1\netArrWithAct{\act_2}\cdots\netArrWithAct{\act_n}\Cfg'$. A configuration $\Cfg$ is said \emph{reachable} if $\Cfg_0\netArrActStar\Cfg$ where $\Cfg_0$ is initial. Reachable configurations are the well-formed ones, and in the following we shall implicitly consider only reachable configurations.
 \item[Firing] 
A firing evolution from $\Cfg$ to $\Cfg'$, written $\Cfg\netArrStar\Cfg'$, is one such that $\Cfg\netArrIdStar\Cfg'$ for some $\overline\deviceId$, namely, where only firings occur.
 \item[Stability] 
A network configuration $\Cfg$ is said \emph{stable} if $\nettran{\Cfg}{\deviceId}{\Cfg'}$ implies $\Cfg=\Cfg'$, namely, the computation of fields reached a fixpoint in the current environment. Note that if $\Cfg$ is stable, then it also holds that $\Cfg\netArrStar\Cfg'$ implies $\Cfg=\Cfg'$.
 \item[Fairness] 
We say that a sequence of device fires is \emph{$k$-fair} ($k\ge 0$) to mean that, for every $h$ ($1\le h\le k$), the $h$-th fire of any device is followed by at least $k-h$ fires of all the other devices. Accordingly, a firing evolution $\Cfg\netArrIdStar\Cfg'$ is said $k$-\fair, written $\Cfg\netArrIdFairStarN{k}\Cfg'$, to mean that $\overline{\deviceId}$ is $k$-fair. We also write
 $\Cfg\netArrFairStarN{k}\Cfg'$ if $\Cfg\netArrIdFairStarN{k}\Cfg'$ for some $\overline\deviceId$. This notion of fairness will be used to characterise finite firing evolutions in which all devices are given equal chance to fire when all others had.
 \item[Strong self-stabilisation] 
A network configuration $\SystS{\Envi}{\Field}$ is said to \emph{(strongly) self-stabilise} (simply, self-stabilise, in the following) to $\SystS{\Envi}{\Field'}$ if there is a $k>0$ and a field $\Field'$ such that $\SystS{\Envi}{\Field}\netArrFairStarN{k}\SystS{\Envi}{\Field'}$ implies $\SystS{\Envi}{\Field'}$ is stable, and $\Field'$ is univocally determined by $\Envi$. Self-stability basically amounts to the inevitable reachability of a stable state depending only on environment conditions, through a sufficiently long fair evolution. Hence, the terminology is abused equivalently saying that a program $\PROGRAM$ or (equivalently) a field expression $\emain$ is self-stabilising if for any environment state $\Envi$ there exists a unique stable field $\Field'$ such that \emph{any} $\SystS{\Envi}{\Field}$ \emph{self-stabilises} to $\SystS{\Envi}{\Field'}$.
\item[Self-stability]
A network configuration $\SystS{\Envi}{\Field}$ is said \emph{self-stable} to mean that it  is stable and $\Field$ is univocally determined by $\Envi$.

\end{description}

\noindent Note that our definition of self-stabilisation is actually a stronger version of the standard definition of self-stabilisation as given e.g. in \cite{dolev}---see more details in Section \ref{sec-related}.
Instead of simply requiring that we enter a ``self-stable set'' of states and never escape from it, we require that \emph{(i)} such a set has a single element, and \emph{(ii)} such an element is globally unique, i.e., it does not depend on the initial state.
Viewed in the context of an open system, it means that we seek for programs self-stabilising in any environment independently of any intermediate computation state.
This is a requirement of key importance, since it entails that any subexpression of the program can be associated to a final and stable field, reached in finite time by fair evolutions and adapting to the shape of the environment.
This acts as the sought bridge between the micro-level (field expression in program code), and the macro-level (expected global outcome).

%
%

\begin{exa}[A self-stabilising program]\label{exa:hop-count gradient}
Consider a network of devices running the program $\PROGRAM$ of Examples~\ref{exa:DeviceSemantics:hop-count gradient} and~\ref{exa:NetworkEvolution:hop-count gradient}
\[
\texttt{def \numNAME\ main() is
\spreadtt{\senv{src}}{\startt{} + \senv{dist}}}, 
\]
and assume that the sensor
\senv{dist} is guaranteed to always return a positive value
(recall that the sensors \senv{src} and \senv{dist} have type \numType\ and the built-in operator  \texttt{+}  has type-signature \texttt{\numNAME(\numNAME,\numNAME)}). Then, the operator \texttt{+} is guaranteed to be used with type $\numType(\numType,\pnumType)$, where 
$\pnumType$ is the type of positive reals (i.e., the \emph{refinement type}~\cite{FreemanPhenning:PLDI-1991} that refines  type \texttt{real} by keeping only the positive values), so that the  following conditions are satisfied:
\begin{enumerate}
\item \texttt{+} is \emph{monotonic nondecreasing} in its first argument, i.e.,  for any $\anyvalue$, $\anyvalue'$ and $\anyvalue_2$ of type $\pnumType$:
\\
 $\anyvalue\wfLeOf{\numType}\anyvalue'$
    implies
    $\anyvalue+\anyvalue_2\wfLeOf{\numType}\anyvalue'+\anyvalue_2$;
\item 
\texttt{+}  is \emph{progressive} in its
    first argument, i.e., for any  $\anyvalue$ and $\anyvalue_2$ of type $\pnumType$: 
\\
$\topNumK{}=\topNumK{}+\anyvalue_2$, and
 $\anyvalue\not=\topNumK{}$ implies
     $\anyvalue\wfLtOf{\numType}\anyvalue+\anyvalue_2$.
\end{enumerate}

%
%
Starting from an initial empty configuration, we move by rule \ruleNameSize{[N-ENV]} to a new environment with the following features:
\begin{itemize}
 \item the domain is formed by $2n$ ($n\ge 1$) devices $\deviceId_1,\myldots,\deviceId_n,\deviceId_{n+1},\myldots,\deviceId_{2n}$;
 \item the topology is such that any device $\deviceId_i$ is connected to $\deviceId_{i+1}$ and $\deviceId_{i-1}$ (if they exist);
 \item sensor \senv{dist} gives $1$ everywhere;
 \item sensor \senv{src} gives $0$ on the devices $\deviceId_i$ ($1\le i\le n$, briefly referred to as \emph{left devices}) and a value $u$
($u> n+1$) on the devices $\deviceId_j$ ($n+1\le j\le 2n$, briefly referred to as \emph{right devices}).
\end{itemize}
Accordingly, the left devices are all assigned to value-tree $0(0,1)$, while the right ones to $u(u,1)$: hence, the resulting field maps left devices to $0$ and right devices to $1$---remember such evaluations are done assuming nodes are isolated, hence the result is exactly the value of the source expression.
With this environment, the firing of a device can only replace the root of a value-tree, making it the minimum of the source expression's value and the successor of neighbour's values.
Hence, any firing of a device that is not $\deviceId_{n+1}$ does not change its value-tree.
When $\deviceId_{n+1}$ fires instead by rule \ruleNameSize{[N-FIR]}, its value-tree becomes $1(u,1)$, and it remains so if more firings occur next.

Now, only a firing at $\deviceId_{n+2}$ causes a change: its value-tree becomes $2(u,1)$.
Going on this way, it is easy to see that after any $n$-fair firing sequence the network self-stabilises to the field state
where left devices still have value-tree $0(u,1)$, while right devices $\deviceId_{n+1},\deviceId_{n+2},\deviceId_{n+3},\myldots$ have value-trees $1(u,1),2(u,1),3(u,1),\myldots$, respectively. That is, the root of such trees form a hop-count gradient, measuring minimum distance to the source nodes, namely, the left devices.

It can also be shown that any environment change, followed by a sufficiently long firing sequence, makes the system self-stabilise again, possibly to a different field state.
For instance, if the two connections of $\deviceId_{2n-1}$ to/from $\deviceId_{2n-2}$ break (assuming $n>2$), the part of the network excluding $\deviceId_{2n-1}$ and $\deviceId_{2n}$ keeps stable in the same state.
The values at $\deviceId_{2n-1}$ and $\deviceId_{2n}$ start raising instead, increasing of $2$ alternatively until both reach the initial value-trees $u(u,1)$---and this happens in finite time by a fair evolution thanks to the local noetherianity property of stabilising \diffusionsText.
Note that the final state is still the hop-count gradient, though adapted to the new environment topology.
\end{exa}

\begin{exa}[A non self-stabilising program]\label{exa:not-self-stabilising}
An example of program that is \emph{not} self-stabilising is
\[
\texttt{def real main() is \spreadtt{\senv{src}}{id(\startt{})}}
\]
 (written \texttt{def real main() is \spreadtt{\senv{src}}{\startt{}}} for short). There, the \diffusionText\
 is the identity function: \texttt{\numNAME\ id(\numNAME\ x) is x}, which (under the assumption, done in  Example~\ref{exa:hop-count gradient}, that the sensor \texttt{\#scr} is guaranteed to return positive values) is guaranteed to be used with signature \texttt{(pr)pr} and is not progressive in its first argument (c.f.\ condition (1) in Example~\ref{exa:hop-count gradient}).\footnote{The function \texttt{id} is progressive whenever its is used with a signature of the form \texttt{(real\_n)real\_n}, where \texttt{real\_n} is the refinement type that refines  \texttt{real\_n} by keeping only the value \texttt{n}---in the example, this corresponds to the case when the sensor \texttt{\#scr} is guaranteed to always return the constant value \texttt{n} on all the devices.}
 
Assuming a connected network, and \senv{src} holding value $\anyvalue_s$
in one node and $\topNumK{}$ in all others, then
\emph{any} configuration where all nodes hold the same value
$\anyvalue$ less than or equal to $\anyvalue_s$ is trivially stable.
This would model a source gossiping a
fixed value $\anyvalue_{s}$ everywhere: if the source suddenly gossips a
value $\anyvalue_{s}'$ smaller than $\anyvalue$, then the network would
self-organise and all nodes would eventually hold $\anyvalue_{s}'$.
However, if the source then gossips a value $\anyvalue_{s}''$ greater
than $\anyvalue_{s}'$, the network
would \emph{not} self-organise and all nodes would remain stuck to value
$\anyvalue_{s}'$.
%
%
\end{exa}

\section{Sorts, Stabilising Diffusions and the Stabilising-Diffusion Condition}\label{sec-stabilisingProgression}


In this section we state a sufficient  condition for self-stabilisation. This condition is about the behaviour of a diffusion (cf.\ Definition~\ref{def:Progression}) on a subsets of its arguments (cf.\ Example~\ref{exa:hop-count gradient}).  We first introduce \emph{refinement types} (or \emph{sorts}) as a convenient way to denote these subsets  (in Section~\ref{sec-ref-types}) and then use them to formulate the notion of \emph{stabilising diffusion} (in Section~\ref{sec-Progression}) and the sufficient condition for self-stabilisation (in Section~\ref {sec-program-with-valid-assumptions}).

\subsection{Refinement Types (or Sorts)}\label{sec-ref-types}

\emph{Refinement types}~\cite{FreemanPhenning:PLDI-1991} provide a mean to conservatively extend the static type system of a language by providing the ability of specify \emph{refinements} of each type. All the programs accepted by the original type system are accepted by the refinement-type system and vice-versa, but refinement-types provide additional information that support stating and checking properties of programs. 
Following~\cite{RowanDavies:Phd-2005}, we refer to refinement types as \emph{sorts} and use terms like \emph{subsorting} and \emph{sort checking}.

For instance, for the ground type $\numType$ of \numNAMEs\ we consider the six \emph{ground 
sorts} $\nnumType$ (negative \numNAMEs), $\znumType$ (the sort for \texttt{0}),  $\pnumType$ (positive \numNAMEs), $\znnumType$ (zero or negative \numNAMEs), $\zpnumType$ (zero or positive \numNAMEs),  and $\numType$ (each type trivially refines itself); while for the type $\boolType$ we consider the thee sorts $\falseType$ (the sort for $\falseK$), $\trueType$ (the sort for  $\trueK$) and $\boolType$.
Each sort-signature has the same structure of the type-signature it refines.
For instance, we can build $9 (=3^2)$ \emph{sort-signatures} for the type-signature $\boolType(\boolType)$:
\[
\begin{array}{lll}
\falseType(\falseType), & \falseType(\trueType), & \falseType(\boolType), \\
\trueType(\falseType), & \trueType(\trueType), & \trueType(\boolType), \\
\boolType(\falseType), & \boolType(\trueType), & \boolType(\boolType).
\end{array}
\]
We assume
 a mapping  $\refof{\cdot}$ that associates to each type the (set of) sorts that refine it, and a mapping $\refsigof{\cdot}$ that associates to each type-signature the (set of) sort-signatures that refine it (note that the latter mapping is determined by the the former, i.e., by the value of $\refof{\cdot}$ on ground types). A type $\anytype$ trivially refines itself, i.e.,  for every type $\anytype$ it holds that  $\anytype\in\refof{\anytype}$. Similarly,  for every type-signature $\anytype(\overline{\anytype})$ it holds that  $\anytype(\overline{\anytype})\in\refsigof{\anytype(\overline{\anytype})}$.
We write $\semOf{\rtype}$ to denote the set of values of sort $\rtype$. Note that, by construction:
\begin{center}
for all $\rtype\in\refof{\anytype}$ it holds that $\semOf{\rtype}\subseteq\semOf{\anytype}$.
\end{center}
Sorts and sort-signatures express properties of expressions and  functions, respectively. 
We say that:
\begin{itemize}
\item
a value $\anyvalue$  \emph{has} (or \emph{satisfies})  sort $\rtype$ to mean that $\anyvalue\in\semOf{\rtype}$ holds---we write $\refof{\anyvalue}$ to denote the set of all the sorts satisfied by $\anyvalue$, and
\item
a pure function $\fname$  \emph{has} (or \emph{satisfies})  sort-signature $\rtype(\overline{\rtype})$ to mean that for all $\overline{\anyvalue}\in\semOf{\overline{\rtype}}$ it holds that $\semOf{\fname}(\overline{\anyvalue})\in\semOf{\rtype}$---we write $\refsigof{\fname}$ to denote the set of all the sort-signatures satisfied by $\fname$.
\end{itemize}
For every  sort $\rtype$ in $\refof{\anytype}$ we write $\wfLeOf{\rtype}$ to denote the restriction to $\semOf{\rtype}$ of  the total order $\wfLeOf{\anytype}$ (cf.\ Section~\ref{sec-basicIngredients}) and write $\topValOf{\rtype}$ to denote the maximum element of $\semOf{\rtype}$ with respect to $\wfLeOf{\rtype}$.

\emph{Subsorting} is a partial order relation over sorts of each type that models the inclusion relationship between them. For instance, each positive number is a zero or positive number, so will write $\rsubof{\nnumType}{\znnumType}$ to indicate that $\nnumType$ is a subsort of $\znnumType$.  
We require that the subsorting relation satisfies the following (fairly standard) conditions:
\begin{enumerate}
\item
The type $\anytype$ is the maximum element with respect to subsorting relation on $\refof{\anytype}$.
\item
For every $\rtype_1,\rtype_2\in\refof{\anytype}$ there exists a least upper bound in $\refof{\anytype}$---that we denote by $\refSupOf{\rtype_1}{\rtype_2}$.
\item
For each value $\anyvalue$ the set of sorts $\refof{\typeof{\anyvalue}}$ has a minimum element w.r.t.\ $\rsub$. 
\item
The subsorting relation is \emph{sound} and \emph{complete} according to the semantics of sorts, i.e., 
\begin{center} 
$\rsubof{\rtype}{\rtype'}$ if and only if $\semOf{\rtype}\subseteq\semOf{\rtype'}$.
\end{center}
\end{enumerate}
\begin{exa}\label{exa:refinement-subsortingGROUND}
Figure~\ref{fig:refinement-subsortingGROUND} illustrates the sorts and the subsorting  for the ground types used in the examples introduced throughout the paper.
\end{exa}
 \begin{figure}[!t]{
 \framebox[1\textwidth]{
 $\begin{array}{@{\hspace{-6cm}}l}
\textbf{Ground sorts:}
  \\
\begin{array}{lcl}
\refof{\boolType} & = & \falseType,\trueType,\boolType   
\\
\refof{\numType} & = & \nnumType,\znumType,\pnumType,\znnumType,\zpnumType,\numType 
\end{array}
\\
\\
\textbf{Subsorting for ground sorts:}
\vspace{-0.5cm}
\\
\begin{tikzpicture}
\node (true) at (0:1) {\trueType};
\node (bool) at (90:1) {\boolType} edge [<-] (true);
\node (false) at (180:1) {\falseType} edge [->] (bool);
\end{tikzpicture}
\hspace{2cm}
\begin{tikzpicture}
\node (zpf) at (0:1) {\zpnumType};
\node (\numNAME) at (90:1) {\numType}  edge [<-] (zpf);  
\node (znf) at (180:1) {\znnumType}  edge [->] (\numNAME);
\node (zf) at (270:1) {\znumType}  edge [->] (zpf)   edge [->] (znf);
\node (nf) at (225:1.42) {\nnumType}  edge [->] (znf);   
\node (pf) at (315:1.44) {\pnumType}  edge [->] (zpf);  
\end{tikzpicture}
\end{array}$}
} 
\caption{Sorts and subsorting  for the ground types used in the examples}\label{fig:refinement-subsortingGROUND}
\end{figure}

\emph{Subsigning} is the partial order obtained by lifting subsorting to sort-signatures by the following subsigning rule (which, as usual, is covariant in the result sort and contravariant in the argument sorts): 
\[
\surfaceTyping{I-SIG}{  \quad
\rsubof{\rtype}{\rtype'}
\quad
\rsubof{\overline{\rtype}'}{\overline{\rtype}}
}{ \rsubof{\rtype(\overline{\rtype})}{\rtype'(\overline{\rtype'})} }
\]
According to the above explanations, for every type $\anytype$ and type-signature $\anytype(\overline{\anytype})$ we have that both  $(\refof{\anytype},\rsub)$ and  $(\refsigof{\anytype(\overline{\anytype})},\rsub)$ are partial orders.

\subsection{Stabilising \DiffusionsText}\label{sec-Progression}

Recall the notion of diffusion (Definition~\ref {def:Progression}). 
In this section we exploit sorts to formulate the notion of \emph{stabilising \diffusionText}: a predicate on the behaviour of a diffusion that will be  used to express the sufficient condition for self-stabilisation. 
The stabilising \diffusionText\ predicate specifies constraints on the behaviour of a  diffusion $\fname$ of type $\anytype_1(\anytype_1,\myldots,\anytype_n)$ by exploiting a sort-signature $\rtype(\rtype_1,\myldots,\rtype_n)\in\refsigof{\fname}$.

\begin{defi}[Stabilising \diffusionText]\label{def:StabilisingProgression}
A \diffusionText\ $\fname$ 
 is  \emph{stabilising  with respect to the sort-signature} $\rtype(\rtype_1\overline{\rtype})\in\refsigof{\fname}$ such that  $\rsubof{\rtype}{\rtype_1}$ (notation $\stabilisingOpOf{\fname}{\rtype(\rtype_1\overline{\rtype})}$) if the following conditions hold:
\begin{enumerate}
\item $\fname$ is \emph{monotonic nondecreasing} in its first argument, i.e.,  for all $\anyvalue\in\semOf{\rtype_1}$, $\anyvalue'\in\semOf{\rtype_1}$ and $\overline{\anyvalue}\in\semOf{\overline{\rtype}}$:
$\quad$
 $\anyvalue\wfLeOf{\rtype_1}\anyvalue'$
    implies
    $\opApply{\opFunOf{\fname}}{\anyvalue,\overline{\anyvalue}}\wfLeOf{\rtype_1}\opApply{\opFunOf{\fname}}{\anyvalue',\overline{\anyvalue}}$);
\item 
$\fname$ is \emph{progressive} in its
    first argument, i.e., for all $\overline{\anyvalue}\in\semOf{\overline{\rtype}}$:
        $\opApply{\opFunOf{\fname}}{\topValOf{\rtype_1},\overline{\anyvalue}}\wfEqOf{\rtype_1}\topValOf{\rtype_1}$
        and, for all
        $\anyvalue\in\semOf{\rtype_1}-\{\topValOf{\rtype_1}\}$,
        $\anyvalue\wfLtOf{\rtype_1}\opApply{\opFunOf{\fname}}{\anyvalue,\overline{\anyvalue}}$.\label{eq:motivationRH}
\end{enumerate}
We say that the sort-signature $\rtype(\overline{\rtype})$ is  \emph{stabilising  for} $\fname$ to mean that $\stabilisingOpOf{\fname}{\rtype(\overline{\rtype})}$ holds, and write $\stabsigof{\fname}$ to denote  set of  the stabilising sort-signatures for $\fname$.
\end{defi}
\begin{exa}\label{exa:stabilising-diffusion}
Consider  the library in Figure~\ref{f:code}. The following predicates hold:
\begin{itemize}
\item
$\stabilisingOpOf{\texttt{+}}{\znumType(\znumType,\znumType)}$, 
\item
$\stabilisingOpOf{\texttt{+}}{\pnumType(\zpnumType,\pnumType)}$, 
\item
$\stabilisingOpOf{\texttt{+}}{\numType(\numType,\pnumType)}$, and
\item
 $\stabilisingOpOf{\texttt{restrictSum}}{ \numType(\numType,\pnumType,\boolType)}$.
\end{itemize}
\end{exa}

Note that Condition~(\ref{eq:motivationRH}) in Definition~\ref{def:StabilisingProgression} introduces a further constraint between the sort of the first argument $\rtype_1$ and the sort of the result $\rtype$ in the  sort-signature $\rtype(\rtype_1,\myldots,\rtype_n)$ used for $\fname$. For instance, given the diffusion
\[
\deftt{\numNAME\ f}{\numNAME\  x, \numNAME\  y}{-(x+y)}
\]
 the sort signature $\nnumType(\pnumType,\pnumType)\in\refsigof{\fname}\subset\refsigof{\numType(\numType,\numType)}$ is not compatible with Condition~(\ref{eq:motivationRH}), since $\anyvalue_1,\anyvalue_2\in\semOf{\pnumType,\pnumType}$ and $\anyvalue=\opApply{\opFunOf{\fname}}{\anyvalue_1,\anyvalue_2}\in\nnumType$ imply
$\anyvalue_1\not\wfLe\anyvalue$. Namely, the sort of the result  $\rtype$ and the sort of the first argument $\rtype_1$ must be such that
$\semOf{\rtype}\semsubwrhLe\semOf{\rtype_1}$, where the relation $\semsubwrhLe$ between two subsets $\setInPowerSet$ and $\setInPowerSet_1$ of $\semOf{\anytype}$ (i.e., between elements   of the  powerset   $\powerSet{\semOf{\anytype}}$), that we call \emph{$\ReverseHoare$ inclusion}, is defined as follows:
\[
\setInPowerSet\semsubwrhLe\setInPowerSet_1 \qquad \mbox{if and only if} \qquad \setInPowerSet\subseteq\setInPowerSet_1 \quad \mbox{and} \quad \topValOf{\rtype}=\topValOf{\rtype_1}.
\]
%
%
We write $\subwrhLe$ to denote the \emph{$\ReverseHoare$ subsorting} relation, which is the restriction of 
 subsorting relation defined as follows:
\begin{center} 
$\rsubwrhof{\rtype}{\rtype'}$ if and only if $\rsemsubwrhof{\semOf{\rtype}}{\semOf{\rtype'}}$.
\end{center}
To summarise: if a sort signature  $\rtype(\overline{\rtype})$ is stabilising for some diffusion, then $\rtype(\overline{\rtype})$ must be \emph{progressive}, according to the following definition.
\begin{defi}[\ProgressiveSortSigNAME\ sort-signature]\label{def:ProgressiveSortSignature}
A sort-signature $\rtype(\overline{\rtype})$ with $\overline{\rtype}=\rtype_1,\myldots,\rtype_n$ $(n\ge 1)$ is a \emph{\progressiveSortSigNAME\ sort-signature} (notation $\progressiveSortSigOf{\rtype(\overline{\rtype})}$) if $\rsubwrhof{\rtype}{\rtype_1}$.
\end{defi}
\noindent
Given a diffusion type-signature $\anytype(\overline{\anytype})$ (cf.\ Definition~\ref{def:Progression}) we write $\progrefsigof{\anytype(\overline{\anytype})}$ to denote the (set of) \progressiveSortSigNAME\ sort-signatures  that refine it.

\begin{exa}\label{exa:refinement-reverseHoare}
 Figure~\ref{fig:refinement-reverseHoare} illustrates the $\ReverseHoare$ subsorting for the ground sorts  used in the examples introduced throughout the paper.
\end{exa}

 \begin{figure}[!t]{
 \framebox[1\textwidth]{
 $\begin{array}{@{\hspace{-6cm}}l}
\begin{tikzpicture}
\node (true) at (0:1) {\trueType};
\node (bool) at (90:1) {\boolType} edge [<-] (true);
\node (false) at (180:1) {\falseType};
\end{tikzpicture}
\hspace{2cm}
\begin{tikzpicture}
\node (zpf) at (0:1) {\zpnumType};
\node (\numNAME) at (90:1) {\numType}  edge [<-] (zpf);  
\node (znf) at (180:1) {\znnumType};
\node (zf) at (270:1) {\znumType}   edge [->] (znf);
\node (nf) at (225:1.42) {\nnumType};
\node (pf) at (315:1.44) {\pnumType}  edge [->] (zpf);   
\end{tikzpicture}
\end{array}$}
} 
\caption{Progressive subsorting for the ground sorts used in the examples  (cf.\ Figure~\ref{fig:refinement-subsortingGROUND})} \label{fig:refinement-reverseHoare}
\end{figure}

The following partial order between progressive sort-signatures, that we call \emph{stabilising subsigning}:
\[
\surfaceTyping{I-S-SIG}{  \quad
\rsubwrhof{\rtype}{\rtype'}
\quad
\rsubwrhof{\rtype'_1}{\rtype_1}
\quad
\rsubof{\overline{\rtype}'}{\overline{\rtype}}
}{ \rsubstabof{\rtype(\rtype_1\overline{\rtype})}{\rtype'(\rtype'_1\overline{\rtype}')} }
\]
 captures the natural implication relation between  stabilisation properties, as stated by the following  proposition.
\begin{prop}[Soudness of stabilising subsigning]\label{prop:Soudness of stabilising subsigning}
If $\stabilisingOpOf{\fname}{\rtype(\overline{\rtype})}$ and $\rsubstabof{\rtype(\overline{\rtype})}{\rtype'(\overline{\rtype}')}$, then  $\stabilisingOpOf{\fname}{\rtype'(\overline{\rtype'})}$.
\end{prop}
\begin{proof}
Straightforward from Definition~\ref{def:StabilisingProgression} and the definition of progressive subsigning (Rule \ruleNameSize{[I-S-SIG]} above).
\end{proof}

\subsection{The Stabilising-Diffusion Condition}\label{sec-program-with-valid-assumptions}

We have now all the ingredients to formulate a sufficient condition for self-stabilisation of a well-typed program $\PROGRAM$, that we call the \emph{stabilising-diffusion condition}: 
any \diffusionText\ $\fname$ used in a spreading expression $\e$ of the program (or library) $\PROGRAM$ must be a 
stabilising \diffusionText\  with respect a sort-signature for $\fname$ that correctly describes the sorts of the arguments of $\fname$ in $\e$.
More formally, a well-typed program (or library) $\PROGRAM$ satisfies the stabilising-diffusion condition if and only if it admits \emph{valid  sort-signature and stabilising  assumptions for diffusions}. I.e., for each diffusion-expression $\spreadThree{\e_1}{\fname}{\e_2,\myldots,\e_n}$ occurring in $\PROGRAM$, there exists a sort-signature  $\rtype(\rtype_1,\ldots,\rtype_n)$ such that the following two conditions are satisfied.
\begin{enumerate}
\item \textbf{Validity of the sort-signature assumption:}
the diffusion $\fname$ has sort-signature $\rtype(\rtype_1,\ldots,\rtype_n)$ and,
in any reachable network configuration, the evaluation of the subexpression $\e_i$ yields a value $\anyvalue_i\in\semOf{\rtype_i}$ ($1\le i\le n$).
\item \textbf{Validity of the  stabilising assumption:}
the sort-signature $\rtype(\rtype_1,\ldots,\rtype_n)$ is stabilising  for the diffusion $\fname$.
\end{enumerate}

\begin{exa}
In the body of the function \texttt{main} in Example~\ref{exa:hop-count gradient}, which defines a self-stabilising field:
\begin{enumerate}
\item   
the \diffusionText\ function \texttt{+} is applied according to the sort-signature  $\numType(\numType,\pnumType)$ since its second argument is the sensor \senv{dist} of type \texttt{\numNAME} that is guaranteed to always return a value of sort $\pnumType$; and
\item
 $\stabilisingOpOf{+}{\numType(\numType,\pnumType)}$ holds (cf.\ Example~\ref{exa:stabilising-diffusion}).
\end{enumerate}
Therefore, the stabilising diffusion  condition is satisfied. Also the library in Figure~\ref{f:code} satisfies the stabilising diffusion  condition.

Instead, for the \diffusionText\ function \texttt{id} of type-signature $\numType(\numType)$ used in the non-self-stabilising spreading expression considered in Example~\ref{exa:not-self-stabilising},   $\rtype(\rtype_1)=\znumType(\znumType)$ is  the only sort-signature such that $\stabilisingOpOf{\texttt{id}}{\rtype(\rtype_1)}$ holds.\footnote{Considering the sorts given in Figure~\ref{fig:refinement-subsortingGROUND}---c.f.\ the footnote in Example~\ref{exa:not-self-stabilising}.}
Therefore, since the sensor \senv{src} returns a value of sort $\pnumType$, the stabilising diffusion  condition cannot be satisfied.
\end{exa}

\begin{rem}[On choosing the refinements of type $\numType$]
The choice of the refinements for type \texttt{real} that we considered in the examples  is somehow arbitrary. We have chosen a set of refinements that is  expressive enough in order the show  that all the  self-stabilising examples considered in the paper satisfy the stabilising-diffusion condition. 
For instance, dropping the
 refinements $\nnumType$ and $\pnumType$ would make it impossible to show that the program considered in Example~\ref{exa:hop-count gradient} satisfies the stabilising-diffusion condition. 

Considering a richer set of refinement would allow to prove more properties of programs (and would make the check more complex). For instance,
adding the refinement $\npnumType$ (negative or positive number) such that:
\[
\nnumType\le\npnumType
\qquad\quad
\pnumType
\le
\npnumType
\qquad\quad
\npnumType\le\numType       
\]
would allow to assign (according to the sort-checking rules presented in Section~\ref{sec-ensuringSelfStabilisation}) sort  
$\npnumType$ to the expressions \texttt{(x ? -1 : 1)}, to assume sort-signature  $\falseType(\znumType,\npnumType)$ for the built-in equality operator on reals $=$, and therefore to check that the user-defined function
\[
\texttt{def <bool> f(<bool> x) is 0 = (x ? -1 : 1)}  
\]
has sort-signature $\falseType(\boolType)$. Although, this would allow  to show that more program satisfies the stabilising-diffusion, the  refinement $\npnumType$ is not needed in order to show that the  self-stabilising examples considered in the paper satisfy the stabilising-diffusion condition. Therefore we have not considered it in the examples.
\end{rem}


\section{Programs that Satisfy the Stabilising-Diffusion Condition Self-stabilise}\label{sec-properties-stabilising}

In this section we prove the main properties of the proposed calculus, namely: type soundness and termination of device computation (in Section~\ref{sec-properties-typing}), and self-stabilisation of network evolution for programs that satisfy the stabilising-diffusion condition (in Section~\ref{sec-properties-self-stabilisation}).

As already mentioned, our notion of self-stabilisation is key as it allows one to conceptually map any (self-stabilising) field expression to its final and stable field state, reached in finite time by fair evolutions and adapting to the shape of the environment.
This acts as the sought bridge between the micro-level (field expression in program code), and the macro-level (expected global outcome). In order facilitate the exploitation of this  bridge it would be useful to have an effective means for checking the stabilising-diffusion condition. A technique providing such an effective means  (for the extension  of the calculus with pairs introduced in Section~\ref{sec-calculus-extended}) is illustrated in Sections~\ref{sec-stabilisation-checking},~\ref{sec-ensuringSelfStabilisation} and~\ref{sec-checkingStabilisingDiffusion}.

\subsection{Type Soundness and Termination of Device Computation}\label{sec-properties-typing}

In order to state the properties of device computation  we introduce the notion of set of well-typed values trees for an expression.

Given an expression $\e$ such that
$\surExpTypJud{\overline{\xname}:\overline{\anytype}}{\e}{\anytype}$,
the set
$\bsWFVT{\overline{\xname}:\overline{\anytype}}{\e}{\anytype}$
of the \emph{well-typed} value-trees for $\e$, is inductively
defined as follows:
    $\vtree$ $\in$ $\bsWFVT{\overline{\xname}:\overline{\anytype}}{\e}{\anytype}$
    if there exist
    \begin{itemize}
    \item a sensor mapping $\snsFun$;
    \item well-formed tree environments
        $\overline{\vtree}\in\bsWFVT{\overline{\xname}:\overline{\anytype}}{\e}{\anytype}$;
        and
    \item values $\overline{\anyvalue}$ such that
        $\lengthOf{\overline{\anyvalue}}=\lengthOf{\overline{\xname}}$
        and
        $\surExpTypJud{\emptyset}{\overline{\anyvalue}}{\overline{\anytype}}$;
    \end{itemize}
    such that
    $\bsopsem{\snsFun}{\overline{\vtree}}{\applySubstitution{\e}{\substitution{\overline{\xname}}{\overline{\anyvalue}}}}{\vtree}$
    holds---note that this definition is inductive, since the sequence  of evaluation trees $\overline{\vtree}$ may be empty.

As this notion is defined we can state the following two theorems, guaranteeing that from a properly typed environment,
    evaluation of a well-typed expression yields a properly typed
    result and always terminates, respectively.

\begin{thm}[Device computation type
preservation]\label{the-DeviceTypePreservation} If
$\surExpTypJud{\overline{\xname}:\overline{\anytype}}{\e}{\anytype}$,
 $\snsFun$ is a sensor mapping,
$\overline{\vtree}\in\bsWFVT{\overline{\xname}:\overline{\anytype}}{\e}{\anytype}$,
$\lengthOf{\overline{\anyvalue}}=\lengthOf{\overline{\xname}}$,
        $\surExpTypJud{\emptyset}{\overline{\anyvalue}}{\overline{\anytype}}$
and
$\bsopsem{\snsFun}{\overline{\vtree}}{\applySubstitution{\e}{\substitution{\overline{\xname}}{\overline{\anyvalue}}}}{\vtree}$,
then $\surExpTypJud{\emptyset}{\vrootOf{\vtree}}{\anytype}$.
\end{thm}

\begin{proof}
See Appendix~\ref{app-proof-typeSounness}.
\end{proof}

\begin{thm}[Device computation termination]\label{the-DeviceTermination} If
$\surExpTypJud{\overline{\xname}:\overline{\anytype}}{\e}{\anytype}$,
 $\snsFun$ is a sensor mapping,
$\overline{\vtree}\in\bsWFVT{\overline{\xname}:\overline{\anytype}}{\e}{\anytype}$,
$\lengthOf{\overline{\anyvalue}}=\lengthOf{\overline{\xname}}$
and
        $\surExpTypJud{\emptyset}{\overline{\anyvalue}}{\overline{\anytype}}$,
then
$\bsopsem{\snsFun}{\overline{\vtree}}{\applySubstitution{\e}{\substitution{\overline{\xname}}{\overline{\anyvalue}}}}{\vtree}$
for some value-tree $\vtree$.
\end{thm}

\begin{proof}
See Appendix~\ref{app-proof-typeSounness}.
\end{proof}

The two theorems above guarantee type-soundness and termination of device computations, that is: 
the evaluation of a well-typed program on any device completes without type errors  assuming that the values received from the sensors and the values in the value trees received from the  neighbours are well-typed.

\subsection{Self-stabilisation of Network Evolution for Programs  that Satisfy the Stabilising-Diffusion Condition}\label{sec-properties-self-stabilisation}

On top of the type soundness and termination result for device computation we can prove the main technical result of the paper: self-stabilisation of any program 
that satisfies the stabilising-diffusion condition.

\begin{thm}[Network self-stabilisation for programs  that satisfy the stabilising-diffusion condition]\label{the-network-self--stabilization}
Given a program with valid  sort and stabilising diffusion assumptions,  
every reachable network configuration $\network$ self-stabilises, i.e.,
there exists
$k\ge 0$ such that $\network\netArrFairStarN{k}\network'$
implies that $\network'$ is self-stable.
\end{thm}
\begin{proof}
See Appendix~\ref{app-proof-selfStabilisation}.
\end{proof}

We conclude this section by giving an outline of the proof of Theorem~\ref{the-network-self--stabilization}. To this aim we first introduce some auxiliary definitions. 

\subsubsection*{Auxiliary Definitions}

In the following we omit the subscript $\rtype$ in
$\wfLtOf{\rtype}$ and $\wfLeOf{\rtype}$ when
     it is clear from the context (i.e., we just write 
    $\wfLt$ and $\wfLe$).

Given a network $\network$ with main expression $\e$, we  write $\vtreeInNC{\deviceId}{\network}$ to denote the
    value-tree of $\e$ on device $\deviceId$ in the network configuration
    $\network$, and write
    $\anyvalueInNC{\deviceId}{\network}$ to denote the value
    $\vrootOf{\vtreeInNC{\deviceId}{\network}}$ of $\e$ on
    device $\deviceId$ in the network configuration
    $\network$. 
 Moreover, when  $\e=\spreadThree{\e_0}{\fname}{\e_1,\ldots,\e_n}$
we write $\vtreeInNC{j,\deviceId}{\network}$ and write
    $\anyvalueInNC{j,\deviceId}{\network}$ to denote  the
    value-tree and the value
    of $\e_j$ ($0\le j\le n$), respectively. In the following we omit to specify the
    network configuration $\network$ when is is clear from
    the context, i.e., we simply write
    $\vtree_{\deviceId}$, $\anyvalue_{\deviceId}$, $\vtree_{j,\deviceId}$ and $\anyvalue_{j,\deviceId}$.

We say that a device $\deviceId$ is \emph{stable} in $\network$ to mean that $\network\netArrPresStar\network'$ implies
      $\vtreeInNC{\deviceId}{\network}=\vtreeInNC{\deviceId}{\network'}$. Note that the following three statements are equivalent:
\begin{itemize}
\item
 $\network$ is stable. 
\item
All the devices of $\network$ are stable.
  \item $\network\netArrPresStar\network'$ implies
      $\network=\network'$.
  \end{itemize}
 We write $\netframeOf{\network}$ to denote the environment
 $ \Envi$ of a network configuration
 $\network=\SystS{\Envi}{\Field}$.
We say that a device $\deviceId$ is \emph{self-stable} in a network $\network$ to mean that it is stable and its value is univocally determined by $\netframeOf{\network}$. Note that a network is self-stable if and only if all its devices are self-stable.

We say that a network $\network$ with main expression $\e=\spreadThree{\e_0}{\fname}{\e_1,\ldots,\e_n}$ is \emph{pre-stable} to mean that for every device $\deviceId$ in $\network$:
\begin{enumerate}
\item
the subexpressions $\e_i$ ($0 \le i\le n$)  are stable, and
\item
$\anyvalue_{\deviceId}\wfLe\anyvalue_{0,\deviceId}$.
\end{enumerate}
We say that the network $\network$ is \emph{pre-self-stable} to mean that it is pre-stable and the value trees of the subexpressions $\e_i$ ($0 \le i\le n$) are self-stable (i.e, they are univocally determined by $\netframeOf{\network}$).
Note that pre-stability is preserved by
firing evolution (i.e., if $\network$ is
pre-stable and $\network\netArrPresStar\network'$, then
$\network'$ is pre-stable).

\subsubsection*{An Outline of the Proof of Theorem~\ref{the-network-self--stabilization}}

The proof is by induction on the syntax of closed expressions $\e$ and on
the number of function calls that may be encountered during the
evaluation of $\e$.
Let $\e$ be the main expression of $\network$ and  $\netframe=\netframeOf{\network}$.
The only interesting case is when $\e$ is a spreading expression
$\spreadThree{\e_0}{\fname}{\e_1,\ldots,\e_n}$.  
    By induction there exists $h\ge 0$ such that if
    $\network\netArrFairStarN{h}\network_1$ then on every
    device $\deviceId$, the evaluation of
    $\e_0,\e_1,\ldots,\e_n$ produce stable value-trees
    $\vtree_{0,\deviceId},\vtree_{1,\deviceId},\ldots,\vtree_{n,\deviceId}$,
    which are univocally determined by $\netframe$. 
    Note that, if $\network\netArrFairStarN{h+1}\network_2$
    then  $\network_2$ is
    pre-self-stable. Therefore we focus on the case when $\network$ is pre-self-stable.
The proof of this case is based on the following auxiliary results (which corresponds to Lemmas~\ref{lem-minimum-value}-\ref{lem-pre-self-stable-network-self-stabilization} of Appendix~\ref{app-proof-selfStabilisation}).

\begin{description}
\item[\ref{lem-minimum-value} (Minimum value)] Any $1$-fair evolution $\network\netArrFairStarN{1}\network'$ increases the value of any not self-stable  device $\deviceId$ in $\network$ such that
$\anyvalueInNC{\deviceId}{\network}$ is minimum (among
the values of the devices in $\network$). The new value $\anyvalueInNC{\deviceId}{\network'}$ is such that there exists a value $\anyvalue'$ such that
    $\anyvalueInNC{\deviceId}{\network}\wfLt\anyvalue'\wfLe\anyvalueInNC{\deviceId}{\network'}$
and  in any subsequent firing evolution $\network'\netArrPresStar\network''$ the value of the device $\deviceId$ will be always greater or equal to $\anyvalue'$ (i.e.,
    $\anyvalue'\wfLe\anyvalueInNC{\deviceId}{\network''}$). 
\item[\ref{lem-self-stabilisation-minimum-value} (Self-stabilisation of the minimum value)] Let $\stableNodeSet_1$ be the subset of the devices in $\network$ 
 such that
$\anyvalue_{0,\deviceId}$ is minimum  (among the values of
$\e_0$ in the devices in $\network$). There exists $k\ge 0$ such that any $k$-fair evolution  $\network\netArrFairStarN{k}\network'$
is such that
\begin{enumerate}
\item
each device $\deviceId$ in $\stableNodeSet_1$ is self-stable in $\network'$.
\item
in $\network'$ each device not in $\stableNodeSet_1$ has a value greater or equal then the values of the devices in $\stableNodeSet_1$ and, during any firing evolution, it will always assume  values greater than the values of the devices in $\stableNodeSet_1$.
\end{enumerate}
\item[\ref{lem-frontier} (Frontier)] Let $\nodeSubSet$ be a set of devices. Given a set of stable devices
$\stableNodeSet\subseteq\nodeSubSet$ we wrote
$\frontierOfIn{\stableNodeSet}{\nodeSubSet}$ to denote the
subset of the devices
$\deviceId\in\nodeSubSet-\stableNodeSet$ such that there
exists a device $\deviceId'\in\stableNodeSet$ such that
$\deviceId'$ is a neighbour of $\deviceId$. If $\nodeSubSet$ are devices of the network and $\stableNodeSet$ satisfies  the following conditions
\begin{enumerate}[label=(\roman*)]
\item
the condition obtained form condition (1) above by replacing $\stableNodeSet_1$ with $\stableNodeSet$,
\item
the condition obtained form condition (2) above by replacing $\stableNodeSet_1$ with $\stableNodeSet$, and
\item
$\frontierOfIn{\stableNodeSet}{\nodeSubSet}\not=\emptyset$,
\end{enumerate}
then any $1$-fair evolution makes the devices in $\frontierOfIn{\stableNodeSet}{\nodeSubSet}$ self-stable.
\item[\ref{lem-minimum-value-outside} (Minimum value not in  $\stableNodeSet$)] 
 If $\nodeSubSet$ are devices of the network and $\stableNodeSet$ satisfies conditions \emph{(i)}-\emph{(iii)} above, satisfies the following condition
\begin{itemize}
\item[\emph{(iv)}]
each device in $\frontierOfIn{\stableNodeSet}{\nodeSubSet}$ is self-stable in $\network$, and
\end{itemize}
$\sourceNodeSubSet\subseteq\nodeSubSet-\stableNodeSet$ is the set of devices $\deviceId$  such that
$\anyvalueInNC{\deviceId}{\network}$ is minimum (among
the values of the devices in $\nodeSubSet-\stableNodeSet$), and $\sourceNodeSubSet\cap\frontierOfIn{\stableNodeSet}{\nodeSubSet}=\emptyset$, then any $1$-fair evolution $\network\netArrFairStarN{1}\network'$ increases the value of any not self-stable  device $\deviceId$ in $\sourceNodeSubSet$. The new value $\anyvalueInNC{\deviceId}{\network'}$ is such that there exists a value $\anyvalue'$ such that
    $\anyvalueInNC{\deviceId}{\network}\wfLt\anyvalue'\wfLe\anyvalueInNC{\deviceId}{\network'}$
and  in any subsequent firing evolution $\network'\netArrPresStar\network''$ the value of the device $\deviceId$ will be always greater or equal to $\anyvalue'$ (i.e.,
    $\anyvalue'\wfLe\anyvalueInNC{\deviceId}{\network''}$). 
\item[\ref{lem-self-stabilisation-minimum-value-outside} (Self-stabilisation of the minimum value not in  $\stableNodeSet$)] 
 If $\nodeSubSet$ are devices of the network and $\stableNodeSet$ satisfies conditions \emph{(i)}-\emph{(iv)} above, and $\sourceNodeSubSet\subseteq\nodeSubSet-\stableNodeSet$ is the set of devices $\deviceId$  such that
$\anyvalueInNC{\deviceId}{\network}$ is minimum (among
the values of the devices in $\nodeSubSet-\stableNodeSet$), then there exists $k\ge 0$ such that any $k$-fair evolution  $\network\netArrFairStarN{k}\network'$
is such that there exists a device $\deviceId_1$ in  $\nodeSubSet-\stableNodeSet$ such that  $\stableNodeSet_1=\stableNodeSet\cup\{\deviceId_1\}$
satisfies the conditions (1) and (2) above.
\item[\ref{lem-pre-self-stable-network-self-stabilization} (Pre-self-stable network self-stabilization)] For every reachable pre-self-stable network configuration $\network$ there exists
$k\ge 0$ such that $\network\netArrFairStarN{k}\network'$
implies that $\network'$ is self-stable. This sixth  auxiliary result, which concludes our proof outline, follows from the previous five auxiliary results.
The idea is to consider the auxiliary results~\ref{lem-self-stabilisation-minimum-value}, \ref{lem-frontier} and~\ref{lem-self-stabilisation-minimum-value-outside}  as reasoning steps that may be iterated.
We start by applying the  auxiliary result~\ref{lem-self-stabilisation-minimum-value} to produce a non-empty set of devices $\stableNodeSet_1$ that satisfies conditions (1) and (2) above. Then, we rename $\stableNodeSet_1$ to $\stableNodeSet$ and iterate the following  two reasoning steps until the set of devices $\stableNodeSet$ is such that  $\frontierOfIn{\stableNodeSet}{\nodeSubSet}=\emptyset$:
\begin{itemize}
\item
apply the auxiliary results~\ref{lem-frontier} and~\ref{lem-self-stabilisation-minimum-value-outside} to produce
a non-empty set of devices  $\stableNodeSet_1$ that satisfies conditions (1) and (2) above; and
\item
rename $\stableNodeSet_1$ to $\stableNodeSet$.
\end{itemize} 
Clearly the number of iterations is finite (note that  $\stableNodeSet=\nodeSubSet$ implies $\frontierOfIn{\stableNodeSet}{\nodeSubSet}=\emptyset$).
If  $\stableNodeSet=\nodeSubSet$ we have done.
Otherwise note that, since $\frontierOfIn{\stableNodeSet}{\nodeSubSet}=\emptyset$,
 the evolution of the devices in
$\nodeSubSet-\stableNodeSet$ is independent from the
devices in $\stableNodeSet$. Therefore we we can iterate the whole reasoning (i.e., starting from the  auxiliary result~\ref{lem-self-stabilisation-minimum-value}) on the  the portion of the network with devices in $\nodeSubSet-\stableNodeSet$.
\end{description}

\section{Extending the Calculus with Pairs}\label{sec-calculus-extended}

The calculus presented in Sections~\ref{sec-fields} and~\ref{sec-calculus} does not model data structures. In this section we 
point out that the definitions and results presented in Sections~\ref{sec-stabilisingProgression} and~\ref{sec-properties-stabilising}  are  parametric in the set of modeled types by considering an extension of the calculus that models the pair data structure. Recent works, such as \cite{BV-FOCAS2014}, show that the ability of modelling pairs of values is key in realising context-dependent propagation of information, where along with distance one let the computational field carry additional information as gathered during propagation. In fact, pairs are needed for programming further interesting examples (illustrated in Section~\ref{sec-examples-extended}), and their introduction here is useful to better grasp the subtleties of our checking algorithm for self-stabilisation, as described in next sections.

\subsection{Syntax}\label{sec-basicIngredients-extended}

The extensions to the  syntax of the calculus  for modeling pairs are reported in Figure~\ref{fig:source:syntax-extended}. Now types include \emph{pair types} (like \texttt{<\numNAME,bool>}, \texttt{<\numNAME,<bool,\numNAME>>},...\ and so on), expressions include \emph{pair construction}   ($\mkpair{\e}{\e}$) and \emph{pair deconstruction} ($\fstK\,\e$  or $\sndK\,\e$), and values includes  \emph{pair values}
(\texttt{<1,TRUE>}, \texttt{<2,3.5>}, \texttt{<<1,FALSE>,3>},...\ and so on). 
\begin{figure}[!t]{
\centerline{\framebox[\textwidth]{$
\begin{array}{lcl@{\hspace{7cm}}r}
 \anytype & \BNFcce &  \groundType \;
    \; \BNFmid \;  \mkpairType{\anytype}{\anytype}
    &   {\mbox{\footnotesize type}} \\[2pt]
 \e & \BNFcce & \cdots \;
    \; \BNFmid \;  \mkpair{\e}{\e}
    \; \BNFmid \;  \fstK \; \e
    \; \BNFmid \;  \sndK \; \e
    &   {\mbox{\footnotesize expression}} 
\\
 \end{array}
 $}}}
 \caption{Extensions to the syntax of types and expressions (in Fig.~\ref{fig:source:syntax}) to model pairs}
\label{fig:source:syntax-extended}
\end{figure}

The ordering for ground types has to be somehow lifted to non-ground types.  A naural choice for pairs is the lexicographic preorder, i.e, to define $\mkpair{\anyvalue_1}{\anyvalue_2}\wfLeOf{\mkpairType{\anytype_1}{\anytype_2}}\mkpair{\anyvalue'_1}{\anyvalue'_2}$ if  either $\anyvalue_1\wfLtOf{\anytype_1}\anyvalue'_1$ holds or both $\anyvalue_1=\anyvalue'_1$ and $\anyvalue_2\wfLeOf{\anytype_2}\anyvalue'_2$ hold.

\subsection{Examples}\label{sec-examples-extended}

In this section we build on the examples introduced in Section~\ref{sec-examples} and show how pairs can be used to program
further kinds of behaviour useful in self-organisation mechanisms in general, and also in the specific case of crowd steering, as reported in Figure~\ref{f:code-extended}.
%

\begin{figure}[t]
\[{\footnotesize \begin{array}{|l|}
\hline
\deftt{<\numNAME,bool> \sumOrNAME}{<\numNAME,bool> x, <\numNAME,bool> y}{}\texttt{<fst x + fst y, snd x or snd y>}\\[1mm]
\deftt{<\numNAME,bool>  \ptRealBoolNAME}{<\numNAME, bool> x}{((\fstK\ x)=\topNumK) ? <\topNumK,\trueK> : x}\\[1mm]
\deftt{<\numNAME,bool> \spSumOrNAME}{<\numNAME,bool> x, <\numNAME,bool> y}{\ptRealBoolNAME(\sumOrNAME(x,y))}\\[1mm]
\deftt{bool sector}{\numNAME\ i, bool c}{snd \spreadtt{<i, c>}{sum\_or(\startt,<\senv{dist},c>)}}\\ \hline\\[-2mm]
\deftt{<\numNAME,\numNAME> \addToFstNAME}{<\numNAME,\numNAME> x, \numNAME\ y}{<fst x + y, snd x>}\\[1mm]
\deftt{<\numNAME,\numNAME>  \ptRealRealNAME}{<\numNAME, \numNAME> x}{((\fstK\ x)=\topNumK) ? <\topNumK,\topNumK> : x}\\[1mm]
\deftt{<\numNAME,\numNAME> \spAddToFstNAME}{<\numNAME,\numNAME> x, \numNAME\ y}{\ptRealRealNAME(\addToFstNAME(x,y))}\\[1mm]
\deftt{<\numNAME,\numNAME> gradcast}{\numNAME\ i, \numNAME\ j}{\spreadtt{<i, j>}{add\_to\_1st(\startt, \senv{dist})}}\\ \hline\\[-2mm]
\deftt{\numNAME\ dist}{\numNAME\ i, \numNAME\ j}{gradcast(restrict(j,j==0),grad(i))}\\[1mm]
\deftt{bool path}{\numNAME\ i, \numNAME\ j, \numNAME\ w}{(grad(i)+grad(j))+(-w) < dist(i, j)}\\[1mm]
\deftt{\numNAME\ channel}{\numNAME\ i, \numNAME\ j, \numNAME\ w}{gradobs(grad(j),not path(i, j, w))}\\
\hline
\end{array}}\]
\caption{Definitions of examples using pairs  (see Fig.~\ref{f:code} for the definitions of functions \texttt{restrict}, \texttt{grad} and \texttt{gradobs})}\label{f:code-extended}\end{figure}

The fourth function in Figure~\ref{f:code-extended}, called \texttt{sector}, can be used to keep track of specific
situations during the propagation process.
It takes a zero-field source \texttt{i} (a field holding value $0$ in a ``source'', as usual) and a boolean field
\texttt{c} denoting an area of interest: it creates a gradient
of pairs, orderly holding distance from source and a boolean
value representing whether the route towards the source crossed
area \texttt{c}.
As one such gradient is produced, it is wholly applied to operator \texttt{snd}, extracting a sector-like boolean field as shown in Figure \ref{f:grad} (left).
To do so, we use a special \diffusionText\ function \texttt{sum\_or} working on \texttt{\numNAME,bool} pairs, which sums the first components, and apply disjunction to the second.
In crowd steering, this pattern is useful to make people be aware of certain areas that the proposed path would cross, so as to support proper choices among alternatives \cite{MPV-SASO2012}.
Figure \ref{f:grad-extended} (left) shows a pictorial representation.

However, our self-stabilisation result reveals a subtlety. Function \sumOrNAME\ has to be tuned by composing it with \ptRealBoolNAME\  (which propagates the top value from the first to the second component of the pair), leading to function \spSumOrNAME\  (which is a stabilising diffusion).
This is needed to make sure that the top value \texttt{<\topNumK,\trueK>} of pairs  of type \texttt{<\numNAME,bool>} is used when distance reaches \topNumK: this is required to achieve progressiveness, and hence self-stabilisation.
Without it, in the whole area where distance is \topNumK{} we would have a behaviour similar to that of Example~\ref{exa:not-self-stabilising}: in particular, if \texttt{c} is \texttt{true} and \texttt{i} is \topNumK{} everywhere, both states where all nodes have second component equal to \texttt{true} (state $s_1$) and where all nodes have second component equal or \texttt{false} (state $s_2$) would be stable, and an even temporaneous flip of \texttt{c} to \texttt{false} in some node would make the system inevitable move to $s_2$---a clear indication of non self-stabilisation.

Note that \texttt{sector} function can be easily changed to propagate values of any sort by changing the type of the second component of pairs, and generalising over the \texttt{or} function.
E.g., one could easily define a spreading of ``set of values'' representing the obstacles encountered during the spread.

\begin{figure}[t]
\begin{center}
\includegraphics[width=0.2\textwidth]{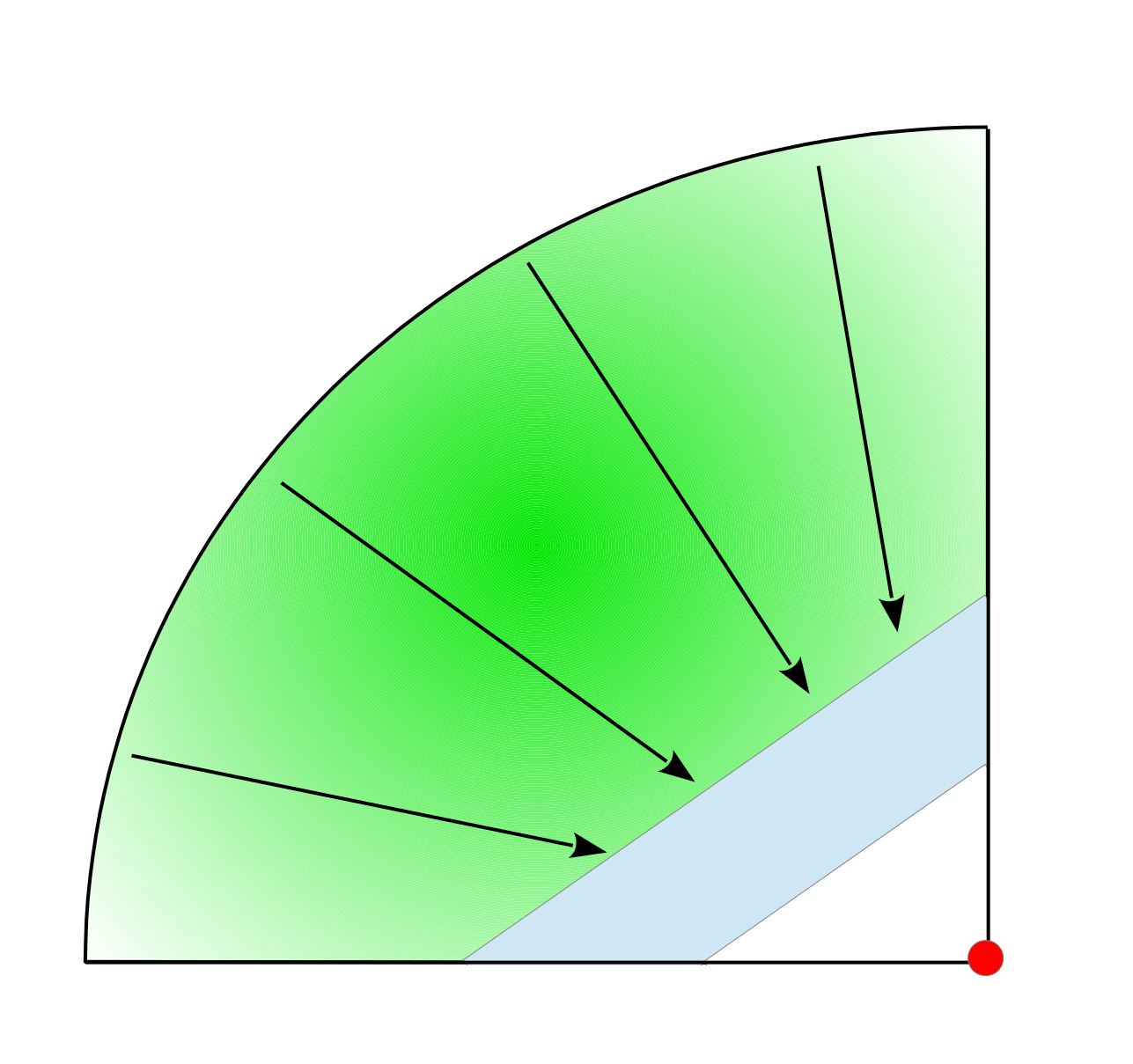}~~
\hspace{0.1\textwidth}
\includegraphics[width=0.2\textwidth]{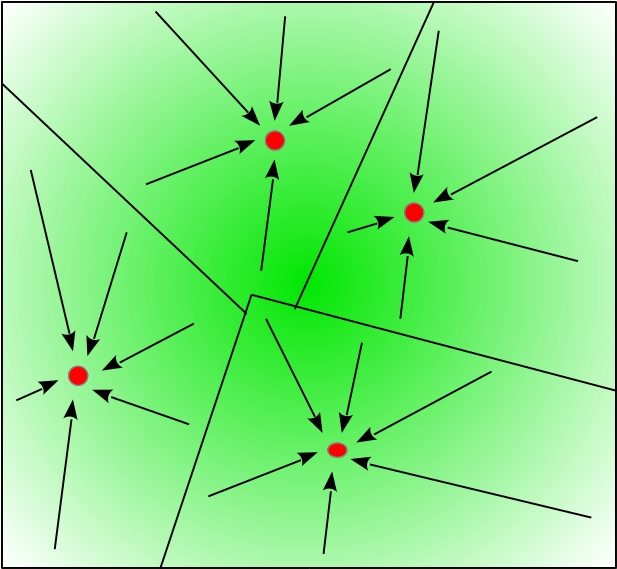}~~
\hspace{0.1\textwidth}
\includegraphics[width=0.2\textwidth]{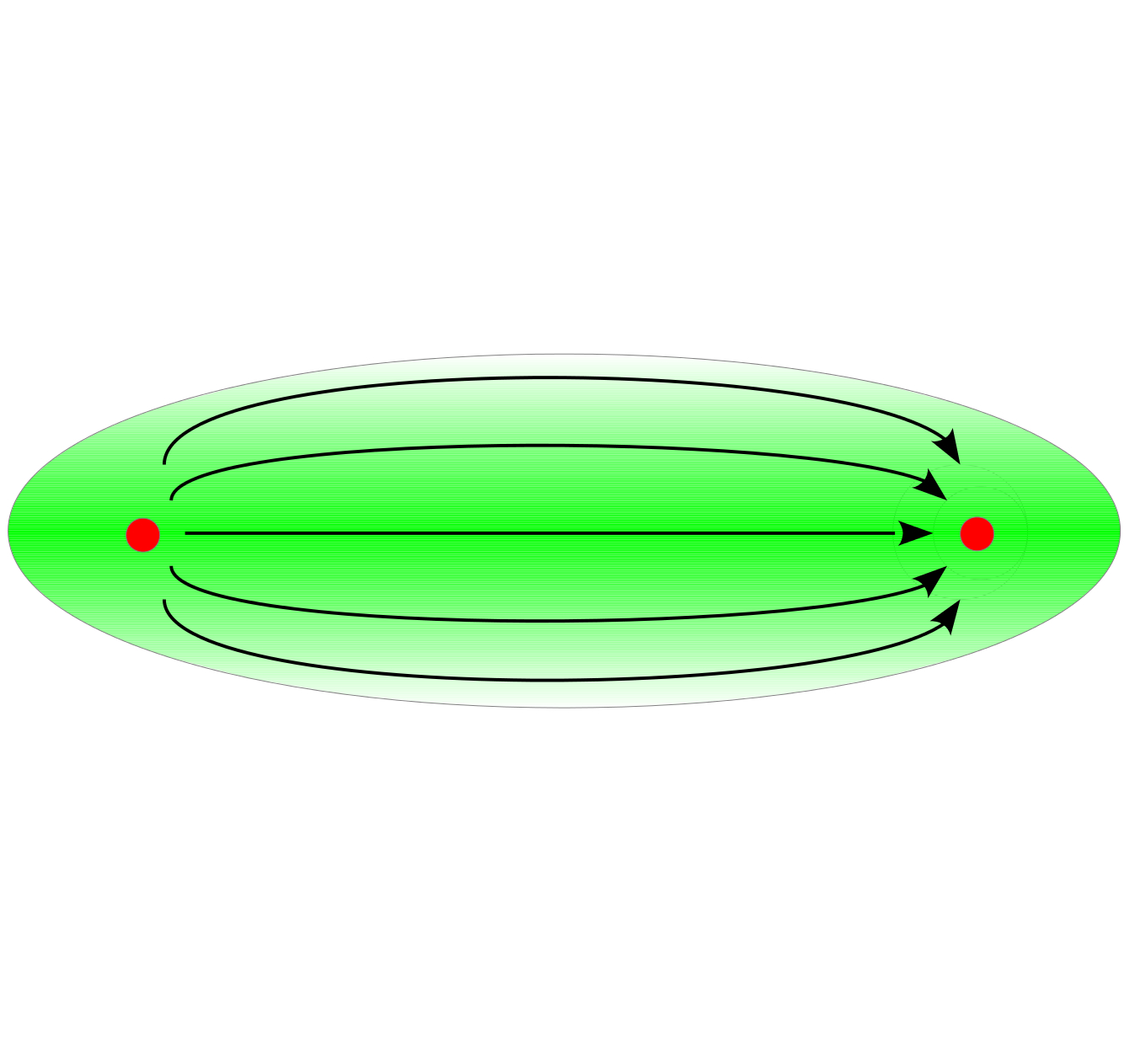}~~
\caption{Pictorial representation of  sector field (left), partition (center) and channel field (right)}\label{f:grad-extended}
\end{center}
\end{figure}

Another interesting function working on pairs is \texttt{gradcast}, that is useful to let a gradient carry some information existing in the source.
Expression \texttt{gradcast(i,j)}, where \texttt{i} is a $0$-field as usual, results in a field of pairs in which the first component is the minimum distance from a source and the second component is the value of \texttt{j} at that source.
Assuming \texttt{j} is a field of unique values in each node, e.g. a unique  identifier ID, \texttt{gradcast} performs a partitioning behaviour (namely, a so-called Voronoi partition): its second component forms a partition of the network into regions based on ``closeness'' to sources, and each region of the partition holds the value of \texttt{j} in that source; first component still forms a gradient towards such sources.
See Figure \ref{f:grad-extended} (center). 
Again, function \addToFstNAME\ has to be changed to \spAddToFstNAME\ to preserve self-stabilisation.

The remaining functions \texttt{dist}, \texttt{path} and \texttt{channel} are used to obtain a spatial pattern more heavily relying on multi-level composition, known as \emph{channel} \cite{VDB-FOCLASA-CIC2013,mitproto}.
Assume \texttt{i} and \texttt{j} are 0-fields, and suppose to be willing to steer people in complex and large environments from area \texttt{i} to destination
\texttt{j}, i.e., from a node where
\texttt{i} holds $0$ to a node where \texttt{j} holds $0$.
Typically, it is important to activate the steering service (spreading information,
providing signs, and detecting contextual information such as
congestion) only along the shortest path, possibly properly
extended of a distance width \texttt{w} to tolerate some
randomness of people movement---see Figure \ref{f:grad-extended} (right).
%
%
Function \texttt{dist} uses \texttt{gradcast}
to broadcasts the distance $d$ between \texttt{i} and
\texttt{j}---i.e., the minimum distance between a node where
\texttt{i} holds $0$ and a node where \texttt{j} holds $0$.
This is done by sending a gradcast from the source of
\texttt{j} holding the value  of \texttt{grad(i)} there, which
is exactly the distance $d$.
Function \texttt{path} simply marks
as positive those nodes whose distance from the shortest path
between \texttt{i} and \texttt{j} is smaller than \texttt{w}.
Finally, function \texttt{channel}
generates from \texttt{j} a gradient confined inside
\texttt{path(i,j,w)}, which can be used to steer people towards
the POI at \texttt{j} without escaping the path area.

\subsection{Type checking}\label{sec-typing-extended}

The typing rules for pair construction and deconstruction expressions are given in Figure~\ref{fig:SurfaceTyping-extended}.
 \begin{figure}[!t]{
 \framebox[1\textwidth]{
 $\begin{array}{l}
\begin{array}{c}
\surfaceTyping{T-PAIR}{  \quad
\surExpTypJud{\SurTypEnv}{\e_1}{\anytype_1}
\quad
\surExpTypJud{\SurTypEnv}{\e_2}{\anytype_2}
}{ \surExpTypJud{\SurTypEnv}{\mkpair{\e_1}{\e_2}}{\mkpairType{\anytype_1}{\anytype_2}} }
\quad
\surfaceTyping{T-FST}{ \;\;
\surExpTypJud{\SurTypEnv}{\e}{\mkpairType{\anytype_1}{\anytype_2}} }{
\surExpTypJud{\SurTypEnv}{\fstK\;\e}{\anytype_1} }
\quad
\surfaceTyping{T-SND}{ \;\;
\surExpTypJud{\SurTypEnv}{\e}{\mkpairType{\anytype_1}{\anytype_2}} }{
\surExpTypJud{\SurTypEnv}{\sndK\;\e}{\anytype_2} }
\end{array}
\end{array}$}
} 
\caption{Type-checking rules for pair construction and deconstruction expressions (cf.\ Fig.~\ref{fig:SurfaceTyping})} \label{fig:SurfaceTyping-extended}
\end{figure}
Note that adding these rules to the rules in Fig.~\ref{fig:SurfaceTyping} preserves the property that no choice may be done when building a derivation for a given  type-checking judgment, so the type-checking rules straightforwardly describe a type-checking algorithm (cf.\ end of Section~\ref{sec-typing}).

\begin{exa}\label{exa:diffusion-extended}
Consider the library in Figure~\ref{f:code-extended}. The following predicates hold:
\begin{itemize}
\item
 $\progressionOpOf{\sumOrNAME}$,
\item
 $\progressionOpOf{\spSumOrNAME}$,
\item
 $\progressionOpOf{\addToFstNAME}$, and
\item
 $\progressionOpOf{\spAddToFstNAME}$
\end{itemize}
\end{exa}

\begin{exa}
 The library in Figure~\ref{f:code-extended} type checks by using the ground types, sensors, and type-signatures for built-in functions  in Figure~\ref{fig:source-instance}.
\end{exa}

\subsection{Device Computation}\label{sec-device-comp-extended}

The big-step operational semantics rules for pair construction and deconstruction expressions are given in
Figure~\ref{fig:deviceSemantics-extended}. Note that they are syntax directed (in particular, the first premise of rule
\ruleNameSize{[E-PAIR]} ensures that there is no conflict with rule \ruleNameSize{[E-VAL]} of Fig.~\ref{fig:deviceSemantics}).

\begin{figure}[!t]{
 \framebox[1\textwidth]{
 $\begin{array}{l}
\begin{array}{rcl@{\hspace{8.5cm}}r}
 \anyvalue & \BNFcce &  \groundvalue \;  \; \BNFmid \;  \mkpair{\anyvalue}{\anyvalue}   &   {\mbox{\footnotesize value}} \\[2pt]
\end{array}\\
\hline\\[-8pt]
\begin{array}{c}
\surfaceTyping{E-PAIR}{
\\
\mkpair{\e_1}{\e_2} \; \textrm{not a value}
\quad
  \bsopsem{\senstate}{\premiseNumOf{1}{\overline\vtree}}{\e_1}{\vtreealt_1}
  \quad
    \bsopsem{\senstate}{\premiseNumOf{2}{\overline\vtree}}{\e_2}{\vtreealt_2}
  \quad
  \anyvalue_1=\vrootOf{\vtreealt_1}
 \quad
  \anyvalue_2=\vrootOf{\vtreealt_2}
 }{
\bsopsem{\senstate}{\overline\vtree}{\mkpair{\e_1}{\e_2}}{\mkpair{\anyvalue_1}{\anyvalue_2}(\vtreealt_1,\vtreealt_2)}
}
\skiptransition
\surfaceTyping{E-FST}{
          \\
    \bsopsem{\senstate}{\premiseNumOf{1}{\overline\vtree}}{\e}{\vtreealt}
  \quad
  \mkpair{\anyvalue_1}{\anyvalue_2}=\vrootOf{\vtreealt}
 }{
\bsopsem{\senstate}{\overline\vtree}{\fstK\;\e}{\anyvalue_1(\vtreealt)}
}
\qquad\qquad
\surfaceTyping{E-SND}{
          \\
    \bsopsem{\senstate}{\premiseNumOf{1}{\overline\vtree}}{\e}{\vtreealt}
  \quad
  \mkpair{\anyvalue_1}{\anyvalue_2}=\vrootOf{\vtreealt}
 }{
\bsopsem{\senstate}{\overline\vtree}{\sndK\;\e}{\anyvalue_2(\vtreealt)}
}
\skiptransition
\end{array}
\end{array}$}
} 
 \caption{Big-step operational semantics for pair construction and deconstruction expressions (cf.\ Fig.~\ref{fig:deviceSemantics})} \label{fig:deviceSemantics-extended}
\end{figure}

\subsection{Sorts}\label{sec-ref-types-extended}

Each sort has the same structure of the type it refines.
For instance, considering the sorts for ground types given in Figure~\ref{fig:refinement-subsortingGROUND}, we can build $36 (=6^2)$ \emph{pair sorts} for the pair type $\mkpairType{\numType}{\numType}$:
\[\begin{array}{llllll}
\mkpairType{\nnumType}{\nnumType}, & \mkpairType{\nnumType}{\znnumType}, & \mkpairType{\nnumType}{\znumType}, & \mkpairType{\nnumType}{\zpnumType}, & \mkpairType{\nnumType}{\pnumType}, & \mkpairType{\nnumType}{\numType}
\\
\vdots & \vdots & \vdots & \vdots & \vdots & \vdots
\\
\mkpairType{\numType}{\nnumType}, & \mkpairType{\numType}{\znnumType}, & \mkpairType{\numType}{\znumType}, & \mkpairType{\numType}{\zpnumType}, & \mkpairType{\numType}{\pnumType}, & \mkpairType{\numType}{\numType}
\end{array}
\] 
and $108 (36*3)$ sorts for the type $\mkpairType{\mkpairType{\numType}{\numType}}{\boolType}$:
\[
\mkpairType{\mkpairType{\nnumType}{\nnumType}}{\falseType}, \quad \mkpairType{\mkpairType{\nnumType}{\znnumType}}{\falseType}, \quad \myldots, \quad \mkpairType{\mkpairType{\numType}{\numType}}{\boolType}.
\]

Subsorting between ground sorts can be lifted to the sorts for non-ground types by suitable subsorting rules.
The following subsorting rule:
\[
\surfaceTyping{I-PAIR}{  \quad
\rsubof{\rtype_1}{\rtype'_1}
\quad
\rsubof{\rtype_2}{\rtype'_2}
}{ \rsubof{\mkpairType{\rtype_1}{\rtype_2}}{\mkpairType{\rtype'_1}{\rtype'_2}} }
\]
lifts subsorting between ground sorts  to pair sorts by 
modelling pointwise ordering on pairs. Note that the subsorting relation is determined by the subsorting for ground sorts.  Using the inclusions $\rsubof{\nnumType}{\znnumType}$ and $\rsubof{\trueType}{\boolType}$, and the above rule it is possible to derive, e.g., the inclusion  
$\rsubof{\mkpairType{\mkpairType{\nnumType}{\nnumType}}{\trueType}}{\mkpairType{\mkpairType{\nnumType}{\znnumType}}{\boolType}}$. 
Note that no choice may be done when building a derivation for a given  subsorting judgment $\rsubof{\rtype_1}{\rtype_2}$, so the subsorting rules (i.e, the rule \ruleNameSize{[I-PAIR]} and the subsorting for ground sorts) describe a deterministic algorithm.  

Similarly, 
$\ReverseHoare$ subsorting between ground sorts is lifted to pair-sorts by the following $\ReverseHoare$ subsorting rule:
\[
\surfaceTyping{P-I-PAIR}{  \quad
\rsubwrhof{\rtype_1}{\rtype'_1}
\quad
\rsubwrhof{\rtype_1}{\rtype'_1}
}{ \rsubwrhof{\mkpairType{\rtype_1}{\rtype_2}}{\mkpairType{\rtype'_1}{\rtype'_2}} }
\]
which (together with the $\ReverseHoare$  subsorting for ground sorts) describes a deterministic algorithm.

%

\subsection{Stabilising Diffusion predicate and Properties}\label{sec-properties-extended}

The stabilising diffusion predicate (Definition~\ref{def:StabilisingProgression}) and the stabilising diffusion condition (Section~\ref{sec-program-with-valid-assumptions}) are parametric in the set of types modeled by the calculus. 

\begin{exa}\label{exa:stabilising-diffusion-extended}
The library in Figure~\ref{f:code-extended} satisfies the stabilising diffusion condition. In particular, the following predicates hold:
\begin{itemize}
\item
$\stabilisingOpOf{\spSumOrNAME}{\mkpairType{\numType}{\boolType}(\mkpairType{\numType}{\boolType},\mkpairType{\pnumType}{\boolType})}$, and
\item
$\stabilisingOpOf{\spAddToFstNAME}{\mkpairType{\numType}{\numType}(\mkpairType{\numType}{\numType},\pnumType)}$. 
\end{itemize}
\end{exa}

The statements of device computation type preservation (Theorem~\ref{the-DeviceTypePreservation}), device computation termination (Theorem~\ref{the-DeviceTermination}),
and network self-stabilisation  for programs with valid sort and stabilising assumptions (Theorem~\ref{the-network-self--stabilization})
are parametric in the set of types modeled by the calculus. The proofs  (see Appendix~\ref{app-proof-typeSounness} and Appendix~\ref{app-proof-selfStabilisation}) are indeed given for the calculus with pairs---the cases for pair construction and deconstruction are straightforward by induction. 
The rest of the paper considers the calculus with pairs---in particular, the rules for checking sort and stabilising assumptions for diffusions are able to check the examples presented in Section~\ref{sec-examples-extended}.

\section{On Checking the  Stabilising-Diffusion Condition}\label{sec-stabilisation-checking}

The rest of the paper is devoted to illustrate a type-based analysis for checking the stabilising-diffusion condition (cf.\  Section~\ref{sec-program-with-valid-assumptions}) for the extension  of the calculus with pairs introduced in Section~\ref{sec-calculus-extended}. 

\subsection{A Type-based approach for checking the  Stabilising-Diffusion Condition}\label{sec-type-basedstabilisation-checking}

Recall that the stabilising-diffusion condition consists of two parts:
\begin{itemize}
\item
validity of the sort-signature assumptions (Condition (1) of Section~\ref{sec-program-with-valid-assumptions}; and
\item
validity of the  stabilising assumptions (Condition (2) of Section~\ref{sec-program-with-valid-assumptions}).
\end{itemize}
In order to check  the stabilising-diffusion condition we assume that each program $\PROGRAM$ comes with:
\begin{itemize}
\item
a non-empty set of sort-signature assumptions $\rsignaturesOf{\fname}$ for each function $\fname$; and 
\item
a possibly empty set of  stabilising sort-signature assumptions $\stabilisingSignaturesOf{\fname}$ for  each diffusion $\fname$.
\end{itemize}
The assumptions $\rsignaturesOf{\oname}$ and $\stabilisingSignaturesOf{\oname}$ for the built-in functions $\oname$ are considered valid---they should come with the definition of the language.  Instead, the validity of the assumptions $\rsignaturesOf{\dname}$ and $\stabilisingSignaturesOf{\dname}$ for the user-defined functions $\dname$ must be checked---these assumptions could be either (possibly partially) provided by the user or automatically inferred.\footnote{The naive inference approach, that is: inferring $\rsignaturesOf{\dname}$ by checking all the  possible refinements of the type-signature of $\dname$ is linear in the number of elements of $\refsigof{\signatureOf{\dname}}$. Some optimizations are possible. We do not address this issue in the paper.}

\subsubsection{On checking Condition (1) of Section~\ref{sec-program-with-valid-assumptions}}

Section~\ref{sec-ensuringSelfStabilisation} introduces a sort-checking system that  checks  Condition (1) of Section~\ref{sec-program-with-valid-assumptions}. Namely,  given a program (or library) $\PROGRAM$, it checks that: $(i)$ each user-defined function  $\dname$ in $\PROGRAM$ has all the sort signatures in $\rsignaturesOf{\dname}$ and, if $\dname$ is a diffusion, it has also the sort signatures in  
$\stabilisingSignaturesOf{\dname}$;  
 and $(ii)$ 
every  diffusion-expressions $\spreadThree{\e_1}{\fname}{\e_2,\myldots,\e_n}$ occurring in   $\PROGRAM$ is sort-checked by considering for $\fname$ only the sort-signatures in $\stabilisingSignaturesOf{\fname}$.
The soundness of the sort-checking system (shown in Section~\ref{sec-properties-sorting}) guarantees that, if the check is passed, then 
for every  diffusion-expressions $\spreadThree{\e_1}{\fname}{\e_2,\myldots,\e_n}$ occurring in the  $\PROGRAM$
there is a sort-signature $\rtype(\rtype_1,\ldots,\rtype_n)\in\stabilisingSignaturesOf{\fname}$ such that the evaluation of the subexpression $\e_i$ yields a value $\anyvalue_i\in\semOf{\rtype_i}$ ($1\le i\le n$). I.e.,
Condition (1) of Section~\ref{sec-program-with-valid-assumptions} holds.

\subsubsection{On checking Condition (2) of Section~\ref{sec-program-with-valid-assumptions}}\label{sec-On checking Condition (2)}

Note that, if  Condition (1) of Section~\ref{sec-program-with-valid-assumptions} has been checked by the sort-checking system of Section~\ref{sec-ensuringSelfStabilisation}, then in order to check Condition (2) of Section~\ref{sec-program-with-valid-assumptions} holds it is enough to check 
 that for each user-defined diffusion $\dname$,
 each sort-signature $\rtype(\rtype_1,\ldots,\rtype_n)\in\stabilisingSignaturesOf{\dname}$ is stabilising  for  $\dname$.

In order to check that  for each user-defined diffusion $\dname$ of signature $\anytype_1(\anytype_1,\ldots,\anytype_n)$ each sort-signature $\rtype(\rtype_1,\ldots,\rtype_n)\in\stabilisingSignaturesOf{\dname}$ is stabilising for $\dname$, we introduce additional requirements. Namely, we require that for every user-defined diffusion $\dname$ such that  $\stabilisingSignaturesOf{\dname}\not=\emptyset$:
\begin{enumerate}
\item
there exists a sort-signature $\rtype(\cdots)\in\stabilisingSignaturesOf{\dname}$ such that $\rtype'(\cdots)\in\stabilisingSignaturesOf{\dname}$ implies $\rtype\subwrhLe\rtype'$; and
\item
if $\anytype_1$ is not ground (i.e., if it is a pair type), then
the user-defined diffusion $\dname$ is of the form
\begin{equation}\label{eq:ctd}
\defK\;\anytype_1 \;  \dname(\anytype_1 \; \xname_1,\ldots,\anytype_n \; \xname_n) \; \isK \; \canonicalTopFor{\topValOf{\rtype}}(\fname(\xname_1,\ldots,\xname_n))
\end{equation}
where 
\begin{itemize}
\item
 $\canonicalTopFor{\topValOf{\rtype}}$ (defined in Section~\ref{sec-preliminary-iProgression}) is a
 pure function of sort-signature $\rtype(\rtype)$,  
\item
$\fname$ is a diffusion, and
\item
if $\fname$ is user-defined then  $\stabilisingSignaturesOf{\fname}=\emptyset$.
\end{itemize}
\end{enumerate}
\noindent
Note that the above additional requirements can be checked automatically.

In the rest of this section we first (in Section~\ref{sec-preliminary-iProgression}) introduce  some auxiliary definitions (including, for each sort $\rtype$,  the definition of the pure function $\canonicalTopFor{\topValOf{\rtype}}$); then (in Section~\ref{sec-iProgression}) we introduce the notion of \emph{$\PpAnn$-prestabilising diffusion with respect to a progressive sort-signature $\rtype(\rtype_1,\ldots,\rtype_n)$} and show that 
\begin{itemize}
\item
$\rtype$ ground implies that: if $\rtype(\rtype_1,\ldots,\rtype_n)$ is $\PpAnn$-prestabilising for the user-defined diffusion  $\dname$ then 
$\rtype(\rtype_1,\ldots,\rtype_n)$ is stabilising for $\dname$;
\item
if $\rtype(\rtype_1,\ldots,\rtype_n)$ is $\PpAnn$-prestabilising for the diffusion $\fname$ then 
$\rtype(\rtype_1,\ldots,\rtype_n)$ is stabilising for the user-defined diffusion $\dname$  displayed in Equation~\ref{eq:ctd}; and
\end{itemize}
finally (in Section~\ref{sec-annotatedSorts}) we introduce  \emph{annotated sort-signatures} and \emph{annotated sorts} as convenient notations to be used in writing type-based  rules for checking  $\PpAnn$-prestabilisation.

Section~\ref{sec-checkingStabilisingDiffusion} introduces an annotated sort checking system that checks that 
for each user-defined diffusion $\dname$ of signature $\anytype_1(\anytype_1,\ldots,\anytype_n)$ and for each sort-signature $\rtype(\rtype_1,\ldots,\rtype_n)\in\stabilisingSignaturesOf{\dname}$:
\begin{itemize}
\item
$\rtype$ ground implies that $\rtype(\rtype_1,\ldots,\rtype_n)$ is $\PpAnn$-prestabilising for  $\dname$; and
\item
$\rtype$ not ground implies that $\rtype(\rtype_1,\ldots,\rtype_n)$ is $\PpAnn$-prestabilising for the diffusion $\fname$ occurring in Equation~\ref{eq:ctd}.
\end{itemize}
The soundness of the annotated sort checking system (shown in Section~\ref{sec-properties-annotating}) guarantees that, if the check is passed, then Condition (2) of Section~\ref{sec-program-with-valid-assumptions} holds.


\subsection{Auxiliary definitions}\label{sec-preliminary-iProgression}

For any type $\anytype$ the \emph{leftmost-as-key preorder} $\keywfLeOf{\anytype}$ is the preorder that weakens the order $\wfLeOf{\anytype}$ by considering each pair as a record where the leftmost ground element is the key. It 
is defined by:
\begin{itemize}
\item
$\anyvalue\keywfLeOf{\anytype}\anyvalue'\quad$ if $\quad\anyvalue\wfLeOf{\anytype}\anyvalue'$, where $\anytype$ is a ground type; and
\item
$\mkpair{\anyvalue_1}{\anyvalue_2}\keywfLeOf{\mkpairType{\anytype_1}{\anytype_2}}\mkpair{\anyvalue'_1}{\anyvalue'_2}\quad$ if $\quad\anyvalue_1\keywfLeOf{\anytype_1}\anyvalue'_1$.
\end{itemize}
Note that that the leftmost-as-key preorder is total, i.e., for every $\anyvalue,\anyvalue'\in\semOf{\anytype}$ we have that either $\anyvalue\keywfLeOf{\anytype}\anyvalue'$ or $\anyvalue'\keywfLeOf{\anytype}\anyvalue$ holds. We write  $\anyvalue\keywfEqOf{\anytype}\anyvalue'$ to mean that both $\anyvalue\keywfLeOf{\anytype}\anyvalue'$ and $\anyvalue'\keywfLeOf{\anytype}\anyvalue$ hold. Of course $\anyvalue\keywfEqOf{\anytype}\anyvalue'$ does not imply $\anyvalue\wfEqOf{\anytype}\anyvalue'$. Note also that
$\anyvalue\keywfLtOf{\anytype}\anyvalue'$ implies $\anyvalue\wfLtOf{\anytype}\anyvalue'$.

For every sort $\rtype$ of $\anytype$, we wrote $\keywfLeOf{\rtype}$ to denote the restriction of $\keywfLeOf{\anytype}$ to $\semOf{\rtype}$.
According to the previous definition, we define the $\leftKey$ of a sort $\rtype$ as the the leftmost ground sort occurring in $\rtype$,
and  the $\leftKey$ of a value $\anyvalue$ as the the leftmost ground value occurring in $\anyvalue$. The   $\leftKey$ mappings can be inductively defined as follows:
\begin{itemize}
\item
$\leftKeyOf{\rtype}=\rtype$, if $\rtype$ is a ground sort
\item
$\leftKeyOf{\rtype}=\leftKeyOf{\rtype_1}$, if $\rtype=\mkpairType{\rtype_1}{\cdots}$
\end{itemize}
and
\begin{itemize}
\item
$\leftKeyOf{\anyvalue}=\anyvalue$, if $\anyvalue$ is a ground value
\item
$\leftKeyOf{\anyvalue}=\leftKeyOf{\anyvalue_1}$, if $\anyvalue=\mkpair{\anyvalue_1}{\cdots}$.
\end{itemize}
Note that, for every $\anyvalue$ and $\anyvalue'$ of sort $\rtype$ it holds that:
\begin{equation}\label{eq:leftKeyProperty}
\anyvalue\keywfLeOf{\rtype}\anyvalue' \qquad \mbox{ if and only if } \qquad \leftKeyOf{\anyvalue}\wfLeOf{\leftKeyOf{\rtype}}\leftKeyOf{\anyvalue'}.
\end{equation}

For every sort $\rtype$ of type $\anytype$ we write $\canonicalTopFor{\topValOf{\rtype}}$ to denote the pure  function (which satisfies the sort-signature $\rtype(\rtype)$), that maps the elements  $\anyvalue$ such that $\anyvalue\keywfEqOf{\rtype}\topValOf{\rtype}$ to $\topValOf{\rtype}$ (i.e., that propagates the top value from the key of a pair value to all the other components of the value), and is the identity otherwise. Note that the function $\canonicalTopFor{\topValOf{\rtype}}$ can be inductively defined as:  $\defK\;\anytype \;  \canonicalTopFor{\topValOf{\rtype}}(\anytype \; \xname) \; \isK \; \e$, where
\[
\e = \left\{\begin{array}{ll}
 \xname & \mbox{if $\rtype$ is ground}
\\
\canonicalTop{\topValOf{\rtype'}}{\fstK \; \xname} = \topValOf{\rtype'}) \; ? \; \topValOf{\rtype} \; : \xname & \mbox{if } \topValOf{\rtype}=\mkpairType{\topValOf{\rtype'}}{\cdots} 
\end{array}\right.
\]

For every diffusion $\fname$ with signature$\anytype_1(\anytype_1,\ldots,\anytype_n)$ and sort signature $\rtype(\rtype_1,\ldots,\rtype_n)$ $(n\ge 1)$ we write $\ctdOfFor{\fname}{\topValOf{\rtype}}$ to denote the composition of $\fname$ and  $ \canonicalTopFor{\topValOf{\rtype}}$, which is the diffusion (which satisfies the sort-signature $\rtype(\rtype_1,\ldots,\rtype_n)$) defined as follows:
\[
\defK\;\anytype_1 \;  \ctdOfFor{\fname}{\topValOf{\rtype}}(\anytype_1 \; \xname_1,\ldots,\anytype_n \; \xname_n) \; \isK \; \canonicalTopFor{\topValOf{\rtype}}(\fname(\xname_1,\ldots,\xname_n))
\]

\subsection{\texorpdfstring{$\PpAnn$}{!}-Prestabilising \DiffusionsText\ and \texorpdfstring{$\WpAnn$}{?}-Prestabilising \DiffusionsText}\label{sec-iProgression}

 A $\pAnn$-prestabilising \diffusionText{} is a \diffusionText{} whose progressiveness behaviour  is expressed by the  annotation $\pAnn$ that ranges over $\PpAnn$ (for \emph{\certainlyprestabilisingText}) and $\WpAnn$ (for \emph{\possiblyprestabilisingText}), as illustrated by the following definition.

\begin{defi}[$\pAnn$-prestabilising \diffusionText]\label{def:MuPiProgression}
A \diffusionText\ $\fname$ 
 is  $\pAnn$-\emph{prestabilising with respect to the \progressiveSortSigNAME\ sort signature} $\rtype(\rtype_1\overline{\rtype})\in\refsigof{\fname}$ (notation $\pOf{\fname}{\rtype(\rtype_1\overline{\rtype})}$)  if for any $\overline{\anyvalue}\in\semOf{\overline{\rtype}}$:
\begin{enumerate}
    \item
     if $\pAnn=\PpAnn$ then
 \begin{itemize}
\item
 $\anyvalue\keywfLtOf{\rtype_1}\anyvalue'$ and $\opApply{\opFunOf{\fname}}{\anyvalue,\overline{\anyvalue}}=\anyvalue''\not\keywfEqOf{\rtype_1}\topValOf{\rtype_1}$
 imply
    $\anyvalue''\keywfLtOf{\rtype_1}\opApply{\opFunOf{\fname}}{\anyvalue',\overline{\anyvalue}}$;
\item
for all
        $\anyvalue\in\semOf{\rtype_1}-\{\topValOf{\rtype_1}\}$,
        $\anyvalue\keywfLtOf{\rtype_1}\opApply{\opFunOf{\fname}}{\anyvalue,\overline{\anyvalue}}$;
\end{itemize}
    \item
     if $\pAnn\in\{\PpAnn,\WpAnn\}$ then
\begin{itemize}
\item
 $\anyvalue\keywfLeOf{\rtype_1}\anyvalue'$
    implies
    $\opApply{\opFunOf{\fname}}{\anyvalue,\overline{\anyvalue}}\keywfLeOf{\rtype_1}\opApply{\opFunOf{\fname}}{\anyvalue',\overline{\anyvalue}}$;
\item 
for all
        $\anyvalue\in\semOf{\rtype_1}$,
     $\anyvalue\keywfLeOf{\rtype_1}\opApply{\opFunOf{\fname}}{\anyvalue,\overline{\anyvalue}}$.
 \end{itemize}
\end{enumerate}
\end{defi}
\noindent We say that the sort-signature $\rtype(\overline{\rtype})$ is  $\pAnn$-\emph{prestabilising} \emph{for} $\fname$ to mean that $\pOf{\fname}{\rtype(\overline{\rtype})}$ holds, and write  $\psigof{\fname}$  to denote  set of  the $\pAnn$-prestabilising sort-signatures for $\fname$.

Recall the definition of $\ctdOfFor{\fname}{\topValOf{\rtype}}$ given at the end of Section~\ref{sec-preliminary-iProgression}.
The following proposition guarantees that if $\rtype(\overline{\rtype})$ is $\PpAnn$-prestabilising for the diffusion $\fname$ then:
\begin{itemize}
\item
  $\rtype$ ground implies that $\rtype(\overline{\rtype})$ is stabilising for $\fname$; and 
\item
 $\rtype(\overline{\rtype})$ is stabilising for the user-defined diffusion $\dname$  displayed in Equation~\ref{eq:ctd} of Section~\ref{sec-preliminary-iProgression}.
\end{itemize}

\begin{prop}
\begin{enumerate}
\item
if  $\rtype$ is ground then: $\PpOf{\fname}{\rtype(\overline{\rtype})}$ implies $\stabilisingOpOf{\fname}{\rtype(\overline{\rtype})}$, i.e,
\[
\Ppsigof{\fname}\subseteq\stabsigof{\fname}.
\]
\item
$\PpOf{\fname}{\rtype(\overline{\rtype})}$ implies $\stabilisingOpOf{\ctdOfFor{\fname}{\topValOf{\rtype}}}{\rtype(\overline{\rtype})}$, i.e,
\[
\Ppsigof{\fname}\subseteq\stabsigof{\ctdOfFor{\fname}{\rtype}}.
\]
\end{enumerate}
\end{prop}
\begin{proof}
Straightforward from Definition~\ref{def:MuPiProgression},  Definition~\ref{def:StabilisingProgression}, and the definition of $\ctdOfFor{\fname}{\rtype}$.
\end{proof}

\begin{exa}
 Consider the libraries in Figure~\ref{f:code} and~\ref{f:code-extended}.
The following predicates hold (cf.\ Examples~\ref{exa:stabilising-diffusion} and~\ref{exa:stabilising-diffusion-extended}):
\begin{itemize}
\item
$\PpOf{+}{\numType(\numType,\pnumType)}$ and $\WpOf{+}{\numType(\pnumType,\zpnumType)}$
\item
$\PpOf{\texttt{max}}{\pnumType(\znnumType,\pnumType)}$ and  $\WpOf{\texttt{max}}{\numType(\numType,\numType)}$, where \texttt{max} is the  the binary maximum function:
\begin{center}
\texttt{def max(\numNAME\ x, \numNAME\ y)  is  x < y ? y : x }
\end{center}
\item
$\WpOf{\texttt{id}}{\numType(\numType)}$
\item
$\WpOf{\texttt{restrict}}{\numType(\numType,\boolType)}$
\item
$\PpOf{\texttt{restrictSum}}{\numType(\numType,\pnumType,\boolType)}$
\item
$\PpOf{\texttt{sum\_or}}{\mkpairType{\numType}{\boolType}(\mkpairType{\numType}{\boolType},\mkpairType{\pnumType}{\boolType})}$
\item
$\PpOf{\texttt{add\_to\_1st}}{\mkpairType{\numType}{\numType}(\mkpairType{\numType}{\numType},\pnumType)}$
\end{itemize}
\end{exa}

\subsection{Annotated Sort-Signatures and Annotated Sorts}\label{sec-annotatedSorts}

 In order to be able to write type-based  rules for checking  $\pAnn$-stabilisation we introduce, as convenient notations,  \emph{annotated sort-signatures} and \emph{annotated sorts}.

An \emph{annotated sort-signature} $\rtype(\overline{\rtype})\pAnnSB$  is a \progressiveSortSigNAME\ sort-signature (cf.\ Definition~\ref{def:ProgressiveSortSignature}) with a  $\pAnn$ annotation. It provides a convenient notation to express the fact that the predicate  $\pOf{\fname}{\rtype(\overline{\rtype})}$ holds. Namely, we say that a diffusion $\fname$ \emph{has} (or \emph{satisfies}) the annotated sort-signature $\rtype(\overline{\rtype})\pAnnSB$ to mean that the predicate  $\pOf{\fname}{\rtype(\overline{\rtype})}$ holds. 
We write $\pannsigof{\fname}$ to denote the set of the annotated sort-signatures with annotation $\pAnn$ that are satisfied by the diffusion $\fname$, and we write $\annsigof{\fname}$ to denote $\Ppannsigof{\fname}\cup\Wpannsigof{\fname}$.

The \emph{support} of an annotated sort signature $\rtype(\overline{\rtype})\pAnnSB$ is the progressive sort signature $\rtype(\overline{\rtype})$. 
Given an annotated sort-signature $\rtype(\overline{\rtype})\pAnnSB$ we write $\annErasure{\rtype(\overline{\rtype})\pAnnSB}$ to denote its support.
 Note that, according to the above definitions, the mapping $\pannsigof{\cdot}$ provides the same information of the mapping  $\psigof{\cdot}$   introduced in Section~\ref{sec-iProgression}, i.e., $\psigof{\fname}=\annErasure{\pannsigof{\fname}}$.

Given a diffusion type-signature $\anytype(\anytype\overline{\anytype})$ (cf.\ Definition~\ref{def:Progression}) we write $\annsigof{\anytype(\anytype\overline{\anytype})}$ to denote the (set of) annotated sort-signatures  that refine it, i.e., the set
\[
\{    \rtype(\rtype'\overline{\rtype})\pAnnSB   \;\vert\; \rtype(\rtype'\overline{\rtype})\in\progrefsigof{\anytype(\anytype\overline{\anytype})}
\mbox{ and } \pAnn\in\{\PpAnn,\WpAnn\} \}.
\]

Recall the stabilising subsigning partial order between progressive signatures introduced at the end of Section~\ref{sec-Progression}. The following order between progressiveness annotations, that we call \emph{subannotating} relation and denote by $\le$,  

\begin{tikzpicture}
\node (WpAnn) at (0:0) {$\WpAnn$};
\node (PpAnn) at (270:1) {$\PpAnn$}  edge [->] (WpAnn);
\end{tikzpicture}

\noindent induces the following partial order between annotated sort-signatures, that we call \emph{annotated subsigning}:
\[
\surfaceTyping{I-A-SIG}{  \quad
\rsubstabof{\rtype(\rtype_1\overline{\rtype})}{\rtype'(\rtype'_1\overline{\rtype}')}
\quad
\rsubof{\pAnn}{\pAnnScript{'}}
}{ \rsubof{\rtype(\rtype_1\overline{\rtype})\pAnnSB}{\rtype'(\rtype'_1\overline{\rtype}')\pAnnScriptSB{'}} }
\]
%
%
The following proposition shows that annotated subsigning captures the natural implication relation between  $\pAnn$-prestabilisation properties.
\begin{prop}[Sounness of annotated subsigning]\label{prop:Sounness of annotated subsigning}
If the diffusion $\fname$ satisfies the the annotated sort-signature $\rtype(\overline{\rtype})\pAnnSB$ and $\rsubof{\rtype(\overline{\rtype})\pAnnSB}{\rtype'(\overline{\rtype}')\pAnnScriptSB{'}}$, then $\fname$ satisfies  $\rtype'(\overline{\rtype'})\pAnnScriptSB{'}$.
\end{prop}
\begin{proof}
Straightforward from Definition~\ref{def:MuPiProgression} and the definition of annotated subsigning (Rule \ruleNameSize{[I-A-SIG]} above).
\end{proof}


We say that an value $\anyvalue$ has (or satisfies) \emph{annotated sort} $\rtype'_1\pAnnSB$  to mean that $\anyvalue$ has sort $\rtype'_1$, and
\begin{itemize}
\item
the application 
of any \diffusionText\ with annotated sort signature $\rtype(\rtype_1,\myldots,\rtype_n)\pAnnSB$ such that both $\rsubwrhof{\leftKeyOf{\rtype'_1}}{\leftKeyOf{\rtype_1}}$
and $\rsubof{\rtype'_1}{\rtype_1}$ hold
\item
to $\anyvalue$ and to values $\anyvalue_2,\ldots,\anyvalue_n$ of sorts 
$\rtype'_2,\myldots,\rtype'_n$ (respectively) such that $\rsubof{\rtype'_2}{\rtype_2},\myldots,\rsubof{\rtype'_n}{\rtype_n}$,  
\end{itemize}
produces a result of annotated sort $\rtype\pAnnScriptSB{''}$, where $\pAnn''=\pAnn(\pAnn')$.
According to this definition, the following property holds.
\begin{prop}[Annotated sorts for ground values]\label{prop:AnnotatedSortsForGroundValues}
For every sort $\rtype\in\refof{\anytype}$ the maximum element $\anyvalue$ of $\semOf{\rtype}$ w.r.t.\ $\wfLeOf{\rtype}$
has both  sort $\rtype\PpAnnSB$ and sort  $\rtype\WpAnnSB$.
\end{prop}
\begin{proof}
Straightforward from Definition~\ref{def:MuPiProgression}.
\end{proof}
\noindent
The following partial order between annotated sorts, that we call \emph{annotated subsorting}, models the natural implication between the properties they represent:
\[
\surfaceTyping{I-A-SORT}{  \quad
\rsubwrhof{\leftKeyOf{\rtype}}{\leftKeyOf{\rtype'}}
\quad
\rsubof{\rtype}{\rtype'}
\quad
\rsubof{\pAnn}{\pAnnScript{'}}
}{ \rsubof{\rtype\pAnnSB}{\rtype'\pAnnScriptSB{'}} }
\]
The \emph{support} of an annotated sort 
$\rtype\pAnnSB$ is the sort $\rtype$. 
Given an annotated sort $\artype$ we write $\annErasure{\artype}$ to denote its support.
\section{Checking Sort-Signature Assumptions for User-Defined Functions}\label{sec-ensuringSelfStabilisation}

In this section we present a decidable sort-checking system that guarantees that if a program (or library)  $\PROGRAM$ is \emph{well-sorted} (i.e., it can be successfully checked by the rules of the system) then  Condition (1) of Section~\ref{sec-program-with-valid-assumptions} holds. 
\footnote{Note that, in this section, we do not use the additional requirements (1) and (2) introduced at the beginning of Section~\ref{sec-On checking Condition (2)}. This generality might become useful since Condition (2) of Section~\ref{sec-program-with-valid-assumptions} might be checked by using a technique different from the one presented in Section~\ref{sec-stabilisation-checking}.}

We first introduce some auxiliary definitions  (in Section~\ref{sec-ref-typing-auxiliary}); then we consider the issue of associating sorts to values and sensors (in Section~\ref{sec-ref-for-values}), sort-signatures to functions (in Section~\ref{sec-ref-for-functions}) and stabilising sort-signatures to diffusions (in Section~\ref{sec-ref-for-diffusions}); and finally we present 
a decidable sort system for checking the correctness of sort-signature declarations for user-defined functions  (in Section~\ref{sec-ref-typing})
and show that it is sound (in Section~\ref{sec-properties-sorting}).

\subsection{Auxiliary Definitions}\label{sec-ref-typing-auxiliary}

The auxiliary definitions presented in this section will be used (in Section~\ref{sec-ref-typing}) to formulate the sort-checking rules for function applications and spreading expressions.

Recall that for every type $\anytype$ and type-signature $\anytype(\overline{\anytype})$ both  $(\refof{\anytype},\rsub)$ and  $(\refsigof{\anytype(\overline{\anytype})},\rsub)$ are partial orders (cf.\ Section~\ref{sec-ref-types}).
Sort checking an expression $\e$ of type $\anytype$  amounts to compute an \emph{abstract interpretation}~\cite{Cousot:ACM-CS-1996} over these partial orders.

Given a partial order $(P,\le)$ and a subset $Q$ of $P$ we say that:
\begin{itemize}
\item
an element $q_0\in Q$ is \emph{minimal in} $Q$ to mean that: if   $q\in Q$ and $q\le q_0$ then 
 $q=q_0$---the set of the minimal elements of $Q$ is denoted by  $\minimalsOf{Q}$;
\item
$Q$ is \emph{minimised} to mean that every $q\in Q$ is minimal in $Q$, i.e.,  that $Q=\minimalsOf{Q}$.
\end{itemize} 
%

Given a set of sort-signatures $Q\subseteq\refsigof{\anytype(\overline{\anytype})}$ and some sorts $\overline{\rtype}'\in\refsigof{\overline{\anytype}}$, consider the (possibly empty) subset of $Q$ defined as follows:
\[
Q(\overline{\rtype}') = \{ \rtype(\overline{\rtype})\in Q \;\vert\; \rsubof{\overline{\rtype}'}{\overline{\rtype}} \}. 
\]
We say that $Q$
 is \emph{\deterministicText}, notation $\deterministicOf{Q}$, to mean that
for all sorts $\overline{\rtype}'$ there exists a sort-signature $\rtype(\overline{\rtype})\in Q(\overline{\rtype}')$, called the \emph{most specific sort-signature for} $\overline{\rtype}'$ \emph{in} $Q$,  such that:
\begin{center}
for all $\rtype''(\overline{\rtype}'')\in Q(\overline{\rtype}')$ it holds that  $\rsubof{\rtype}{\rtype''}$.
\end{center}
The mapping  $\mostSpecificOf{Q}{\overline{\rtype}'}$,
given a deterministic set of sort-signatures $Q$ $\subseteq$ $\refsigof{\anytype(\overline{\anytype})}$  and some sorts $\overline{\rtype}'\in\refsigof{\overline{\anytype}}$, returns the \emph{most specific sort-signature for} $\overline{\rtype}'$ \emph{in} $Q$ if $Q(\overline{\rtype}')$ is not empty, and is undefined otherwise.

\subsection{Sorts for Values and Sensors}\label{sec-ref-for-values}

We assume  a mapping $\rtypeof{\cdot}$ that associates:
\begin{itemize}
\item
to each ground value $\groundvalue$  the minimum (w.r.t.\  $\rsub$) sort $\rtypeof{\groundvalue}$ in $\refof{\groundvalue}$,  and 
\item
to each sensor $\snsname$ the minimum (w.r.t.\  $\rsub$) sort  $\rtypeof{\snsname}$ in $\refof{\typeof{\snsname}}$ such that
$\groundvalue\in\semOf{\rtypeof{\snsname}}$ for every ground value $\groundvalue$ that may be returned by  $\snsname$.
\end{itemize}

\begin{exa}\label{exa:refinement-valuesANDsensorsGROUND}
 Figure~\ref{fig:refinement-valuesANDsensorsGROUND} illustrates the sorts for the ground values and sensors used in the examples introduced throughout the paper.
\end{exa}

\subsection{Sort-signatures for Functions}\label{sec-ref-for-functions}

We assume a mapping $\rsignaturesOf{\cdot}$ that associates  to each built-in function $\oname$ a  set of sort-signatures $\rsignaturesOf{\oname}\subseteq\refsigof{\oname}$ such that the following conditions are satisfied:
\begin{itemize}
\item
$\rsignaturesOf{\oname}$ is non-empty, minimised and deterministic, and
\item
$\rsignaturesOf{\oname}$ represents all the sort-signatures satisfied by $\fname$, i.e.,
 for each $\rtype'(\overline{\rtype}') \in\refsigof{\oname}$ there exists  $\rtype(\overline{\rtype}) \in\rsignaturesOf{\oname}$ such that $\rtype(\overline{\rtype})\rsub\rtype'(\overline{\rtype}')$ holds.
\end{itemize}
Note that the first of the above two conditions on the mapping $\rsignaturesOf{\oname}$ can be checked automatically.
%

\begin{exa}\label{exa:refinement-functions}
Figure~\ref{fig:refinement-instanceBUILT-IN} illustrates the sort-signatures for built-in functions used in the examples introduced throughout the paper.
\end{exa}
 \begin{figure}[!t]{
 \framebox[1\textwidth]{
 $\begin{array}{@{\hspace{-2cm}}l}
\textbf{Ground values sort:} 
  \\
\begin{array}{lcl}
\rtypeof{\falseK} & = & \falseType  
\\
\rtypeof{\trueK} & = & \trueType  
\\
\rtypeof{\groundvalue} & = & \nnumType, \quad \mbox{if $\;\typeof{\groundvalue}=\numType\;$ and $\;\groundvalue<0$}
\\
\rtypeof{\groundvalue} & = & \znumType, \quad \mbox{if  $\;\groundvalue=0$}
\\
\rtypeof{\groundvalue} & = & \pnumType, \quad \mbox{if $\;\typeof{\groundvalue}=\numType\;$ and $\;\groundvalue>0$}
\end{array}
\\
\\
\textbf{Sensors sort:} 
  \\
\begin{array}{lcl}
\rtypeof{\texttt{\#src}} & = & \zpnumType  
\\
\rtypeof{\texttt{\#dist}} & = & \pnumType  
\end{array}
\end{array}$}
} 
\caption{Sorts for ground values and sensors used in the examples  (cf.\ Figure~\ref{fig:source-instance})} \label{fig:refinement-valuesANDsensorsGROUND}
\end{figure}
 \begin{figure}[!t]{
 \framebox[1\textwidth]{
 $\begin{array}{@{\hspace{-2cm}}l}
\begin{array}{lclr}
\rsignaturesOf{\notK} & = & \trueType(\falseType),   & 
\\
                        &  & \falseType(\trueType),   & 
\\
                        &  & \boolType(\boolType)   & 
\\
\rsignaturesOf{\orK} & = & \falseType(\falseType,\falseType),  &
\\
                        &  & \trueType(\trueType,\boolType), &  
\\
                        & & \trueType(\boolType,\trueType),   & 
\\
                        &  & \boolType(\boolType,\boolType)   & 
\\
\rsignaturesOf{-} & = & \nnumType(\pnumType),   & 
\\
                        &  & \znnumType(\zpnumType),   & 
\\
                        &  & \znumType(\znumType),   & 
\\
                        &  & \zpnumType(\znnumType),   & 
\\
                        &  & \pnumType(\nnumType),   & 
\\
                        &  & \numType(\numType)   & 
\\
\rsignaturesOf{+} & = & \nnumType(\nnumType,\znnumType),  &
\\
                        &  & \nnumType(\znnumType,\nnumType), &  
\\
                        & & \znnumType(\znnumType,\znnumType),   &  
\\
                        &  &  \znumType(\znumType,\znumType),   & 
\\
                        &  &  \zpnumType(\zpnumType,\zpnumType),   &   
\\
                        &  &  \pnumType(\zpnumType,\pnumType),   &
\\
                        & & \pnumType(\pnumType,\zpnumType),   &  
\\
                        & & \numType(\numType,\numType)   &  
\\
\rsignaturesOf{=} & = & \falseType(\znnumType,\pnumType),   & 
\\
                        &  &  \falseType(\nnumType,\zpnumType),   & 
\\       
                        &  &  \falseType(\zpnumType,\nnumType),  &
\\
                        &  & \falseType(\pnumType,\znnumType), &  
\\                                   
                        & & \trueType(\znumType,\znumType),   &  
\\
                        &  &  \boolType(\numType,\numType)   & 
\\
\rsignaturesOf{<} & = & \falseType(\zpnumType,\nnumType),   & 
\\
                        &  &  \falseType(\pnumType,\znnumType),   & 
\\                           
                        & & \falseType(\znumType,\znumType),   &  
\\
                        &  &  \trueType(\nnumType,\zpnumType),  &
\\
                        &  & \trueType(\znnumType,\pnumType), &  
\\
                        &  &  \boolType(\numType,\numType)   & 
\end{array}
\end{array}$}
} 
\caption{Sort-signatures for built-in functions used in the examples (cf.\ Figure~\ref{fig:source-instance})} \label{fig:refinement-instanceBUILT-IN}
\end{figure}
%

We also assume that the mapping $\rsignaturesOf{\cdot}$ associates to each user-defined function $\dname$ a  set of  sort-signatures  $\rsignaturesOf{\dname}\subseteq\refsigof{\signatureOf{\dname}}$ such that the following conditions are satisfied:
\begin{itemize}
\item
$\rsignaturesOf{\dname}$ is non-empty, minimised and \deterministicText, and
\item
$\rsignaturesOf{\dname}$  contains at least a sort-signature which is smaller than the type-signature $\signatureOf{\dname}$, i.e.,
there exists  $\rtype(\overline{\rtype}) \in\rsignaturesOf{\dname}$ such that $\rtype(\overline{\rtype})\rsub\signatureOf{\dname}$ holds.
\end{itemize}
Note that the above two conditions on the mapping $\rsignaturesOf{\dname}$ can be checked automatically. 

%
%

\subsection{Stabilising Sort-signatures for Diffusions}\label{sec-ref-for-diffusions}

We assume a mapping $\stabilisingSignaturesOf{\cdot}$ that associates to each built-in  \diffusionText\ $\oname$  a (possibly empty) set of sort-signatures $\stabilisingSignaturesOf{\oname}$ such that the following conditions are satisfied:
\begin{itemize}
\item
$\stabilisingSignaturesOf{\oname}$ is  minimised and deterministic;
\item
$\stabilisingSignaturesOf{\oname}\subseteq\stabsigof{\oname}$; and
\item
$\stabilisingSignaturesOf{\oname}$ represents all the stabilising sort-signatures satisfied by $\oname$, i.e.,
 for each $\rtype'(\overline{\rtype}') \in\stabsigof{\oname}$ there exists  $\rtype(\overline{\rtype}) \in\stabilisingSignaturesOf{\oname}$ such that $\rtype(\overline{\rtype})\rsub\rtype'(\overline{\rtype}')$ holds.
\end{itemize}
Note that the first of the above three conditions on the mapping $\stabilisingSignaturesOf{\oname}$ can be checked automatically.

\begin{exa}\label{exa:stabilising-instanceBUILT-IN}
Figure~\ref{fig:stabilising-instanceBUILT-IN} gives the stabilising sort-signatures for the built-in \diffusionsText\ $\oname$ used in the examples introduced throughout the paper---the built-in \diffusionsText\ without stabilising sort-signatures are omitted.
Note that $\stabilisingSignaturesOf{+}\not\subseteq\rsignaturesOf{+}$, since the stabilising sort-signature  $\numType(\numType,\pnumType)\in\stabilisingSignaturesOf{+}$ is not minimal in  $\rsignaturesOf{+}\cup\{\numType(\numType,\pnumType)\}$ and therefore it cannot be included in  $\rsignaturesOf{+}$---it would break both the requirement that $\rsignaturesOf{+}$  must be minimised and deterministic (condition (1) at the beginning of Section~\ref{sec-On checking Condition (2)}).
\end{exa}
 \begin{figure}[!t]{
 \framebox[1\textwidth]{
 $\begin{array}{@{\hspace{-2cm}}l}
\begin{array}{lclr}
%
\stabilisingSignaturesOf{\orK} & = &   \falseType(\falseType,\falseType),
\\
                                           & & \trueType(\trueType,\boolType),
\\
                                           & & \trueType(\boolType,\trueType)
\\
\stabilisingSignaturesOf{+} & = &   \znumType(\znumType,\znumType),   
\\                         
                       & &        \pnumType(\zpnumType,\pnumType),   
\\
                        &  &  \numType(\numType,\pnumType)   &  {\footnotesize  \mbox{$(\not\in\rsignaturesOf{+})$}}
\end{array}
\end{array}$}
} 
\caption{Stabilising sort-signatures for built-in functions used in the examples (cf.\ Figure~\ref{fig:refinement-instanceBUILT-IN})} \label{fig:stabilising-instanceBUILT-IN}
\end{figure}

We also assume that the mapping $\stabilisingSignaturesOf{\cdot}$  associates to each user-defined \diffusionText\ $\dname$  a (possibly empty)
set of  stabilising sort-signatures  $\stabilisingSignaturesOf{\dname}\subseteq\stabsigof{\dname}$ such that the following conditions are satisfied:
\begin{itemize}
\item
$\stabilisingSignaturesOf{\dname}$ is minimised and \deterministicText, 
 and
\item
$\stabilisingSignaturesOf{\dname}$ is implied by $ \rsignaturesOf{\dname}$, i.e.,
for each $\rtype'(\overline{\rtype}') \in \stabilisingSignaturesOf{\dname}$ there exists  $\rtype(\overline{\rtype}) \in \rsignaturesOf{\dname}$ such that $\rtype(\overline{\rtype})\rsub\rtype'(\overline{\rtype}')$.
\end{itemize}
Note that the above two conditions on the mapping $\stabilisingSignaturesOf{\dname}$ can be checked automatically.

%
\begin{exa}\label{exa:UserDeclaredStabilisingSortSignaturesForExamples}
We assume that for the user-defined functions $\dname$ used in the examples introduced throughout the paper  
\[
\rsignaturesOf{\dname} = \minimalsOf{\{\signatureOf{\dname}\}\cup\stabilisingSignaturesOf{\dname}}.
\]
Figure~\ref{fig:UserDeclaredStabilisingSortSignaturesForExamples} gives minimised deterministic sets of stabilising sort-signatures that allow to successfully check the user-defined \diffusionsText\ $\dname$ used in the examples introduced in the paper---note that both the additional requirements (1) and (2) given at the beginning of Section~\ref{sec-stabilisation-checking} are satisfied.
\end{exa}

 \begin{figure}[!t]{
 \framebox[1\textwidth]{
 $\begin{array}{l}
\begin{array}{lcl@{\!\!\!\!\!\!\!\!\!\!\!\!\!\!\!\!\!}r}
\stabilisingSignaturesOf{\texttt{restrictSum}} & = & \numType(\numType,\pnumType,\boolType),   & 
\\                                      
\stabilisingSignaturesOf{\spSumOrNAME}    & = & \mkpairType{\numType}{\boolType}(\mkpairType{\numType}{\boolType},\mkpairType{\pnumType}{\boolType}) & 
\\
\stabilisingSignaturesOf{\spAddToFstNAME} & = & \mkpairType{\numType}{\numType}(\mkpairType{\numType}{\numType},\pnumType)    & 
\end{array}
\end{array}$}
} 
\caption{Stabilising sort-signatures for the user-defined functions used in the examples}\label{fig:UserDeclaredStabilisingSortSignaturesForExamples}
\end{figure}

\subsection{Sort Checking}\label{sec-ref-typing}

In this section we present a decidable sort checking system for user-defined functions to check whether the sort-signature declarations provided by the mapping  $\rsignaturesOf{\cdot}$ are correct.
The sort-checking rules are given in Figure~\ref{fig:RefinementTyping}.
 \begin{figure}[!t]{
 \framebox[1\textwidth]{
 $\begin{array}{l}
$\quad$\textbf{Expression sort checking:}  \hfill
 \boxed{\surExpTypJud{\RefTypEnv}{\e}{\rtype}}$\quad$
  \\
\begin{array}{c}
\nullsurfaceTyping{S-VAR}{
\surExpTypJud{\RefTypEnv,\xname:\rtype}{\xname}{\rtype}
}
\qquad\qquad
\nullsurfaceTyping{S-SNS}{
\surExpTypJud{\RefTypEnv}{\snsname}{\rtypeof{\snsname}}
}
\qquad\qquad
\nullsurfaceTyping{S-GVAL}{
\surExpTypJud{\RefTypEnv}{\groundvalue}{\rtypeof{\groundvalue}}
}
\skiptransition
\surfaceTyping{S-PAIR}{  \quad
\surExpTypJud{\RefTypEnv}{\e_1}{\rtype_1}
\quad
\surExpTypJud{\RefTypEnv}{\e_2}{\rtype_2}
}{ \surExpTypJud{\RefTypEnv}{\mkpair{\e_1}{\e_2}}{\mkpairType{\rtype_1}{\rtype_2}} }
\quad
\surfaceTyping{S-FST}{ \;\;
\surExpTypJud{\RefTypEnv}{\e}{\mkpairType{\rtype_1}{\rtype_2}} }{
\surExpTypJud{\RefTypEnv}{\fstK\;\e}{\rtype_1} }
\quad
\surfaceTyping{S-SND}{ \;\;
\surExpTypJud{\RefTypEnv}{\e}{\mkpairType{\rtype_1}{\rtype_2}} }{
\surExpTypJud{\RefTypEnv}{\sndK\;\e}{\rtype_2} }
\skiptransition
\surfaceTyping{S-COND}{  \quad
\surExpTypJud{\RefTypEnv}{\e_0}{\boolType}
\quad
\surExpTypJud{\RefTypEnv}{\e_1}{\rtype_1}
\quad
\surExpTypJud{\RefTypEnv}{\e_2}{\rtype_2}
\quad
\rtype=\refSupOf{\rtype_1}{\rtype_2}
}{ \surExpTypJud{\RefTypEnv}{ \condExpr{\e_0}{\e_1}{\e_2}}{\rtype} }
\skiptransition
\surfaceTyping{S-COND-TRUE}{  \quad
\surExpTypJud{\RefTypEnv}{\e_0}{\trueType}
\quad
\surExpTypJud{\RefTypEnv}{\e_1}{\rtype_1}
\quad
\surExpTypJud{\RefTypEnv}{\e_2}{\rtype_2}
}{ \surExpTypJud{\RefTypEnv}{ \condExpr{\e_0}{\e_1}{\e_2}}{\rtype_1} }
\skiptransition
\surfaceTyping{S-COND-FALSE}{  \quad
\surExpTypJud{\RefTypEnv}{\e_0}{\falseType}
\quad
\surExpTypJud{\RefTypEnv}{\e_1}{\rtype_1}
\quad
\surExpTypJud{\RefTypEnv}{\e_2}{\rtype_2}
}{ \surExpTypJud{\RefTypEnv}{ \condExpr{\e_0}{\e_1}{\e_2}}{\rtype_2} }
\skiptransition
\surfaceTyping{S-FUN}{  \quad
\surExpTypJud{\RefTypEnv}{\overline{\e}}{\overline{\rtype}}
\quad 
\rtype(\cdots)=\mostSpecificOf{\rsignaturesOf{\fname}}{\overline{\rtype}} 
}{ \surExpTypJud{\RefTypEnv}{\fname(\overline{\e})}{\rtype} }
\skiptransition
\surfaceTyping{S-SPR}{ \quad
\progressionOpOf{\fname}
\quad
 \surExpTypJud{\RefTypEnv}{\e_0\overline{\e}}{\rtype_0\overline{\rtype}}
\quad
\rtype'(\cdots)=\mostSpecificOf{\stabilisingSignaturesOf{\fname}}{\rtype_0\overline{\rtype}}
\quad
\rtype=\refSupOf{\rtype_0}{\rtype'}
}{
\surExpTypJud{\RefTypEnv}{\spreadThree{\e_0}{\fname}{\overline{\e}}}{\rtype}
 }
\skiptransition
%
%
%
\end{array}
\\
$\quad$\textbf{User-defined function sort checking:} \hfill
  \boxed{\surFunTypJud{}{\FUNCTION}{\overline{\rtype(\overline{\rtype})}}}$\quad$
  \\
\qquad\begin{array}{c}
\surfaceTyping{S-DEF}{ \quad
\mbox{for all $\rtype(\overline{\rtype})\in\rsignaturesOf{\dname}$,}
\quad
\surExpTypJud{\overline{\xname}:\overline{\rtype}}{\e}{\rtype'}
\quad
\rsubof{\rtype'}{\rtype}
}{ \surFunTypJud{}{\defK \; \anytype \;\dname
(\overline{\anytype\;\xname}) \; \isK \; \e}{\rsignaturesOf{\dname}} }
\end{array}
\end{array}$}
} 
\caption{Sort-checking rules for expressions and function definitions} \label{fig:RefinementTyping}
\end{figure}
\emph{Sort environments}, ranged over by $\RefTypEnv$ and written $\overline{\xname}:\overline{\rtype}$, contain sort assumptions for program variables.
The sort-checking judgement for expressions is of the form $\surExpTypJud{\RefTypEnv}{\e}{\rtype}$, to be read: $\e$ has sort $\rtype$ under the sort assumptions
$\RefTypEnv$ for the program variables occurring in $\e$.
Sort checking of variables, sensors, ground values, pair constructions and deconstructions, and conditionals  is similar to type checking. In particular, ground values and sensors are given a sort by construction by exploiting the mapping  $\rtypeof{\cdot}$  introduced in Section~\ref{sec-ref-for-values}, and the sort assigned to a  conditional-expression is:
\begin{itemize}
\item
 the least upper bound $\refSupOf{\rtype_1}{\rtype_2}$ of the sorts assigned to the branches (cf.\ Section~\ref{sec-ref-types})  when the condition has sort $\boolType$;
\item
 the sort assigned to the left branch  when the condition has sort $\trueType$; and
\item
 the sort assigned to the right branch  when the condition has sort $\falseType$.
\end{itemize}
The  sort-checking rule \ruleNameSize{[S-FUN]} for function application exploits the mapping  $\rsignaturesOf{\cdot}$  introduced in Section~\ref{sec-ref-for-functions} and the auxiliary mapping $\mostSpecificOf{\cdot}{\cdot}$  introduced in Section~\ref{sec-ref-typing-auxiliary}. It first infers the sorts $\overline{\rtype}$ for the arguments $\overline{\e}$ of $\fname$, then uses the most specific sort-signature for $\overline{\rtype}$  in $\rsignaturesOf{\fname}$  for assigning to the application $\fname(\overline{\e})$ the minimum sort  $\rtype$ that can be assigned to $\fname(\overline{\e})$ by using for $\fname$ any of the sort signatures in 
$\rsignaturesOf{\fname}$.

In a similar way, the sort-checking rule \ruleNameSize{[S-SPR]} for spreading expressions  first infers the sorts $\rtype_0\overline{\rtype}$ for  $\e_0\overline{\e}$, then retrieves the most specific sort-signature for $\rtype_0\overline{\rtype}$  in $\stabilisingSignaturesOf{\fname}$, $\rtype(\cdots)$, and finally assigns  to the spreading expression the the least upper bound of $\rtype_0$ and $\rtype$. 
%

The sort-checking rule for function definitions (which derives  judgements of the form
$\surFunTypJud{\RefTypEnv}{\FUNCTION}{\overline{\rtype(\overline{\rtype})}}$, where $\overline{\rtype(\overline{\rtype})} = \rtype^{(1)}(\overline{\rtype}^{(1)}),\myldots,\rtype^{(n)}(\overline{\rtype}^{(n)})$ and $n\ge 1$) requires to check the definition $\FUNCTION$ of a user-defined function $\dname$ with respect to all the sort-signatures in $\rsignaturesOf{\dname}$. 

We say that  a program (or library) $\PROGRAM$ is \emph{well sorted}
 to mean that all the user-defined function definitions
in $\PROGRAM$ sort check by using the rules in Figure~\ref{fig:RefinementTyping}.

Since no choice may be done when building a derivation for a given  sort-checking judgment, the sort-checking rules straightforwardly describe a sort-checking algorithm.

\begin{exa} All user-defined functions provided in the examples in Sections~\ref{sec-examples} and~\ref{sec-examples-extended} sort check by assuming the
ground sorts and the subsorting given in Figure~\ref{fig:refinement-subsortingGROUND}, 
  the sorts for the ground values and sensors given in Figure~\ref{fig:refinement-valuesANDsensorsGROUND}, 
the 
sort-signatures for built-in functions given in
Figure~\ref{fig:refinement-instanceBUILT-IN}, the 
stabilising sort-signatures for built-in functions given in Figures~\ref{fig:stabilising-instanceBUILT-IN} and the stabilising sort-signatures for user-defined diffusions given in~\ref{fig:UserDeclaredStabilisingSortSignaturesForExamples}.
\end{exa}



\subsection{Sort Soundness of Device Computation }\label{sec-properties-sorting}

In order to state the correctness of the sort-checking system  presented in Section~\ref{sec-ref-typing}  we introduce
the notion of set of well-sorted values trees for an expression, which generalizes to sorts the notion of set of well-typed values trees for an expression introduced in Section~\ref{sec-properties-typing}.

Given an expression $\e$ such that
$\surExpTypJud{\overline{\xname}:\overline{\rtype}}{\e}{\rtype}$,
the set
$\bsWSVT{\overline{\xname}:\overline{\rtype}}{\e}{\rtype}$
of the \emph{well-sorted} value-trees for $\e$, is inductively
defined as follows:
    $\vtree$ $\in$ $\bsWSVT{\overline{\xname}:\overline{\rtype}}{\e}{\rtype}$
    if there exist
    \begin{itemize}
    \item a sensor mapping $\snsFun$;
    \item well-formed tree environments
        $\overline{\vtree}\in\bsWSVT{\overline{\xname}:\overline{\rtype}}{\e}{\rtype}$;
        and
    \item values $\overline{\anyvalue}$ such that
        $\lengthOf{\overline{\anyvalue}}=\lengthOf{\overline{\xname}}$,
        $\surExpTypJud{\emptyset}{\overline{\anyvalue}}{\overline{\rtype}'}$ and $\overline{\rtype}'\rsub\overline{\rtype}$;
    \end{itemize}
    such that
    $\bsopsem{\snsFun}{\overline{\vtree}}{\applySubstitution{\e}{\substitution{\overline{\xname}}{\overline{\anyvalue}}}}{\vtree}$
    holds---note that this definition is inductive, since the sequence  of evaluation trees $\overline{\vtree}$ may be empty.

%
%
%

The following theorem guarantees that from a properly sorted environment,
    evaluation of a well-sorted expression yields a properly sorted
    result.

\begin{thm}[Device computation sort
preservation]\label{the-DeviceSortPreservation} If
$\surExpTypJud{\overline{\xname}:\overline{\rtype}}{\e}{\rtype}$,
 $\snsFun$ is a sensor mapping,
$\overline{\vtree}\in\bsWSVT{\overline{\xname}:\overline{\rtype}}{\e}{\rtype}$,
$\lengthOf{\overline{\anyvalue}}=\lengthOf{\overline{\xname}}$,
        $\surExpTypJud{\emptyset}{\overline{\anyvalue}}{\overline{\rtype}'}$,
$\overline{\rtype}'\rsub\overline{\rtype}$,
and
$\bsopsem{\snsFun}{\overline{\vtree}}{\applySubstitution{\e}{\substitution{\overline{\xname}}{\overline{\anyvalue}}}}{\vtree}$,
then $\surExpTypJud{\emptyset}{\vrootOf{\vtree}}{\rtype'}$ for some $\rtype'$ such that  $\rtype'\rsub\rtype$.
\end{thm}

\begin{proof}
See Appendix~\ref{app-proof-sortSounness}.
\end{proof}

\begin{rem}[On the relation between type checking and sort checking]\label{rem:TypeVsSortChecking}
A sort system should be such that all the programs (or libraries) accepted by the sort system are accepted by the original type system  and vice-versa (cf.\ the discussion at the beginning of Section~\ref{sec-ref-types}).
However, the sort system considered in this paper has a peculiarity: it checks that every  diffusion-expressions $\spreadThree{\e_1}{\fname}{\e_2,\myldots,\e_n}$ occurring in $\PROGRAM$ is sort-checked by considering for $\fname$ only the sort-signatures in $\stabilisingSignaturesOf{\fname}$---which is assumed to be such that for all $\rtype(\overline{\rtype})\in\stabilisingSignaturesOf{\fname}$ the predicate  $\stabilisingOpOf{\fname}{\rtype(\overline{\rtype})}$ holds. Therefore, some well-typed programs (or libraries)---including all the non-self-stabilising programs  (like the one considered in Example~\ref{exa:not-self-stabilising})---do not sort check. 

The standard relation between the sort system and the type system that it refines holds for programs (or libraries)  that do not contain spreading expressions. I.e., all the programs (or libraries) that do not contain spreading expressions that accepted by the original type system are accepted by the sort system and vice-versa. In particular, whenever all the sorts are trivial (i.e., $\refof{\anytype}=\{\anytype\}$, for every type $\anytype$) we have that, on programs (or libraries)  that do not contain spreading expressions, the sort checking rules (in Figure~\ref{fig:RefinementTyping}) behave exactly as the type-checking rules (in Figure~\ref{fig:SurfaceTyping}).
\end{rem}

\section{Checking Stabilising Assuptions for User-Defined Diffusions}\label{sec-checkingStabilisingDiffusion}

In this section we present a decidable annotated sort checking system that guarantees that if a program $\PROGRAM$ is \emph{well-sorted} (i.e., it can be successfully checked by the rules of the system given in Section~\ref{sec-ensuringSelfStabilisation} and, therefore, satisfies Condition (1) of Section~\ref{sec-program-with-valid-assumptions}) guarantees that, under the  additional requirements (1) and (2) introduced at the beginning of Section~\ref{sec-On checking Condition (2)},  if  $\PROGRAM$ is \emph{well-annotated} (i.e., it can be successfully checked by the rules of the system) then  Condition (2) of Section~\ref{sec-program-with-valid-assumptions} holds, and hence self-stabilisation follows. 

To this aim, under the  additional requirements (1) and (2) introduced at the beginning of Section~\ref{sec-stabilisation-checking}, we  assume
 that each program $\PROGRAM$ comes with
 a mapping $\asignaturesOf{\cdot}$ that associates to each  \diffusionText\ $\fname$   a (possibly empty) set of  annotated sort-signatures $\asignaturesOf{\fname}$  (we write $\psignaturesOf{\pAnn}{\cdot}$ to denote the mapping such that $\psignaturesOf{\pAnn}{\fname}=\{\rtype(\overline{\rtype})\pAnnSB \;\vert\;\rtype(\overline{\rtype})\pAnnSB \in\asignaturesOf{\fname}\}$)
such that the following conditions are satisfied:
\begin{itemize}
\item
for every diffusion $\fname$ of type $\anytype(\cdots)$,
if $\anytype$ is ground, then $\stabilisingSignaturesOf{\fname}=\annErasure{\PpsignaturesOf{\fname}}$; 
\item
for every user-defined diffusion  $\dname$ of the form  displayed in Equation~\ref{eq:ctd} of Section~\ref{sec-preliminary-iProgression},
consider the function $\fname$ occurring in the body of $\dname$: 
\begin{itemize}
\item 
if $\fname$ is built-in function, then $\stabilisingSignaturesOf{\dname}\subseteq\annErasure{\PpsignaturesOf{\fname}}$, and
\item
if $\fname$ is  user-defined, then
$\stabilisingSignaturesOf{\dname}=\annErasure{\PpsignaturesOf{\fname}}$.
\end{itemize}
\end{itemize}
\noindent
Note that the above conditions (which can be checked automatically) imply that, for every user defined function $\dname$, the value of $\stabilisingSignaturesOf{\dname}$ is completely defined by the mapping $\asignaturesOf{\cdot}$---therefore, there is no need to explicitly define the value of $\stabilisingSignaturesOf{\cdot}$ for user-defined diffusions.
The assumptions $\asignaturesOf{\oname}$ for the built-in functions $\oname$ are considered valid---they should come with the definition of the language.  Instead, the validity of the assumptions $\asignaturesOf{\dname}$ for the user-defined functions $\dname$ must be checked---these assumptions could be either (possibly partially) provided by the user or automatically inferred.\footnote{The naive inference approach, that is: inferring $\asignaturesOf{\dname}$ by checking all the sort-signature of $\dname$ is linear in the number of elements of $\rsignaturesOf{\dname}$. Some optimizations are possible. We do not address this issue in the paper.}
Therefore, in order to check that Condition (2) of Section~\ref{sec-program-with-valid-assumptions} holds, it is enough to check  that each the user-defined diffusions $\dname$ of $\PROGRAM$ has all the annotated sort signatures $\asignaturesOf{\dname}$.

We first introduce some auxiliary definitions (in Section~\ref{sec-ann-typing-auxiliary}), then we consider the issue of associating  annotated sort-signatures to diffusions (in Section~\ref{sec-AnnotatedSortSignaturesMappings})
 and the issue of associating annotated sorts to values (in Section~\ref{sec-AnnotatedSortMapping}), and finally we present a decidable  annotated sort checking system for checking the correctness of annotated sort-signature declarations for user-defined diffusions  (in Section~\ref{sec-annotatedSortChecking}) and show its soundness (in Section~\ref{sec-properties-annotating}).

\subsection{Auxiliary definitions}\label{sec-ann-typing-auxiliary}

In this section we adapt the notions of minimal set of sort-signatures, deterministic set of sort signatures and most specific sort-signature (cf.\ Section~\ref{sec-ref-typing-auxiliary})  to annotated sort-signatures. These notions will be used (in Section~\ref{sec-annotatedSortChecking}) to formulate the annotated  sort-checking rules for function applications.  

Given a diffusion signature $\anytype(\anytype\overline{\anytype})$ and a set of annotated sort-signatures $Q\subseteq\annsigof{\anytype(\anytype\overline{\anytype})}$, an annotated sort $\rtype'_1\pAnnScriptSB{'}$ such that $\rtype'_1\in\refsigof{\anytype_1}$ and some sorts $\overline{\rtype}'\in\refsigof{\overline{\anytype}}$, consider the (possibly empty) subset of $Q$ defined as follows:
\[
Q(\rtype'_1\pAnnScriptSB{'}\overline{\rtype}') = \{ \rtype(\rtype_1\overline{\rtype})\SB{\pAnn''}\in Q \;\vert\; \leftKeyOf{\rtype'_1}\subwrhLe\leftKeyOf{\rtype_1},
\;\;
\rtype'_1\rsub\rtype_1
\mbox{ and }
\rsubof{\overline{\rtype}'}{\overline{\rtype}}
 \}. 
\]
We say that $Q$
 is \emph{\deterministicText}, notation $\deterministicOf{Q}$, to mean that
for all  $\rtype'_1\pAnnScriptSB{'}\overline{\rtype}'$ there exists an annotated sort-signature $\rtype(\rtype_1\overline{\rtype})\SB{\pAnn}\in Q(\rtype'_1\pAnnScriptSB{'}\overline{\rtype}')$, called the \emph{most specific annotated sort-signature for} $\overline{\rtype}'$ \emph{in} $Q$,  such that:
\begin{center}
for all $\rtype''(\overline{\rtype}'')\SB{\pAnn''}\in Q(\rtype'_1\pAnnScriptSB{'}\overline{\rtype}')$ it holds that  $\rsubof{\rtype\SB{\pAnn(\pAnn')}}{\rtype''\SB{\pAnn''(\pAnn')}}$.
\end{center}
The mapping  $\mostSpecificOf{Q}{\rtype'_1\pAnnScriptSB{'}\overline{\rtype}'}$,
given a deterministic set of sort-signatures $Q$ $\subseteq$ $\annsigof{\anytype(\anytype\overline{\anytype})}$,   an annotated sort $\rtype'_1\pAnnScriptSB{'}$ such that $\rtype'_1\in\refsigof{\anytype_1}$ and some sorts $\overline{\rtype}'\in\refsigof{\overline{\anytype}}$, returns the \emph{most specific annotated sort-signature for} $\rtype'_1\pAnnScriptSB{'}\overline{\rtype}'$ \emph{in} $Q$ if $Q(\rtype'_1\pAnnScriptSB{'}\overline{\rtype}')$ is not empty, and is undefined otherwise.

\subsection{Annotated Sort-Signatures for Diffusions}\label{sec-AnnotatedSortSignaturesMappings}

The mapping  $\asignaturesOf{\cdot}$ and the conditions that provides its link with the mapping $\stabilisingSignaturesOf{}$ have been illustrated at the beginning of Section~\ref{sec-checkingStabilisingDiffusion}. Here, we illustrate some additional condition that is needed to simplify the formulation of the annotated sort checking rules and to guarantee their soundness.

We assume that for each built-in \diffusionText\ $\oname$  the (possibly empty) set of annotated sort-signatures $\asignaturesOf{\oname}$ is such that the following conditions are satisfied:
\begin{itemize}
\item
$\asignaturesOf{\oname}$ is  minimized and deterministic;
\item
$\asignaturesOf{\oname}\subseteq\annsigof{\oname}$; and
\item
$\asignaturesOf{\oname}$ represents all the annotated sort-signatures satisfied by $\oname$, i.e.,
 for each $\rtype'(\overline{\rtype}')\pAnnScriptSB{'}\in\annsigof{\oname}$ there exists  $\rtype(\overline{\rtype})\pAnnSB \in\asignaturesOf{\oname}$ such that $\rtype(\overline{\rtype})\pAnnSB\rsub\rtype'(\overline{\rtype}')\pAnnScriptSB{'}$ holds.
\end{itemize}
Note that the first of the above three conditions on the mapping $\asignaturesOf{\oname}$ can be checked automatically.

\begin{exa}\label{exa:annotated-instanceBUILT-IN}
Figure~\ref{fig:annotated-instanceBUILT-IN} illustrates the annotated sort-signatures for the built-in \diffusionsText\  used in the examples introduced thought the paper.
\end{exa}
 \begin{figure}[!t]{
 \framebox[1\textwidth]{
 $\begin{array}{@{\hspace{-2cm}}l}
\begin{array}{lclr}
%
\asignaturesOf{\orK} & = & \falseType(\falseType,\falseType)\PpAnnSB,  &
\\
                        &  & \trueType(\trueType,\boolType)\PpAnnSB,  & 
\\
                        &  & \trueType(\boolType,\trueType)\PpAnnSB & 
\\[10pt]
\asignaturesOf{+} & = &  \nnumType(\nnumType,\znumType)\WpAnnSB,   & 
\\
                       &  &  \znnumType(\znnumType,\znumType)\WpAnnSB,   & 
\\
                        & &  \znumType(\znumType,\znumType)\PpAnnSB,   & 
\\
                        &  &  \zpnumType(\zpnumType,\zpnumType)\WpAnnSB,   &   
\\
                        &  &  \pnumType(\zpnumType,\pnumType)\PpAnnSB,   &  
\\
                        & & \pnumType(\pnumType,\zpnumType)\WpAnnSB,   &  
\\
                        &  &  \numType(\numType,\zpnumType)\WpAnnSB,   &  
\\
                        &  &  \numType(\numType,\pnumType)\PpAnnSB   &  
%
\end{array}
\end{array}$}
} 
\caption{Annotated sort signatures for built-in  $\pAnn$-prestabilising \diffusionsText\  used in the examples (cf.\ Figure~\ref{fig:refinement-instanceBUILT-IN})} \label{fig:annotated-instanceBUILT-IN}
\end{figure}
%


We also assume that for each user-defined \diffusionText\ $\dname$  the (possibly empty) set of  annotated sort-signatures $\asignaturesOf{\dname}$  is minimized and \deterministicText\ (note that this condition can be checked automatically).

%
%
%
%


\begin{exa}\label{exa:UserDeclaredAnnotatedSignaturesForExamples}
Figure~\ref{fig:UserDeclaredAnnotatedSignaturesForExamples} gives minimal deterministic sets of annotated sort signatures for the user-defined $\pAnn$-prestabilising \diffusionsText\ that allow to successfully check the user-defined \diffusionsText\  used in the examples introduced thought the paper.
\end{exa}
 \begin{figure}[!t]{
 \framebox[1\textwidth]{
 $\begin{array}{l}
\begin{array}{lclr}
\asignaturesOf{\texttt{restrict}} & = & \numType(\numType,\boolType)\WpAnnSB  &
\\
\asignaturesOf{\texttt{restrictSum}} & = & \numType(\numType,\pnumType,\boolType)\PpAnnSB &
\\
\asignaturesOf{\texttt{sum\_or}}    & = & \mkpairType{\numType}{\boolType}(\mkpairType{\numType}{\boolType},\mkpairType{\pnumType}{\boolType})\PpAnnSB &  
\\
\asignaturesOf{\texttt{add\_to\_1st}} & = & \mkpairType{\numType}{\numType}(\mkpairType{\numType}{\numType},\pnumType)\PpAnnSB   & 
\end{array}
\end{array}$}
} 
\caption{Annotated sort signatures for the user-defined $\pAnn$-prestabilising \diffusionsText\ used in the examples (cf.\ Figure~\ref{fig:UserDeclaredStabilisingSortSignaturesForExamples})}\label{fig:UserDeclaredAnnotatedSignaturesForExamples}
\end{figure}

\subsection{Annotated Sorts for Values}\label{sec-AnnotatedSortMapping}
We assume  a partial mapping $\atypeof{\cdot}$ that for each ground value $\groundvalue$: 
\begin{itemize}
\item
returns the annotated sort $\rtypeof{\groundvalue}\PpAnnSB$, if $\groundvalue$ is the maximum element of $\semOf{\rtypeof{\groundvalue}}$ w.r.t.\ $\wfLeOf{\typeof{\groundvalue}}$, and
\item
 is undefined, otherwise. 
\end{itemize}
Note that Proposition~\ref{prop:AnnotatedSortsForGroundValues} guarantees the soundness of the mapping  $\atypeof{\cdot}$.

\begin{exa}\label{exa:annotated-instanceGROUND}
 Figure~\ref{fig:annotated-instanceGROUND} illustrates the $\PpAnn$-annotated sorts for the ground values 
used in the examples introduced thought the paper.
\end{exa}
 \begin{figure}[!t]{
 \framebox[1\textwidth]{
 $\begin{array}{l}
\begin{array}{lcl}
\atypeof{\falseK} & = & \falseType\PpAnnSB  
\\
\atypeof{\trueK} & = & \trueType\PpAnnSB
\\
\atypeof{0} & = & \znumType\PpAnnSB
\\
\atypeof{\topNumK} & = & \pnumType\PpAnnSB
\end{array}
\end{array}$}
} 
\caption{Annotated sorts for the ground values used in the examples} \label{fig:annotated-instanceGROUND}
\end{figure}
%

\subsection{Annotated Sort Checking for User-Defined \DiffusionsText}\label{sec-annotatedSortChecking}

In this section we present a decidable annotated sort checking system to check whether the  $\pAnn$-annotated sort-signature assumptions for the user-defined diffusions $\dname$  provided by the mapping  $\asignaturesOf{\cdot}$ are correct. The annotated sort-checking rules are given in Figure~\ref{fig:AnnotatedTyping}.

The check a user defined diffusion $\dname$ has the annotated sort signature $(\rtype_1,\myldots,\rtype_n)\rtype\pAnnSB$ can be done by
assuming annotated sort $\rtype_1\WpAnnSB$ for the first formal parameter of $\dname$, assuming sorts $\rtype_2,\myldots,\rtype_n$ for the other formal parameters of $\dname$,  and trying to assign to the body of $\dname$ an annotated sort $\rtype'\pAnnScriptSB{'}$ such that $\rsubof{\rtype'\pAnnScriptSB{'}}{\rtype\pAnnSB}$. According to this observation we introduce the notion of
\emph{annotated sort environments}, ranged over by $\AnnTypEnv$ and written $\xname:\rtype\WpAnnSB,\,\overline{\xname}:\overline{\rtype}$, that contain one $\WpAnn$-annotated sort assumption and some (possibly none) sort assumptions for program variables. The annotated sort checking rule for user-defined diffusions \ruleNameSize{[A-DEF]} (which derives  judgements of the form
$\annFunTypJud{}{\FUNCTION}{\overline{\rtype(\overline{\rtype})\pAnnSB}}$) uses this strategy to check that the definition of a user-defined diffusion $\dname$ with respect to all the annotated sort signatures in 
$\asignaturesOf{\dname}$.

 \begin{figure}[!t]{
 \framebox[1\textwidth]{
 $\begin{array}{l}
$\quad$\textbf{Pure expression annotated sort checking:}  \hfill
 \boxed{\annExpTypJud{\AnnTypEnv}{\e}{\artype}}$\quad$
  \\
\begin{array}{c}
\nullsurfaceTyping{A-VAR}{
\annExpTypJud{\AnnTypEnv,\xname:\rtype\inputAnnSB}{\xname}{\rtype\inputAnnSB}
}
\qquad\qquad
\surfaceTyping{A-GVAL}{ \quad
\topValOf{\leftKeyOf{\rtype}}=\groundvalue
}{
\annExpTypJud{\AnnTypEnv,\xname:\rtype\inputAnnSB}{\groundvalue}{\atypeof{\groundvalue}}
}
\skiptransition
\surfaceTyping{A-PAIR}{  \quad
\annExpTypJud{\AnnTypEnv}{\e_1}{\rtype_1\pAnnSB}
\quad
\surExpTypJud{\annErasure{\AnnTypEnv}}{\e_2}{\rtype_2}
}{ \annExpTypJud{\AnnTypEnv}{\mkpair{\e_1}{\e_2}}{\mkpairType{\rtype_1}{\rtype_2}\pAnnSB} }
\qquad\qquad
\surfaceTyping{A-FST}{ \;\;
\annExpTypJud{\AnnTypEnv}{\e}{\mkpairType{\rtype_1}{\rtype_2}\pAnnSB} }{
\annExpTypJud{\AnnTypEnv}{\fstK\;\e}{\rtype_1\pAnnSB} }
\skiptransition
\surfaceTyping{A-COND}{  \quad
\surExpTypJud{\annErasure{\AnnTypEnv}}{\e_0}{\boolType}
\quad
\annExpTypJud{\AnnTypEnv}{\e_1}{\artype_1}
\quad
\annExpTypJud{\AnnTypEnv}{\e_2}{\artype_2}
\quad
\artype=\refSupOf{\artype_1}{\artype_2}
}{ \annExpTypJud{\AnnTypEnv}{ \condExpr{\e_0}{\e_1}{\e_2}}{\artype} }
\skiptransition
\surfaceTyping{A-COND-TRUE}{  \quad
\surExpTypJud{\annErasure{\AnnTypEnv}}{\e_0}{\trueType}
\quad
\annExpTypJud{\AnnTypEnv}{\e_1}{\artype}
\quad
\surExpTypJud{\annErasure{\AnnTypEnv}}{\e_2}{\rtype}
}{ \annExpTypJud{\AnnTypEnv}{ \condExpr{\e_0}{\e_1}{\e_2}}{\artype} }
\skiptransition
\surfaceTyping{A-COND-FALSE}{  \quad
\surExpTypJud{\annErasure{\AnnTypEnv}}{\e_0}{\falseType}
\quad
\surExpTypJud{\annErasure{\AnnTypEnv}}{\e_1}{\rtype}
\quad
\annExpTypJud{\AnnTypEnv}{\e_2}{\artype}
}{\annExpTypJud{\AnnTypEnv}{ \condExpr{\e_0}{\e_1}{\e_2}}{\artype} }
\skiptransition
\surfaceTyping{A-FUN}{ \\
 \annExpTypJud{\AnnTypEnv}{\e_1}{\rtype_1\SB{\pAnn''}}
\qquad
\surExpTypJud{\annErasure{\AnnTypEnv}}{\overline{\e}}{\overline{\rtype}}
\qquad
 \rtype(\cdots)\SB{\pAnn'}\in\mostSpecificOf{\asignaturesOf{\fname}}{\rtype_1\overline{\rtype}}
\qquad
\pAnn=\pAnn'(\pAnn'') 
}{ \annExpTypJud{\AnnTypEnv}{\fname(\e_1,\overline{\e})}{\rtype\pAnnSB} }
\skiptransition
%
%
%
\end{array}
\\
$\quad$\textbf{User-defined \diffusionText\ annotated sort checking:} \hfill
  \boxed{\annFunTypJud{}{\FUNCTION}{\overline{\rtype(\overline{\rtype})\pAnnSB}}}$\quad$
  \\
\qquad\begin{array}{c}
\surfaceTyping{A-DEF}{ \\
\mbox{for all $\rtype(\rtype_1\overline{\rtype})\pAnnSB\in\asignaturesOf{\dname}$,}
\\
\qquad
\surExpTypJud{\overline{\xname}:\rtype_1\WpAnnSB\,\overline{\rtype}}{\e}{\rtype'\pAnnScriptSB{'}}
\qquad
\rsubof{\rtype'\pAnnScriptSB{'}}{\rtype\pAnnSB}
}{ \annFunTypJud{}{\defK \; \anytype \;\dname
(\overline{\anytype\;\xname}) \; \isK \; \e}{\asignaturesOf{\dname}} }
\end{array}
\end{array}$}
} 
\caption{Annotated sort checking rules for expressions and \diffusionText\ definitions} \label{fig:AnnotatedTyping}
\end{figure}

The annotated sort-checking judgement for expressions is of the form $\annExpTypJud{\AnnTypEnv}{\e}{\artype}$, to be read: pure-expression $\e$ has annotated sort $\artype$ under the assumptions
$\AnnTypEnv$ for the program variables occurring in $\e$. 
The support of an annotated sort environment $\AnnTypEnv$, denoted by $\annErasure{\AnnTypEnv}$, is the sort environment obtained from $\AnnTypEnv$ by removing the input annotation, i.e.,
\[
\annErasure{\xname:\rtype\WpAnnSB,\,\overline{\xname}:\overline{\rtype}} = \xname:\rtype,\,\overline{\xname}:\overline{\rtype}
\] 
(cf.\ Section~\ref{sec-annotatedSorts}).
Some of the annotated sort-checking rules rely on the judgements of the sort checking system introduced in Section~\ref{sec-AnnotatedSortSignaturesMappings} to sort check some subexpressions. Namely: the right element of the pair in rule \ruleNameSize{[A-PAIR]}; the condition of the conditional-expression in rules \ruleNameSize{[A-COND]}, \ruleNameSize{[A-COND-TRUE]} and \ruleNameSize{[A-COND-FALSE]}; the right branch of the conditional-expression in rule \ruleNameSize{[A-COND-TRUE]}; the left branch of the conditional-expression in rule \ruleNameSize{[A-COND-FALSE]}; and the arguments (excluding the first one)  of the pure function  $\fname$ (that must be an $\pAnn$-diffusion) in rule \ruleNameSize{[A-FUN]}.
 
Annotated sort checking of variables and ground values is similar to sort checking. In particular, ground values  may be given an annotated sort by construction by exploiting the mapping  $\atypeof{\cdot}$---the premise $\topValOf{\leftKeyOf{\rtype}}=\groundvalue$ ensures that the value $\groundvalue$ is relevant to the overall goal of the annotated sort checking derivation process, that is: deriving an annotated sort $\rtype'\pAnnScriptSB{'}$ such that $\rsubof{\rtype'\pAnnScriptSB{'}}{\rtype\pAnnSB}$ for the body $\e$ of a user-defined diffusion $\dname$ in  order to check that $\dname$ has the annotated sort-signature $\rtype(\rtype_1,\myldots,\rtype_n)\pAnnSB$.

 Note that there is no annotated sort checking rule for expressions of the form $\sndK\;\e$---because of the leftmost-as-key preorder such an expression is not relevant  to the overall goal of the annotated sort checking derivation process.

The  sort-checking rule \ruleNameSize{[A-FUN]} for diffusion application exploits the mapping  $\asignaturesOf{\cdot}$  and the auxiliary mapping $\mostSpecificOf{\cdot}{\cdot}$  introduced in Section~\ref{sec-ann-typing-auxiliary}. It first infers 
the sort $\rtype_1\pAnnScriptSB{''}$ for the first argument $\e_1$ of $\fname$ and
the sorts $\overline{\rtype}$ for remaining the arguments $\overline{\e}$ of $\fname$, then uses the most specific annotated sort signature for $\overline{\rtype}$  in $\asignaturesOf{\fname}$  for assigning to the application $\fname(\overline{\e})$ the minimum annotated sort  $\rtype\pAnnScriptSB{'}$ that can be assigned to $\fname(\overline{\e})$ by using for $\fname$ any of the annotated sort signatures in 
$\rsignaturesOf{\fname}$.

We say that a program (or library) $\PROGRAM$ is \emph{well annotated}
 to mean that:  all the user-defined function definitions
in $\PROGRAM$ check by using  the sort-checking rules in Figure~\ref{fig:RefinementTyping}, and
all the user-defined \diffusionText\ definitions sort checking rules in Figure~\ref{fig:AnnotatedTyping}.

\begin{exa} All user-defined functions provided in the examples in Sections~\ref{sec-examples} and~\ref{sec-examples-extended} sort-check and (when they are \diffusionsText) annotate sort check by assuming the
ground sorts and the subsorting given in Figure~\ref{fig:refinement-subsortingGROUND}, 
  the sorts for the ground values and sensors given in Figure~\ref{fig:refinement-valuesANDsensorsGROUND}, 
the 
sort-signatures for built-in functions given in
Figure~\ref{fig:refinement-instanceBUILT-IN},
the annotated sorts for ground values given in
Figure~\ref{fig:annotated-instanceGROUND},
 the 
annotated sort-signatures for built-in functions given in Figures~\ref{fig:annotated-instanceBUILT-IN} and the annotated sort-signatures for user-defined diffusions given in~\ref{fig:UserDeclaredAnnotatedSignaturesForExamples}.
\end{exa}

Since no choice may be done when building a derivation for a given annotated sort-checking judgment, the annotated sort-checking rules straightforwardly describe an annotated sort-checking algorithm.

\subsection{Annotation Soundness}\label{sec-properties-annotating}

%
%

The following theorem states the correctness of the annotation-checking system  presented in Section~\ref{sec-annotatedSortChecking} 

\begin{thm}[Annotation soundness]\label{the-AnnotationSoundness} \emph{If
$\annFunTypJud{}{\FUNCTION}{\overline{\rtype(\overline{\rtype})\pAnnSB}}$ holds, then $\pOf{\fname}{\rtype(\overline{\rtype})}$ holds for all 
$\rtype(\overline{\rtype})\pAnnSB\in\overline{\rtype(\overline{\rtype})\pAnnSB}$.}
\end{thm}

\begin{proof}
See Appendix~\ref{app-proof-annotationSoundness}.
\end{proof}

\section{Related Work and Discussion}
\label{sec-related}

We here discuss the main related pieces of work, rooted in previous research on finding core models for spatial computing and self-organisation, formal approaches for large-scale systems, and finally on self-stabilisation in distributed systems.

\subsection{Spatial computing and self-organisation}

A first step in studying general behavioural properties of self-organising systems is the identification of a reference model, making it possible to reuse results across many different models, languages and platforms.
The review in \cite{SpatialIGI2013} surveys a good deal of the approaches considering some notion of space-time computations, which are the basis for any self-organising system.
Examples of such models include the Hood sensor network abstraction~\cite{hood}, the $\sigma\tau$-Linda model \cite{spatialcoord-coord2012}, the SAPERE computing model \cite{VPMSZ-SCP2015}, and TOTA middleware \cite{tota}, which all implement computational fields using similar notions of spreading.
More generally, Proto \cite{mitproto,proto06a} and its core formalisation as the ``field calculus'' \cite{DVPB-FORTE2015}, provides a functional model that appear general enough to serve as a starting point for investigating behavioural aspects of spatial computation and self-organisation \cite{BV-PTRS2015}.
In fact, in \cite{BVD-SCW14} it is proved that the field calculus is universal, in the sense that it can be used to describe any causal and discretely-approximable computation in space-time.

Hence, we started from the field calculus, in which computation is expressed by the functional combination of input fields (as provided by sensors), combined with mechanisms of space-based (neighbour) data aggregation, restriction (distributed branch) and state persistence.
The calculus presented here is a fragment of the field calculus, focussing on only two basic computational elements: \emph{(i)} functional composition of fields, and \emph{(ii)} a spreading expression.
In particular, the latter is a suitable combination of basic mechanisms of the field calculus, for which we were able to prove convergence to a single final state.
Namely, a spreading expression \texttt{\{e${}_0$ : g(@,e${}_1$,..,e${}_n$)\}} in our calculus is equivalent to the following field calculus expression:

\[\texttt{(rep x (inf) (min e${}_0$ (g (min-hood+ (nbr x)) e${}_1$ .. e${}_n$)))}\]
In particular, it was key to our end to neglect recursive
function calls (in order to ensure termination of device
fires, since the calculus does not model the domain restriction construct~\cite{VDB-FOCLASA-CIC2013,mitproto} and both the branches of a conditional expression are evaluated), stateful operations (in our model, the state of a
device is always cleaned up before computing the new one), and
to restrict aggregation to minimum function and progression to
what we called ``stabilising diffusion'' functions.

Other than applying to fragments of the field calculus, the result provided here can be applied to rule-based systems like those of the SAPERE approach \cite{VPMSZ-SCP2015} and of rewrite-based coordination models \cite{spatialcoord-coord2012}, along the lines depicted in \cite{V-SCW2013}.
Note that our condition for self-stabilisation is only a sufficient one.
A primary example of the fact that it is not necessary is Laplacian consensus \cite{DBLP:conf/atal/ElhageB10}, expressed as follows in the field calculus:
\[\texttt{(rep x e${}_i$ (+ x (* e${}_{\epsilon}$ (sum-hood (- (nbr x) x))))))}\]
It cannot be expressed in the calculus we propose here, but still stabilises to a plateau field, computed as a consensus among the values of input field \texttt{e${}_i$} (with \texttt{e${}_\epsilon$)} driving the dynamics of the output field.
Other cases include so-called convergence cast \cite{BV-FOCAS2014}.

\subsection{Formal approaches}

In this paper we are interested in formally predicting the behaviour of a complex system, in which the local interactions among a possibly miriad of devices make a global and robust pattern of behaviour emerge.
In the general case, one such kind of prediction can hardly be obtained.

The quintessential formal approach, model-checking \cite{clarke-99}, cannot typically scale with the number of involved components: suitable abstractions are needed to model arbitrary-size systems (as in \cite{Delzanno02}), which however only work in very constrained situations.
Approximate model-checking \cite{HLMP04,VC-ASENSIS2012}, basically consisting in a high number of simulation bursts, is viable in principle, but it still falls under the umbrella of semi-empirical evaluations, for only statistical results are provided.
Recently, fluid flow approximation has been proposed to turn large-scale computational systems into systems of differential equations that one could solve analytically or use to derive an evaluation of system behaviour \cite{DBLP:conf/coordination/BortolussiLM13}.
Unfortunately, this approach seems developed yet only to abstract from the number of (equivalent and non-situated) agents performing a repetitive task, instead of abstracting from the discreteness of a large-scale situated computational network.

Recent works finally aim at proving properties of large-scale systems by hand-written proofs, which are the works most related to the result of present paper.
The only work aiming at a mathematical proof of stabilisation for the specific case of computational fields is \cite{crf}.
There, a self-healing gradient algorithm called CRF (constraints and restoring forces) is introduced to estimate physical distance in a spatial computer, where the neighbouring relation is fixed to unit-disc radio, and node firing is strictly connected to physical time.
Compared to our approach, the work in \cite{crf} tackles a more specific problem, and is highly dependent on the underlying spatial computer assumptions.
Another work presenting a proof methodology that could be helpful in future stages of our research is the universality study in \cite{BVD-SCW14}.

From the viewpoint of rewrite semantics \cite{Huet:1980:CRA:322217.322230}, which is the meta-model closest to our formalisation attempt, our proof is most closely related to the \emph{confluence} property, that is, don't care non-determinism.
Our result entails confluence, but it is actually a much strongest property of global uniqueness of a normal form, independently of initial state.

\subsection{Self-stabilisation}

Our work concerns the problem of identifying complex network computations whose outcome is predictable.
The notion we focus on requires a unique global state being reached in finite time independently of the initial state, that is, depending only on the state of the environment (topology and sensors).
It is named \emph{(strong) self-stabilisation} since it is related with a usual notion of self-stabilisation to \emph{correct} states for distributed systems \cite{dolev}, defined in terms of a set $C$ of correct states in which the system eventually enters in finite time, and then never escapes from -- in our case, $C$ is made by the single state corresponding to the sougth result of computation.

Actually many different versions of the notion of self-stabilisation have been adopted in past, surveyed in \cite{S93c}, from works of Dijkstra's \cite{D73b,D74} to more recent and abstract ones \cite{AG93}, typically depending on the reference model for the system to study---protocols, state machines.
In our case, self-stabilisation is studied for a distributed data structure (the computational field).
Previous work on this context like \cite{HP01} however only considers the case of heap-like data structures in a non-distributed settings: this generally makes it difficult to draw a bridge with existing research.

Several variations of the definition also deal with different levels of quality (fairness, performance).
For instance, the notion of superstabilisation \cite{Dolev:1997:SPD:866056} adds to the standard self-stabilisation definition a requirement on a ``passage predicate'' that should hold while a system recovers from a specific topological change.
Our work does not address this very issue, since we currently completely equate the treatment of topological changes and changes to the inputs (i.e., sensors), and do not address specific performance requirements.
However, future works addressing performance issues will likely require some of the techniques studied in \cite{Dolev:1997:SPD:866056}.
Performance is also affected by the fairness assumption adopted: we relied on a notion abstracting from more concrete ones typically used \cite{KC98}, which we could use as well though losing a bit of the generality of our result.

Concerning the specific technical result achieved here, the closest one appears to be the creation of a hop-count gradient, which is known to self-stabilise: this is used in \cite{dolev} as a preliminary step in the creation of the spanning tree of a graph.
The main novelty in this context is that self-stabilisation is not proved here for a specific algorithm/system: it is proved for all fields inductively obtained by functional composition of fixed fields (sensors, values) and by a gradient-inspired spreading process.
Other works attempt to devise general methodologies for building self-stabilising systems like we do.
The work in \cite{AV91} depicts a compiler turning any protocol into a self-stabilising one. Though this is technically unrelated to our solution, it shares the philosophy of hiding the details of how self-stabilisation is achieved under the hood of the execution platform: in our case in fact, the designer wants to focus on the macro-level specification, trusting that components behave and interact so as to achieve the global outcome in a self-stabilising way.
The work in \cite{GH91} suggests that hierarchical composition of self-stabilising programs is self-stabilising: an idea that is key to the construction of a functional language of self-stabilising ``programs''.

In spite of the connection with some of these previous works, to the best of our knowledge ours is novel under different dimensions.
First, it is the first attempt of providing a notion of self-stabilisation directly connected to the problem of engineering self-organisation.
Secondly, the idea of appliying it to a whole specification language is also new, along with the fact that we apply a type-based approach, providing a correct checking procedure that paves the way towards compiler support.
As we use now a type-based approach, other static analysis techniques
may be worth studying in future attempts (see,
e.g.,~\cite{NielsonNielsonHankin:BOOK-1999}).

\section{Conclusions and Future Work}
\label{sec-conclusion}

Emerging application scenarios like pervasive computing, robotic systems, and wireless sensor networks call for developing robust and predictable large-scale situated systems.
However, the diffusion/aggregation processes that are typically to be implemented therein are source of complex phenomena, and are notoriously very hard to be formally treated.
The goal of this work is to bootstrap a research thread in which mechanisms of self-organisation are captured by linguistic constructs, so that static analysis in the programming language style can be used to isolate fragments with provable predictable behaviour.
In the medium term, we believe this is key to provide a tool-chain (programming language, libraries, simulation and execution platforms) which enables the development of complex software systems whose behaviour has still some predictability obtained ``by construction''.

Along this line, this paper studies a notion of strong self-stabilisation, identifying a sufficient condition expressed on the diffusion/aggregation mechanisms occurring in the system.
This targets the remarkable situation in which the final shape of the distributed data structure created (i.e., the computational field) is deterministically established, does not depend on transient events (such as temporaneous failures), and is only determined by the stabilised network topology.
Namely, this is the case in which we can associate to a complex computation a deterministic and easily computable result.

It would be interesting to relax some of the conditions and assumptions we relied upon in this paper, so as to provide a more general self-stabilisation result.
First of all, the current definition of self-stabilisation requires values to have upper-bounds, to prevent network subparts that become isolated from ``sources'' (e.g., of a gradient) to be associated with values that grow to infinity without reaching a fixpoint: a more involved definition of self-stabilisation could be given that declares such a divergence as admitted, allowing us to relax ``noetherianity'' of values.
We also plan to extend the result to encompass the domain restriction construct~\cite{VDB-FOCLASA-CIC2013,mitproto} (i.e., to add to the calculus a form of conditional expression where only one of the branches is evaluated). In this way also recursive function definitions could be added (then, in order to guarantee termination of computation rounds, standard analysis techniques for checking termination of recursive function definitions might be used).
We currently focus only on spreading-like self-stabilisation, whether recent works \cite{BV-FOCAS2014} suggest that ``aggregation'' patterns can be similarly addressed as well, though they might require a completely different language and proof methodology.

Of course, other behavioural properties are of interest which we plan to study in future work as an extension to the results discussed here.
First, it is key to study notions of self-stabilisation for computational fields which are designed so as to be dynamically evolving, like e.g. the anticipative gradient \cite{MPV-SASO2012}.
Second, it would be interesting to extend our notion of self-stabilisation so as to take into account those cases in which only approximate reachability of the sought state is required, since it can lead to computations with better average performance, as proposed in \cite{DBLP:conf/sac/Beal09}. 
Other aspects of interest that can be formally handled include performance characterisation, code mobility, expressiveness of mechanisms, and independence of network density, which will likely be subject of next investigations as well.

\section*{Acknowledgement}
We thank the anonymous COORDINATION 2014
referees for useful comments on an earlier version of this paper, and the anonymous LMCS referees for insightful comments  and suggestions  for improving the
presentation.


\bibliographystyle{abbrv}
\bibliography{main}

\appendix
%
\section{ Proof of Theorem~\ref{the-DeviceTypePreservation} and Theorem~\ref{the-DeviceTermination}}\label{app-proof-typeSounness}

The proof are given for the calculus with pairs  (cf. Section~\ref{sec-calculus-extended}).

\begin{lem}[Substitution lemma for typing]\label{lem:substitutionForTyping}
If
$\surExpTypJud{\overline{\xname}:\overline{\anytype}}{\e}{\anytype}$,
$\lengthOf{\overline{\anyvalue}}=\lengthOf{\overline{\xname}}$
and
        $\surExpTypJud{\emptyset}{\overline{\anyvalue}}{\overline{\anytype}}$,
then
$\surExpTypJud{\emptyset}{\applySubstitution{\e}{\substitution{\overline{\xname}}{\overline{\anyvalue}}}}{\anytype}$.
\end{lem}

\begin{proof}
Straightforward by induction on the application of the typing rules for expressions in Fig.~\ref{fig:SurfaceTyping} and Fig.~\ref{fig:SurfaceTyping-extended}. 
\end{proof}

\begin{lem}[Device computation
type preservation]\label{lem:DeviceTypePreservation}
If
$\vtree\in\bsWFVT{\overline{\xname}:\overline{\anytype}}{\e}{\anytype}$,
then $\surExpTypJud{\emptyset}{\vrootOf{\vtree}}{\anytype}$.
\end{lem}

\begin{proof}
Recall that the typing rules (in Fig.~\ref{fig:SurfaceTyping} and Fig.~\ref{fig:SurfaceTyping-extended}) and the evaluation rules  (in Fig.~\ref{fig:deviceSemantics} and Fig.~\ref{fig:deviceSemantics-extended}) are syntax directed.
The proof is by induction on the definition of $\bsWFVT{\overline{\xname}:\overline{\anytype}}{\e}{\anytype}$ (given in Section~\ref{sec-properties-typing}), on
the number of user-defined function calls that may be encountered during the
evaluation of $\applySubstitution{\e}{\substitution{\overline{\xname}}{\overline{\anyvalue}}}$ (cf.\ sanity condition \emph{\textbf{(iii)}} in
Section~\ref{sec-typing}), and on the syntax of closed expressions.

From the hypothesis $\vtree\in\bsWFVT{\overline{\xname}:\overline{\anytype}}{\e}{\anytype}$ we have $\surExpTypJud{\overline{\xname}:\overline{\anytype}}{\e}{\anytype}$, $\bsopsem{\snsFun}{\overline{\vtree}}{\applySubstitution{\e}{\substitution{\overline{\xname}}{\overline{\anyvalue}}}}{\vtree}$  for some sensor mapping $\snsFun$,  evaluation trees $\overline{\vtree}\in\bsWFVT{\overline{\xname}:\overline{\anytype}}{\e}{\anytype}$, and values $\overline{\anyvalue}$ such that $\lengthOf{\overline{\anyvalue}}=\lengthOf{\overline{\xname}}$
and
        $\surExpTypJud{\emptyset}{\overline{\anyvalue}}{\overline{\anytype}}$. Moreover, by Lemma~\ref{lem:substitutionForTyping} we have that $\surExpTypJud{\emptyset}{\applySubstitution{\e}{\substitution{\overline{\xname}}{\overline{\anyvalue}}}}{\anytype}$ holds. The case $\overline{\vtree}$ empty represents the base of the induction on the definition of $\bsWFVT{\overline{\xname}:\overline{\anytype}}{\e}{\anytype}$. Therefore the rest of this proof can be understood as a proof of the base step by assuming  $\overline{\vtree}=\emptyset$ and a proof of the inductive step  by assuming $\overline{\vtree}\not=\emptyset$.

The case when $\e$ does non contain user-defined function calls represents the base on the induction on the number of user-defined function calls that may be encountered during the
evaluation of  $\applySubstitution{\e}{\substitution{\overline{\xname}}{\overline{\anyvalue}}}$. Therefore the rest of this proof can be understood as a proof of the base step by ignoring the cases  $\applySubstitution{\e}{\substitution{\overline{\xname}}{\overline{\anyvalue}}}=\fname(\e_1,\myldots,\e_n)$ and  $\applySubstitution{\e}{\substitution{\overline{\xname}}{\overline{\anyvalue}}}=\spreadThree{\e_0}{\fname}{\e_1,\myldots,\e_n}$ when $\fname$ is a used-defined function $\dname$.
The base of the induction on  $\applySubstitution{\e}{\substitution{\overline{\xname}}{\overline{\anyvalue}}}$ consist of two cases.
\begin{description}
\item[Case $\snsname$] 
From the hypothesis we have  $\surExpTypJud{\emptyset}{\snsname}{\anytype}$ where $\anytype=\typeof{\snsname}$ (by rule \ruleNameSize{[T-SNS]}) and $\bsopsem{\senstate}{\overline\vtree}{\snsname}{\vtree}$ where $\vtree=\anyvalue()$ and $\anyvalue=\senstate(\snsname)$ (by rule \ruleNameSize{[E-SNS]}). Since the sensor  $\snsname$ returns values of type $\typeof{\snsname}$, we have that $\typeof{\anyvalue}=\typeof{\snsname}=\anytype$.  So  the result follows by a straightforward induction of the syntax of values using rules \ruleNameSize{[T-VAL]} and \ruleNameSize{[T-PAIR]}. 
\item[Case $\anyvalue$]
From the hypothesis we have  $\surExpTypJud{\emptyset}{\anyvalue}{\anytype}$  and  (by rule \ruleNameSize{[E-VAL]}) $\bsopsem{\senstate}{\overline\vtree}{\anyvalue}{\vtree}$ where  $\vtree=\anyvalue()$.  So  the result follows by a straightforward induction of the syntax of values using rules \ruleNameSize{[T-VAL]} and \ruleNameSize{[T-PAIR]}.
\end{description}
For the inductive step on  $\applySubstitution{\e}{\substitution{\overline{\xname}}{\overline{\anyvalue}}}$, we show only the two most interesting cases (all the other cases are straightforward by induction). 
\begin{description}
\item[Case $\dname(\e_1,\myldots,\e_n)$] 
From the  hypothesis we have $ \surExpTypJud{\emptyset}{\fname(\e_1,\myldots,\e_n)}{\anytype}$ (by rule \ruleNameSize{T-FUN]}) and $\bsopsem{\senstate}{\overline\vtree}{\dname(\e_1,\myldots,\e_n)}{\anyvalue(\vtree'_1,\myldots,\vtree'_n,\anyvalue(\overline\vtreealt))}$ (by rule \ruleNameSize{E-DEF]}). Therefore we have
$\anytype(\anytype_1,\myldots,\anytype_n)=\signatureOf{\dname}$,  $\surExpTypJud{\emptyset}{\e_1}{\anytype_1}$, $\myldots$, $\surExpTypJud{\emptyset}{\e_n}{\anytype_n}$ (the premises of rule \ruleNameSize{T-FUN]}) and $\defK \; \anytype \;\dname(\anytype_1\;\xname_1,\myldots,\anytype_n\;\xname_n) = \e''$, $\bsopsem{\senstate}{\premiseNumOf{1}{\overline\vtree}}{\e_1}{\vtree'_1}$, $\myldots$,
$\bsopsem{\senstate}{\premiseNumOf{n}{\overline\vtree}}{\e_n}{\vtree'_n}$ and
\begin{equation}\label{eq1:the-DeviceTypePreservation}
\bsopsem{\senstate}{\premiseNumOf{n+1}{\overline\vtree}}{\e'}{\anyvalue(\overline\vtreealt)}
\quad \mbox{where} \; \e'=\applySubstitution{\e''}{\substitution{\xname_1}{\vrootOf{\vtree'_1}},\myldots,\substitution{\xname_n}{\vrootOf{\vtree'_n}}}
\end{equation}
(the premises of rule \ruleNameSize{E-DEF]}).

Since $\premiseNumOf{i}{\overline\vtree}\in\bsWFVT{\emptyset}{\e_i}{\anytype_i}$ ($1\le i\le n$) then $\vtree'_i\in\bsWFVT{\emptyset}{\e_i}{\anytype_i}$; therefore, by  induction we have $\surExpTypJud{\emptyset}{\vrootOf{\vtree'_1}}{\anytype_1}$, $\myldots$, $\surExpTypJud{\emptyset}{\vrootOf{\vtree'_n}}{\anytype_n}$. 

Since the program is well typed (cf.\ Section~\ref{sec-typing}) we have  $\surExpTypJud{\xname_1:\anytype_1,\myldots,\xname_n:\anytype_n}{\e''}{\anytype}$ (by rule \ruleNameSize{T-DEF]}). 


Since $\premiseNumOf{n+1}{\overline\vtree}\in\bsWFVT{\xname_1:\anytype_1,\myldots,\xname_n:\anytype_n}{\e''}{\anytype}$,
then (by~(\ref{eq1:the-DeviceTypePreservation})) we have $\anyvalue(\overline\vtreealt)\in\bsWFVT{\xname_1:\anytype_1,\myldots,\xname_n:\anytype_n}{\e''}{\anytype}$; therefore, by induction we have that $\surExpTypJud{\emptyset}{\anyvalue}{\anytype}$.

\item[Case $\spreadThree{\e_0}{\fname}{\e_1,\myldots,\e_n}$] 
From the  hypothesis we have $\surExpTypJud{\emptyset}{\spreadThree{\e_0}{\fname}{\e_1,\myldots,\e_n}}{\anytype}$ (by rule \ruleNameSize{T-SPR]}) and $\bsopsem{\senstate}{\overline\vtree}{\spreadThree{\e_0}{\fname}{\e_1,\myldots,\e_n}}{\lowerBound\{\anyvalue_0,\anyvaluealt_1,\myldots,\anyvaluealt_m\}(\vtreealt_0,\vtreealt_1,\myldots,\vtreealt_n)}$ (by rule \ruleNameSize{E-SPR]}). Therefore we have
 $\progressionOpOf{\fname}$, $\anytype(\anytype,\anytype_1,\myldots,\anytype_n)=\signatureOf{\dname}$, $\surExpTypJud{\emptyset}{\e_0}{\anytype}$, $\surExpTypJud{\emptyset}{\e_1}{\anytype_1}$, $\myldots$, $\surExpTypJud{\emptyset}{\e_n}{\anytype_n}$ (the premises of rules \ruleNameSize{T-SPR]} and \ruleNameSize{T-FUN]}) and $\bsopsem{\senstate}{\premiseNumOf{1}{\overline\vtree}}{\e_1}{\vtree'_1}$, $\myldots$,
$\bsopsem{\senstate}{\premiseNumOf{n}{\overline\vtree}}{\e_n}{\vtree'_n}$, 
$\conclusionOf{\vtreealt_0,\myldots,\vtreealt_n} = \anyvalue_0 \myldots \anyvalue_n$,
$\conclusionOf{\overline\vtree} = \anyvaluebis_1 \myldots \anyvaluebis_m$,
\begin{equation}\label{eq3:the-DeviceTypePreservation}
\bsopsem{\senstate}{\emptyset}{\fname(\anyvaluebis_1,\anyvalue_1,\myldots,\anyvalue_n)}{\anyvaluealt_1(\cdots)},
\myldots,
\bsopsem{\senstate}{\emptyset}{\fname(\anyvaluebis_m,\anyvalue_1,\myldots,\anyvalue_n)}{\anyvaluealt_m(\cdots)}
\end{equation}
(the premises rule \ruleNameSize{E-SPR]}).
By induction we have $\surExpTypJud{\emptyset}{ \anyvalue_0 \myldots \anyvalue_n}{\anytype \anytype_1 \myldots \anytype_n}$
and  $\surExpTypJud{\emptyset}{\anyvaluebis_1}{\anytype}$, $\myldots$, $\surExpTypJud{\emptyset}{\anyvaluebis_m}{\anytype}$.
We have two subcases:
\begin{itemize}
\item
If $\fname$ is a user-defined function, then from~(\ref{eq3:the-DeviceTypePreservation}) we get (by reasoning as in the proof of case $\dname(\e_1,\myldots,\e_n)$)
$\surExpTypJud{\emptyset}{\anyvaluealt_1}{\anytype}$, $\myldots$, $\surExpTypJud{\emptyset}{\anyvaluealt_m}{\anytype}$.
\item
If $\fname$ is a built-in function, then from~(\ref{eq3:the-DeviceTypePreservation}) we get (by the semantics of built-in functions) $\surExpTypJud{\emptyset}{\anyvaluealt_1}{\anytype}$, $\myldots$, $\surExpTypJud{\emptyset}{\anyvaluealt_m}{\anytype}$.
\end{itemize}
In both cases $\anyvalue=\lowerBound\{\anyvalue_0,\anyvaluealt_1,\myldots,\anyvaluealt_m\}$ has type $\anytype$, i.e., $\surExpTypJud{\emptyset}{\anyvalue}{\anytype}$ holds.\qedhere
\end{description}
\end{proof}

\noindent\textbf{Restatement of
Theorem~\ref{the-DeviceTypePreservation} (Device computation
type preservation).} \emph{If
$\surExpTypJud{\overline{\xname}:\overline{\anytype}}{\e}{\anytype}$,
 $\snsFun$ is a sensor mapping,
$\overline{\vtree}\in\bsWFVT{\overline{\xname}:\overline{\anytype}}{\e}{\anytype}$,
$\lengthOf{\overline{\anyvalue}}=\lengthOf{\overline{\xname}}$,
        $\surExpTypJud{\emptyset}{\overline{\anyvalue}}{\overline{\anytype}}$
and
$\bsopsem{\snsFun}{\overline{\vtree}}{\applySubstitution{\e}{\substitution{\overline{\xname}}{\overline{\anyvalue}}}}{\vtree}$,
then $\surExpTypJud{\emptyset}{\vrootOf{\vtree}}{\anytype}$.}

\begin{proof}
Straightforward by Lemma~\ref{lem:DeviceTypePreservation}, since $\vtree\in\bsWFVT{\overline{\xname}:\overline{\anytype}}{\e}{\anytype}$.
\end{proof}

\noindent\textbf{Restatement of
Theorem~\ref{the-DeviceTermination} (Device computation
termination).} \emph{If
$\surExpTypJud{\overline{\xname}:\overline{\anytype}}{\e}{\anytype}$,
 $\snsFun$ is a sensor mapping,
$\overline{\vtree}\in\bsWFVT{\overline{\xname}:\overline{\anytype}}{\e}{\anytype}$,
$\lengthOf{\overline{\anyvalue}}=\lengthOf{\overline{\xname}}$
and
        $\surExpTypJud{\emptyset}{\overline{\anyvalue}}{\overline{\anytype}}$,
then
$\bsopsem{\snsFun}{\overline{\vtree}}{\applySubstitution{\e}{\substitution{\overline{\xname}}{\overline{\anyvalue}}}}{\vtree}$
for some value-tree $\vtree$.}

\begin{proof}
By induction on the number of
function calls that may be encountered during the evaluation of
$\applySubstitution{\e}{\substitution{\overline{\xname}}{\overline{\anyvalue}}}$
(cf.\ sanity condition \emph{\textbf{(iii)}} in
Section~\ref{sec-typing}) and on  the syntax of closed expressions, using Lemma~\ref{lem:DeviceTypePreservation}.
\end{proof}

\section{Proof of Theorem~\ref{the-network-self--stabilization}}\label{app-proof-selfStabilisation}

Recall the auxiliary definitions and the outline of  the proof of Theorem~\ref{the-network-self--stabilization} given in Section~\ref{sec-properties-self-stabilisation}. The following six lemmas corresponds to the  auxiliary results~\ref{lem-minimum-value}-\ref{lem-pre-self-stable-network-self-stabilization} introduced in Section~\ref{sec-properties-self-stabilisation}.

\begin{lem}[Minimum value]\label{lem-minimum-value}
Given a program $\e=\spreadThree{\e_0}{\fname}{\e_1,\ldots,\e_n}$ with valid  sort and stabilising diffusion assumptions, for every reachable pre-self-stable network configuration $\network$, for any device $\deviceId$ in $\network$ such that
$\anyvalueInNC{\deviceId}{\network}$ is minimum (among
the values of the devices in $\network$): 
\begin{enumerate}
\item
if $\network\netArrIdStar\network'$ and $\deviceId\not\in\overline{\deviceId}$, then $\anyvalueInNC{\deviceId}{\network'}=\anyvalueInNC{\deviceId}{\network}$ is minimum (among
the values of the devices in $\network'$); and
\item
either  $\anyvalueInNC{\deviceId}{\network}=\anyvalue_{0,\deviceId}$
    (i.e., $\deviceId$ is self-stable in $\network$) or 
if $\network\netArrFairStarN{1}\network'$ then there exists $\anyvalue'$ such that:
\begin{enumerate}
\item
    $\anyvalueInNC{\deviceId}{\network}\wfLt\anyvalue'\wfLe\anyvalueInNC{\deviceId}{\network'}$; and
\item if $\network'\netArrPresStar\network''$,
    then
    $\anyvalue'\wfLe\anyvalueInNC{\deviceId}{\network''}$.
\end{enumerate}
\end{enumerate}
\end{lem}

\begin{proof}
\textbf{Point (1).} Consider a device $\deviceId'\not=\deviceId$. Then $\anyvalueInNC{\deviceId}{\network}\wfLe\anyvalueInNC{\deviceId'}{\network}\wfLe\anyvalue_{0,\deviceId'}$  and its $m\ge 0$ neighbours have values $\anyvaluebis_j$  such that 
$\anyvalueInNC{\deviceId}{\network}\wfLe\anyvaluebis_j$ ($1\le j\le m$). Since the  stabilising diffusion assumptions hold, $\anyvaluebis_j\wfLt\opFunOf{\fname}(\anyvaluebis_j,\anyvalue_{1,\deviceId'},\ldots,\anyvalue_{n,\deviceId'})=\anyvaluealt_j$. Therefore,
 $\anyvalueInNC{\deviceId}{\network}\wfLe\lowerBound\{\anyvalue_{0,\deviceId'},\anyvaluealt_1,\myldots,\anyvaluealt_m\}=\anyvalueInNC{\deviceId'}{\network'}$. So $\anyvalueInNC{\deviceId}{\network'}=\anyvalueInNC{\deviceId}{\network}$ is minimum (among
the values of the devices in $\network'$).

\noindent
\textbf{Point (2).} Assume that $\anyvalueInNC{\deviceId}{\network}\wfLt\anyvalue_{0,\deviceId}$. The $m\ge 0$ neighbours of $\deviceId$ have values $\anyvaluebis_j$  such that 
$\anyvalueInNC{\deviceId}{\network}\wfLe\anyvaluebis_j$ ($1\le j\le m$). Since the  stabilising diffusion assumptions hold, $\anyvalueInNC{\deviceId}{\network}\wfLt\opFunOf{\fname}(\anyvalueInNC{\deviceId}{\network},\anyvalue_{1,\deviceId},\ldots,\anyvalue_{n,\deviceId})=\anyvaluealt_0$ and  $\anyvaluealt_0\wfLe\opFunOf{\fname}(\anyvaluebis_j,\anyvalue_{1,\deviceId},\ldots,\anyvalue_{n,\deviceId})=\anyvaluealt_j$. Therefore,
the value $\anyvalue'=\anyvalue_{0,\deviceId}\wedge\anyvaluealt_0$ is such that 
when $\deviceId$ fires its new value $\anyvalueInNC{\deviceId}{\network'}=\lowerBound\{\anyvalue_{0,\deviceId},\anyvaluealt_1,\myldots,\anyvaluealt_m\}$ is such that $\anyvalueInNC{\deviceId}{\network}\wfLt\anyvalue'\wfLe\anyvalueInNC{\deviceId}{\network'}$. Moreover,  since  $\anyvalueInNC{\deviceId}{\network}\wfLt\anyvalueInNC{\deviceId}{\network'}$ we have that in a firing evolution  none of  devices $\deviceId'\not=\deviceId$ will reach a value less than $\anyvalueInNC{\deviceId}{\network}$ and therefore the device $\deviceId$ will never reach a value less than $\anyvalue'$.
\end{proof}

\begin{lem}[Self-stabilisation of the minimum value]\label{lem-self-stabilisation-minimum-value}
Given a program $\e=\spreadThree{\e_0}{\fname}{\e_1,\ldots,\e_n}$ with valid  sort and stabilising diffusion assumptions, for every reachable pre-self-stable network configuration $\network$, 
if $\stableNodeSet_1$ is the subset of the devices in $\network$ 
 such that
$\anyvalue_{0,\deviceId}$ is minimum  (among the values of
$\e_0$ in the devices in $\network$),  then there exists
$k\ge 0$ such that $\network\netArrFairStarN{k}\network'$
implies that  $\stableNodeSet_1$ satisfies the following conditions: 
\begin{enumerate}
\item
each device $\deviceId$ in $\stableNodeSet_1$ is self-stable in $\network'$ and has value
$\anyvalueInNC{\deviceId}{\network'}=\anyvalue_{0,\deviceId}$; 
\item
in $\network'$ each device not in $\stableNodeSet_1$ has a value greater or equal to the values of the devices in $\stableNodeSet_1$ and, during any firing evolution, it will always assume  values greater than the values of the devices in $\stableNodeSet_1$.
\end{enumerate}
\end{lem}

\begin{proof}
The number of devices in the network configuration is
finite, the environment does not change, the network is pre-self-stable, and the  stabilising diffusion assumptions holds. The results follows by  Lemma~\ref{lem-minimum-value}. Namely, if there is a device $\deviceId$ whose value $\anyvalue_\deviceId$ is minimum and such that  $\anyvalue_\deviceId\wfLt\anyvalue_{0,\deviceId}$, then after a 1-fair network evolution  $\deviceId$ reaches a value which is greater or equal to some $\anyvalue'$ such that
\begin{itemize}
\item
    $\anyvalue\wfLt\anyvalue'\wfLe\anyvalue_{0,\deviceId}$; and
\item in any subsequent firing evolution the value of $\deviceId$ will be always greater or equal to $\anyvalue'$. 
\end{itemize}
Therefore, after a finite number $k$ of 1-fair evolutions (i.e., after any $k$-fair evolution)  conditions (1) and (2) in the statement of the lemma are satisfied.
\end{proof}


\begin{lem}[Frontier]\label{lem-frontier}
Given a program $\e=\spreadThree{\e_0}{\fname}{\e_1,\ldots,\e_n}$ with valid  sort and stabilising diffusion assumptions, for every reachable pre-stable network configuration $\network$ with devices $\nodeSubSet$ and a non-empty subset of devices $\stableNodeSet\subset\nodeSubSet$ such that 
\begin{enumerate}[label=(\roman*)]
\item
each device in $\stableNodeSet$ is self-stable in $\network$;
\item
each device in $\nodeSubSet-\stableNodeSet$ has a value greater or equal to the values of the devices in $\stableNodeSet$ and, during any firing evolution, will always assume  values greater or equal to the values of the devices in $\stableNodeSet$; and
\item
$\frontierOfIn{\stableNodeSet}{\nodeSubSet}\not=\emptyset$;
\end{enumerate}
if $\network\netArrFairStarN{1}\network'$ then each device in $\frontierOfIn{\stableNodeSet}{\nodeSubSet}$ is self-stable in $\network'$.
\end{lem}

\begin{proof}
When a device $\deviceId$ in  $\frontierOfIn{\stableNodeSet}{\nodeSubSet}$ fires it gets the value
\[
\anyvalue_{0,\deviceId}\wedge\opFunOf{\fname}(\anyvalue_{\deviceId'},\anyvalue_{1,\deviceId},\ldots,\anyvalue_{n,\deviceId})
\]
where $\deviceId'\in\stableNodeSet$ is the neighbour of $\deviceId$ that has minimum value (among the neighbours of $\deviceId$).
This value is univocally determined by  $\netframeOf{\network}$ and is stable, since the values $\anyvalue_{\deviceId'}$ and $\anyvalue_{j,\deviceId}$ ($0\le j\le n$) are stable and  in any firing evolution  each neighbour of $\deviceId$  assumes only  values greater or equal to $\anyvalue_{\deviceId'}$.
\end{proof}

Note that conditions (i)-(iii) of the following lemma are exactly the same as in Lemma~\ref{lem-frontier}.
\begin{lem}[Minimum value not in  $\stableNodeSet$]\label{lem-minimum-value-outside}
Given a program $\e=\spreadThree{\e_0}{\fname}{\e_1,\ldots,\e_n}$ with valid  sort and stabilising diffusion assumptions, for every reachable pre-self-stable network configuration $\network$ with devices $\nodeSubSet$ and a non-empty subset of devices $\stableNodeSet\subset\nodeSubSet$ such that 
\begin{enumerate}[label=(\roman*)]
\item
each device in $\stableNodeSet$ is self-stable in $\network$;
\item
each device in $\nodeSubSet-\stableNodeSet$ has a value greater or equal to the values of the devices in $\stableNodeSet$ and, during any firing evolution, will always assume  values greater or equal to the values of the devices in $\stableNodeSet$; 
\item
$\frontierOfIn{\stableNodeSet}{\nodeSubSet}\not=\emptyset$; and
\item
each device in $\frontierOfIn{\stableNodeSet}{\nodeSubSet}$ is self-stable in $\network$; 
\end{enumerate}
if $\sourceNodeSubSet\subseteq\nodeSubSet-\stableNodeSet$ is the set of devices $\deviceId$  such that
$\anyvalueInNC{\deviceId}{\network}$ is minimum (among
the values of the devices in $\nodeSubSet-\stableNodeSet$), then
\begin{enumerate}
\item
if $\deviceId\in\sourceNodeSubSet$,  $\network\netArrIdStar\network'$ and $\deviceId\not\in\overline{\deviceId}$, then $\anyvalueInNC{\deviceId}{\network}=\anyvalueInNC{\deviceId}{\network'}$ is minimum (among
the values of the devices in $\nodeSubSet-\stableNodeSet$ in $\network'$);
\item
if $\sourceNodeSubSet\cap\frontierOfIn{\stableNodeSet}{\nodeSubSet}=\emptyset$, then there is a device $\deviceId\in\sourceNodeSubSet$ such that
either  $\anyvalueInNC{\deviceId}{\network}=\anyvalue_{0,\deviceId}$
    (i.e., $\deviceId$ is self-stable in $\network$) or 
if $\network\netArrFairStarN{1}\network'$ then there exists $\anyvalue'$ such that:
\begin{enumerate}
\item
    $\anyvalueInNC{\deviceId}{\network}\wfLt\anyvalue'\wfLe\anyvalueInNC{\deviceId}{\network'}$; and
\item if $\network'\netArrPresStar\network''$,
    then
    $\anyvalue'\wfLe\anyvalueInNC{\deviceId}{\network''}$.
\end{enumerate}
\end{enumerate}
\end{lem}
\begin{proof}
Since the self-stable values in $\frontierOfIn{\stableNodeSet}{\nodeSubSet}$ ensure that in any firing evolution the values of the devices in $\nodeSubSet-(\stableNodeSet\cup\frontierOfIn{\stableNodeSet}{\nodeSubSet})$ are computed without using  the values of the devices in $\stableNodeSet$, the proof is similar to the proof of Lemma~\ref{lem-minimum-value}.

\noindent
\textbf{Point (1).} Consider a device $\deviceId'\not\in\stableNodeSet\cup\sourceNodeSubSet$. Then $\anyvalueInNC{\deviceId}{\network}\wfLe\anyvalueInNC{\deviceId'}{\network}\wfLe\anyvalue_{0,\deviceId'}$  and its $m\ge 0$ neighbours have values $\anyvaluebis_j$  such that 
$\anyvalueInNC{\deviceId}{\network}\wfLe\anyvaluebis_j$ ($1\le j\le m$). Since the  stabilising diffusion assumptions hold, $\anyvaluebis_j\wfLt\opFunOf{\fname}(\anyvaluebis_j,\anyvalue_{1,\deviceId'},\ldots,\anyvalue_{n,\deviceId'})=\anyvaluealt_j$. Therefore,
 $\anyvalueInNC{\deviceId}{\network}\wfLe\lowerBound\{\anyvalue_{0,\deviceId'},\anyvaluealt_1,\myldots,\anyvaluealt_m\}=\anyvalueInNC{\deviceId'}{\network'}$. So $\anyvalueInNC{\deviceId}{\network'}=\anyvalueInNC{\deviceId}{\network}$ is minimum (among
the values of the devices of $\nodeSubSet-\stableNodeSet$ in $\network'$).

\noindent
\textbf{Point (2).} Assume that $\anyvalueInNC{\deviceId}{\network}\wfLt\anyvalue_{0,\deviceId}$. The $m\ge 0$ neighbours of $\deviceId$ have values $\anyvaluebis_j$  such that 
$\anyvalueInNC{\deviceId}{\network}\wfLe\anyvaluebis_j$ ($1\le j\le m$). Since the  stabilising diffusion assumptions hold, $\anyvalueInNC{\deviceId}{\network}\wfLt\opFunOf{\fname}(\anyvalueInNC{\deviceId}{\network},\anyvalue_{1,\deviceId},\ldots,\anyvalue_{n,\deviceId})=\anyvaluealt_0$ and  $\anyvaluealt_0\wfLe\opFunOf{\fname}(\anyvaluebis_j,\anyvalue_{1,\deviceId},\ldots,\anyvalue_{n,\deviceId})=\anyvaluealt_j$. Therefore,
the value $\anyvalue'=\anyvalue_{0,\deviceId}\wedge\anyvaluealt_0$ is such that 
when $\deviceId$ fires its new value $\anyvalueInNC{\deviceId}{\network'}=\lowerBound\{\anyvalue_{0,\deviceId},\anyvaluealt_1,\myldots,\anyvaluealt_m\}$ is such that $\anyvalueInNC{\deviceId}{\network}\wfLt\anyvalue'\wfLe\anyvalueInNC{\deviceId}{\network'}$. Moreover,  since  $\anyvalueInNC{\deviceId}{\network}\wfLt\anyvalueInNC{\deviceId}{\network'}$ we have that in a firing evolution  none of  the devices $\deviceId'\in\nodeSubSet-(\stableNodeSet\cup\{\deviceId\})$  will reach a value less than $\anyvalueInNC{\deviceId}{\network}$ and therefore the device $\deviceId$ will never reach a value less than $\anyvalue'$.
\end{proof}

Note that conditions \emph{(i)}-\emph{(iv)} of the following lemma are exactly the same as in Lemma~\ref{lem-minimum-value-outside}.
\begin{lem}[Self-stabilisation of the minimum value not in  $\stableNodeSet$]\label{lem-self-stabilisation-minimum-value-outside}
Given a program $\e=\spreadThree{\e_0}{\fname}{\e_1,\ldots,\e_n}$ with valid sort and stabilising diffusion assumptions, for every reachable pre-self-stable network configuration $\network$ with devices $\nodeSubSet$ and a non-empty subset of devices $\stableNodeSet\subset\nodeSubSet$ such that 
\begin{enumerate}[label=(\roman*)]
\item
each device in $\stableNodeSet$ is self-stable in $\network$;
\item
each device in $\nodeSubSet-\stableNodeSet$ has a value greater or equal to the values of the devices in $\stableNodeSet$ and, during any firing evolution, will always assume  values greater or equal to the values of the devices in $\stableNodeSet$; 
\item
$\frontierOfIn{\stableNodeSet}{\nodeSubSet}\not=\emptyset$; and
\item
each device in $\frontierOfIn{\stableNodeSet}{\nodeSubSet}$ is self-stable in $\network$; 
\end{enumerate}
there exists
$k\ge 0$ such that $\network\netArrFairStarN{k}\network'$
implies that there exists a device $\deviceId_1$ in  $\nodeSubSet-\stableNodeSet$ such that  $\stableNodeSet_1=\stableNodeSet\cup\{\deviceId_1\}$ satisfies the following conditions: 
\begin{enumerate}
\item
each device $\deviceId$  in $\stableNodeSet_1$ is self-stable in $\network'$; and
\item
in $\network'$ any device  in $\nodeSubSet-\stableNodeSet_1$ has a value greater or equal to the values of the devices in $\stableNodeSet_1$ and, during any firing evolution, will always assume  values greater than the values of the devices in $\stableNodeSet_1$.
\end{enumerate}
\end{lem}

\begin{proof}
Let $\sourceNodeSubSet\subseteq\nodeSubSet-\stableNodeSet$ be the set of devices $\deviceId$  such that
$\anyvalueInNC{\deviceId}{\network}$ is minimum (among
the values of the devices in $\nodeSubSet-\stableNodeSet$).  We consider two cases.
\begin{description}
\item[$\sourceNodeSubSet\cap\frontierOfIn{\stableNodeSet}{\nodeSubSet}\not=\emptyset$] Any of the devices $\deviceId_1\in\sourceNodeSubSet\cap\frontierOfIn{\stableNodeSet}{\nodeSubSet}$ is such that conditions (1) and (2) in the statement of the lemma are satisfied.
\item[$\sourceNodeSubSet\cap\frontierOfIn{\stableNodeSet}{\nodeSubSet}=\emptyset$] 
The number of devices in the network configuration is
finite, the environment does not change, the network is pre-self-stable, and the  stabilising diffusion condition holds. 
The results follows by  Lemma~\ref{lem-minimum-value-outside}.\footnote{In particular, since the self-stable values in $\frontierOfIn{\stableNodeSet}{\nodeSubSet}$ ensure that in any firing evolution the values of the devices in $\nodeSubSet-(\stableNodeSet\cup\frontierOfIn{\stableNodeSet}{\nodeSubSet})$ are computed without using  the values of the devices in $\stableNodeSet$, the proof of this case is similar to the proof of  Lemma~\ref{lem-self-stabilisation-minimum-value}  (by using Lemma~\ref{lem-minimum-value-outside} instead of Lemma~\ref{lem-minimum-value}).} Namely, if there is a device $\deviceId\in\sourceNodeSubSet$  such that  $\anyvalue_\deviceId\wfLt\anyvalue_{0,\deviceId}$, then after a 1-fair network evolution  $\deviceId$ reaches a value which is greater or equal to some $\anyvalue'$ such that
\begin{itemize}
\item
    $\anyvalue\wfLt\anyvalue'\wfLe\anyvalue_{0,\deviceId}$; and
\item in any subsequent firing evolution the value of $\deviceId$ will be always greater or equal to $\anyvalue'$. 
\end{itemize}
Therefore, after a finite number $k$ of 1-fair evolutions (i.e., after any $h$-fair evolution with $h\ge k$)  conditions (1) and (2) in the statement of the lemma are satisfied.\qedhere
\end{description}
\end{proof}

\begin{lem}[Pre-self-stable network self-stabilization]\label{lem-pre-self-stable-network-self-stabilization}
Given a program $\e=\spreadThree{\e_0}{\fname}{\e_1,\ldots,\e_n}$ with valid  sort and stabilising diffusion assumptions, for every reachable pre-self-stable network configuration $\network$ there exists
$k\ge 0$ such that $\network\netArrFairStarN{k}\network'$
implies that $\network'$ is self-stable.
\end{lem}
\begin{proof}
Let $\nodeSubSet$ be a set of devices of $\network$. The proof is by induction on the number of devices in $\nodeSubSet$.
\begin{description}
\item[Case $\nodeSubSet=\emptyset$] 
Immediate. 
\item[Case $\nodeSubSet\not=\emptyset$] 
By Lemma~\ref{lem-self-stabilisation-minimum-value} there exists $k_0\ge 0$ such that after any $k_0$-fair evolution there is a non-empty set of devices  $\stableNodeSet_1$ that satisfies 
conditions (1) and (2) in the statement of Lemma~\ref{lem-self-stabilisation-minimum-value} (and, therefore, also
conditions (1) and (2) in the statement of Lemma~\ref{lem-self-stabilisation-minimum-value-outside}).

Now rename $\stableNodeSet_1$ to $\stableNodeSet$, consider 
 a counter $c$ intially equal to $1$, and iterate the following two reasoning steps while the set of devices $\stableNodeSet$ is such that  $\frontierOfIn{\stableNodeSet}{\nodeSubSet}\not=\emptyset$:
\begin{itemize}
\item
By lemma~\ref{lem-frontier} and lemma~\ref{lem-self-stabilisation-minimum-value-outside}
there exists $k_c\ge 0$ such that after any $k_c$-fair evolution there is a non-empty set of devices  $\stableNodeSet_1$ that satisfies conditions (1) and (2) in the statement of Lemma~\ref{lem-self-stabilisation-minimum-value-outside}. 
\item
Rename $\stableNodeSet_1$ to $\stableNodeSet$ and increment the value of $c$.
\end{itemize} 
Since at each iteration the number devices in  $\stableNodeSet$ is increased by one, the number of iterations is finite (note that  $\stableNodeSet=\nodeSubSet$ implies $\frontierOfIn{\stableNodeSet}{\nodeSubSet}=\emptyset$).
After the last iteration, we have proved that there exists $k'=\Sigma_{j=0}^{c-1} k_j$ such that after any $k'$-fair evolution there is a non-empty set of devices  $\stableNodeSet$ that satisfies conditions (1) and (2) in the statement of Lemma~\ref{lem-self-stabilisation-minimum-value-outside}. 
Since $\frontierOfIn{\stableNodeSet}{\nodeSubSet}=\emptyset$,
then the evolution of the devices in
$\nodeSubSet-\stableNodeSet$ is independent from the
devices in $\stableNodeSet$. By induction there exists
$k''\ge 0$ such that after any $k''$-fair evolution the portion of the network with devices in $\nodeSubSet-\stableNodeSet$ is self-stable.  Therefore, we have proved the lemma with $k=k'+k''$.\qedhere
\end{description}
\end{proof}

\noindent\textbf{Restatement of
Theorem~\ref{the-network-self--stabilization}} (Network self-stabilisation for programs that satisfy the stabilising-diffusion condition)\textbf{.}
\emph{Given a program with valid  sort and stabilising diffusion assumptions,  
every reachable network configuration $\network$ self-stabilises, i.e.,
there exists
$k\ge 0$ such that $\network\netArrFairStarN{k}\network'$
implies that $\network'$ is self-stable.}


\begin{proof}
By induction on the syntax of closed expressions $\e$ and on
the number of function calls that may be encountered during the
evaluation of $\e$. Let $\netframe=\netframeOf{\network}$.
\begin{description}
\item[Case $\anyvalue$] Any device fire produces the
    value-tree $\anyvalue()$ (independently from
    $\netframe$). 
\item[Case $\snsname$] Each fire of device $\deviceId$
    produces the value-tree $\anyvalue()$, where
    $\anyvalue=\snsFunFor{\deviceId}(\snsname)$  is
    univocally determined by $\netframe$.
\item[Case $\oname(\overline{\e})$] Straightforward by
    induction.
\item[Case $\fname(\overline{\e})$] Straightforward by
    induction.
\item[Case $\spreadThree{\e_0}{\fname}{\e_1,\ldots,\e_n}$]
    By induction there exists $h\ge 0$ such that if
    $\network\netArrFairStarN{h}\network_1$ then on every
    device $\deviceId$, the evaluation of
    $\e_0,\e_1,\ldots,\e_n$ produce stable value-trees
    $\vtree_{0,\deviceId},\vtree_{1,\deviceId},\ldots,\vtree_{n,\deviceId}$,
    which are univocally determined by $\netframe$. 
    Note that, if $\network\netArrFairStarN{h+1}\network_2$
    then  $\network_2$ is
    pre-self-stable. Therefore the result follows straightforwardly by Lemma~\ref{lem-pre-self-stable-network-self-stabilization}.\qedhere
\end{description}
\end{proof}

%

\section{ Proof of Theorem~\ref{the-DeviceSortPreservation}}\label{app-proof-sortSounness}

The proof of Theorem~\ref{the-DeviceSortPreservation} is similar to the proof of Theorem~\ref{the-DeviceTypePreservation} (cf.\ Appendix~\ref{app-proof-typeSounness})---see Remark~\ref{rem:TypeVsSortChecking}.

\begin{lem}[Substitution lemma for sorting]\label{lem:substitutionForSorting}
If
$\surExpTypJud{\overline{\xname}:\overline{\rtype}}{\e}{\rtype}$,
$\lengthOf{\overline{\anyvalue}}=\lengthOf{\overline{\xname}}$,
        $\surExpTypJud{\emptyset}{\overline{\anyvalue}}{\overline{\rtype}'}$,
and
$\overline{\rtype}'\rsub\overline{\rtype}$,
then
$\surExpTypJud{\emptyset}{\applySubstitution{\e}{\substitution{\overline{\xname}}{\overline{\anyvalue}}}}{\rtype'}$ for some $\rtype'$ such that  $\rtype'\rsub\rtype$.
\end{lem}

\begin{proof}
Straightforward by induction on the application of the sort-checking rules for expressions in Fig.~\ref{fig:RefinementTyping}.
\end{proof}

\begin{lem}[Device computation
sort preservation]\label{lem:DeviceSortPreservation}
If
$\vtree\in\bsWSVT{\overline{\xname}:\overline{\rtype}}{\e}{\rtype}$,
then $\surExpTypJud{\emptyset}{\vrootOf{\vtree}}{\rtype'}$ for some $\rtype'$ such that  $\rtype'\rsub\rtype$.
\end{lem}

\begin{proof}
Recall that the sorting rules (in Fig.~\ref{fig:RefinementTyping}) and the evaluation rules  (in Fig.~\ref{fig:deviceSemantics} and Fig.~\ref{fig:deviceSemantics-extended}) are syntax directed.
The proof is by induction on the definition of $\bsWSVT{\overline{\xname}:\overline{\rtype}}{\e}{\rtype}$ (given in Section~\ref{sec-properties-sorting}), on
the number of user-defined function calls that may be encountered during the
evaluation of $\applySubstitution{\e}{\substitution{\overline{\xname}}{\overline{\anyvalue}}}$ (cf.\ sanity condition \emph{\textbf{(iii)}} in
Section~\ref{sec-typing}), and on the syntax of closed expressions.

From the hypothesis $\vtree\in\bsWSVT{\overline{\xname}:\overline{\rtype}}{\e}{\rtype}$ we have $\surExpTypJud{\overline{\xname}:\overline{\rtype}}{\e}{\rtype}$, $\bsopsem{\snsFun}{\overline{\vtree}}{\e}{\vtree}$  for some  sensor mapping $\snsFun$,  evaluation trees $\overline{\vtree}\in\bsWSVT{\overline{\xname}:\overline{\rtype}}{\e}{\rtype}$, and values $\overline{\anyvalue}$ such that $\lengthOf{\overline{\anyvalue}}=\lengthOf{\overline{\xname}}$,
        $\surExpTypJud{\emptyset}{\overline{\anyvalue}}{\overline{\rtype}'}$,
and
$\overline{\rtype}'\rsub\overline{\rtype}$. The case $\overline{\vtree}$ empty represents the base of the induction on the definition of $\bsWSVT{\overline{\xname}:\overline{\rtype}}{\e}{\rtype}$. Therefore the rest of this proof can be understood as a proof of the base step by assuming  $\overline{\vtree}=\emptyset$ and a proof of the inductive step  by assuming $\overline{\vtree}\not=\emptyset$.
Moreover, by Lemma~\ref{lem:substitutionForSorting} we have that $\surExpTypJud{\emptyset}{\applySubstitution{\e}{\substitution{\overline{\xname}}{\overline{\anyvalue}}}}{\rtype'}$ for some $\rtype'$ such that  $\rtype'\rsub\rtype$.

The case when $\e$ does non contain user-defined function calls represents the base on the induction on the number of user-defined function calls that may be encountered during the
evaluation of $\applySubstitution{\e}{\substitution{\overline{\xname}}{\overline{\anyvalue}}}$. Therefore the rest of this proof can be understood as a proof of the base step by ignoring the cases $\applySubstitution{\e}{\substitution{\overline{\xname}}{\overline{\anyvalue}}}=\fname(\e_1,\myldots,\e_n)$ and $\applySubstitution{\e}{\substitution{\overline{\xname}}{\overline{\anyvalue}}}=\spreadThree{\e_0}{\fname}{\e_1,\myldots,\e_n}$ when $\fname$ is a used-defined function $\dname$.
The base of the induction on $\applySubstitution{\e}{\substitution{\overline{\xname}}{\overline{\anyvalue}}}$ consist of two cases.
\begin{description}
\item[Case $\snsname$] 
From the hypothesis we have  $\surExpTypJud{\emptyset}{\snsname}{\rtype}$ where $\rtype=\rtypeof{\snsname}$ (by rule \ruleNameSize{[T-SNS]}) and $\bsopsem{\senstate}{\overline\vtree}{\snsname}{\vtree}$ where $\vtree=\anyvalue()$ and $\anyvalue=\senstate(\snsname)$ (by rule \ruleNameSize{[E-SNS]}). Since the sensor  $\snsname$ returns values of sort $\rtypeof{\snsname}$, we have that $\rtypeof{\anyvalue}\rsub\rtypeof{\snsname}=\rtype$. So  the result follows by a straightforward induction of the syntax of values using rules \ruleNameSize{[S-VAL]} and \ruleNameSize{[S-PAIR]}. 
\item[Case $\anyvalue$]
From the hypothesis we have  $\surExpTypJud{\emptyset}{\anyvalue}{\rtype}$  and  (by rule \ruleNameSize{[E-VAL]}) $\bsopsem{\senstate}{\overline\vtree}{\anyvalue}{\vtree}$ where  $\vtree=\anyvalue()$.  So  the result follows by a straightforward induction of the syntax of values using rules \ruleNameSize{[S-VAL]} and \ruleNameSize{[S-PAIR]}.
\end{description}
For the inductive step on $\applySubstitution{\e}{\substitution{\overline{\xname}}{\overline{\anyvalue}}}$, we show only the two most interesting cases (all the other cases are straightforward by induction). 
\begin{description}
\item[Case $\dname(\e_1,\myldots,\e_n)$] 
From the  hypothesis we have $ \surExpTypJud{\emptyset}{\fname(\e_1,\myldots,\e_n)}{\rtype}$ (by rule \ruleNameSize{S-FUN]}) and $\bsopsem{\senstate}{\overline\vtree}{\dname(\e_1,\myldots,\e_n)}{\anyvalue(\vtree'_1,\myldots,\vtree'_n,\anyvalue(\overline\vtreealt))}$ (by rule \ruleNameSize{E-DEF]}). Therefore we have
$\rtype(\rtype_1,\myldots,\rtype_n)=\mostSpecificOf{\rsignaturesOf{\fname}}{\rtype'_1,\ldots,\rtype'_n} $,  $\surExpTypJud{\emptyset}{\e_1}{\rtype'_1}$, $\myldots$, $\surExpTypJud{\emptyset}{\e_n}{\rtype'_n}$ and $\rtype'_1\rsub\rtype_1$, $\myldots$, $\rtype'_n\rsub\rtype_n$ (by the premises of rule \ruleNameSize{S-FUN]}) and $\defK \; \anytype \;\dname(\anytype_1\;\xname_1,\myldots,\anytype_n\;\xname_n) = \e''$, $\bsopsem{\senstate}{\premiseNumOf{1}{\overline\vtree}}{\e_1}{\vtree'_1}$, $\myldots$,
$\bsopsem{\senstate}{\premiseNumOf{n}{\overline\vtree}}{\e_n}{\vtree'_n}$ and
\begin{equation}\label{eq1:the-DeviceSortPreservation}
\bsopsem{\senstate}{\premiseNumOf{n+1}{\overline\vtree}}{\e'}{\anyvalue(\overline\vtreealt)}
\quad \mbox{where} \; \e'=\applySubstitution{\e''}{\substitution{\xname_1}{\vrootOf{\vtree'_1}},\myldots,\substitution{\xname_n}{\vrootOf{\vtree'_n}}}
\end{equation}
(the premises of rule \ruleNameSize{E-DEF]}).

Since $\premiseNumOf{i}{\overline\vtree}\in\bsWSVT{\emptyset}{\e_i}{\rtype'_i}$ ($1\le i\le n$), then $\vtree'_i\in\bsWSVT{\emptyset}{\e_i}{\rtype'_i}$; therefore, by  induction we have $\surExpTypJud{\emptyset}{\vrootOf{\vtree'_1}}{\rtype''_1}$, $\myldots$, $\surExpTypJud{\emptyset}{\vrootOf{\vtree''_n}}{\rtype''_n}$ and $\rtype''_1\rsub\rtype'_1\rsub\rtype_1$, $\myldots$, $\rtype''_n\rsub\rtype'_n\rsub\rtype_n$. 

Since the program is well sorted (cf.\ Section~\ref{sec-ref-typing}) we have  $\surExpTypJud{\xname_1:\rtype_1,\myldots,\xname_n:\rtype_n}{\e''}{\rtype'}$ and $\rtype'\rsub\rtype$ (by rule \ruleNameSize{T-DEF]}). 


Since $\premiseNumOf{n+1}{\overline\vtree}\in\bsWFVT{\xname_1:\rtype_1,\myldots,\xname_n:\rtype_n}{\e''}{\rtype'}$, then (by~(\ref{eq1:the-DeviceSortPreservation})) we have $\anyvalue(\overline\vtreealt)\in\bsWFVT{\xname_1:\rtype_1,\myldots,\xname_n:\rtype_n}{\e''}{\rtype'}$; therefore, by induction we have that $\surExpTypJud{\emptyset}{\anyvalue}{\rtype''}$ with $\rtype''\rsub\rtype'\rsub\rtype$.

\item[Case $\spreadThree{\e_0}{\fname}{\e_1,\myldots,\e_n}$] 
From the  hypothesis we have $\surExpTypJud{\emptyset}{\spreadThree{\e_0}{\fname}{\e_1,\myldots,\e_n}}{\rtype'}$ (by rule \ruleNameSize{S-SPR]}) and $\bsopsem{\senstate}{\overline\vtree}{\spreadThree{\e_0}{\fname}{\overline{\e}}}{\lowerBound\{\anyvalue_0,\anyvaluealt_1,\myldots,\anyvaluealt_m\}(\vtreealt_0,\vtreealt_1,\myldots,\vtreealt_n)}$ (by rule \ruleNameSize{E-SPR]}). Therefore we have
 $\progressionOpOf{\fname}$, $\surExpTypJud{\RefTypEnv}{\e_0\overline{\e}}{\rtype'_0\overline{\rtype}'}$,
$\rtype'(\rtype_0\overline{\rtype})=\mostSpecificOf{\stabilisingSignaturesOf{\fname}}{\rtype'_0\overline{\rtype}'}$,
$\rtype=\refSupOf{\rtype'_0}{\rtype'}$ and $\rtype'_0\overline{\rtype}'\rsub\rtype_0\overline{\rtype}$ (by the premises of rule \ruleNameSize{S-SPR]}) and $\bsopsem{\senstate}{\premiseNumOf{1}{\overline\vtree}}{\e_1}{\vtree'_1}$, $\myldots$,
$\bsopsem{\senstate}{\premiseNumOf{n}{\overline\vtree}}{\e_n}{\vtree'_n}$, 
$\conclusionOf{\vtreealt_0,\myldots,\vtreealt_n} = \anyvalue_0 \myldots \anyvalue_n$,
$\conclusionOf{\overline\vtree} = \anyvaluebis_1 \myldots \anyvaluebis_m$,
\begin{equation}\label{eq3:the-DeviceSortPreservation}
\bsopsem{\senstate}{\emptyset}{\fname(\anyvaluebis_1,\anyvalue_1,\myldots,\anyvalue_n)}{\anyvaluealt_1(\cdots)},
\myldots,
\bsopsem{\senstate}{\emptyset}{\fname(\anyvaluebis_m,\anyvalue_1,\myldots,\anyvalue_n)}{\anyvaluealt_m(\cdots)}
\end{equation}
(the premises rule \ruleNameSize{E-SPR]}).
By induction we have $\surExpTypJud{\emptyset}{ \anyvalue_0 \myldots \anyvalue_n}{\rtype''_0 \rtype''_1 \myldots \rtype''_n}$
with $\rtype''_0 \rtype''_1 \myldots \rtype''_n\rsub\rtype'_0 \rtype'_1 \myldots \rtype'_n$
and  $\surExpTypJud{\emptyset}{\anyvaluebis_1}{\rtype'''_1}$, $\myldots$, $\surExpTypJud{\emptyset}{\anyvaluebis_m}{\rtype'''_m}$
with $\rtype'''_l\rsub\rtype\rsub\rtype_0$ $(1\le l\le m$).
We have two subcases.
\begin{itemize}
\item
If $\fname$ is a user-defined function, then from~(\ref{eq3:the-DeviceSortPreservation}) we get (by reasoning as in the proof of case $\dname(\e_1,\myldots,\e_n)$)
$\surExpTypJud{\emptyset}{\anyvaluealt_1}{\rtype''''_1}$, $\myldots$, $\surExpTypJud{\emptyset}{\anyvaluealt_m}{\rtype''''_m}$ for some $\rtype''''_l$ such that $\rtype''''_l\rsub\rtype'$ $(1\le l\le m$).
\item
If $\fname$ is a built-in function, then from~(\ref{eq3:the-DeviceSortPreservation}) we get (by the semantics of built-in functions) $\surExpTypJud{\emptyset}{\anyvaluealt_1}{\rtype''''_1}$, $\myldots$, $\surExpTypJud{\emptyset}{\anyvaluealt_m}{\rtype''''_m}$ for some $\rtype''''_l$ such that $\rtype''''_l\rsub\rtype'$ $(1\le l\le m$).
\end{itemize}
In both cases $\anyvalue=\lowerBound\{\anyvalue_0,\anyvaluealt_1,\myldots,\anyvaluealt_m\}$ has a sort $\rtype''$ with $\rtype'''\rsub\rtype$, i.e., $\surExpTypJud{\emptyset}{\anyvalue}{\rtype'''}$ with $\rtype'''\rsub\rtype$ holds.\qedhere
\end{description}
\end{proof}

\noindent\textbf{Restatement of
Theorem~\ref{the-DeviceSortPreservation} (Device computation
sort preservation).} \emph{If
$\surExpTypJud{\overline{\xname}:\overline{\rtype}}{\e}{\rtype}$,
 $\snsFun$ is a sensor mapping,
$\overline{\vtree}\in\bsWSVT{\overline{\xname}:\overline{\rtype}}{\e}{\rtype}$,
$\lengthOf{\overline{\anyvalue}}=\lengthOf{\overline{\xname}}$,
        $\surExpTypJud{\emptyset}{\overline{\anyvalue}}{\overline{\rtype}'}$,
$\overline{\rtype}'\rsub\overline{\rtype}$,
and
$\bsopsem{\snsFun}{\overline{\vtree}}{\applySubstitution{\e}{\substitution{\overline{\xname}}{\overline{\anyvalue}}}}{\vtree}$,
then $\surExpTypJud{\emptyset}{\vrootOf{\vtree}}{\rtype'}$ for some $\rtype'$ such that  $\rtype'\rsub\rtype$.}

\begin{proof}
Straightforward by Lemma~\ref{lem:DeviceSortPreservation}, since $\vtree\in\bsWSVT{\overline{\xname}:\overline{\rtype}}{\e}{\rtype}$.
\end{proof}

\section{ Proof of Theorem~\ref{the-AnnotationSoundness}}\label{app-proof-annotationSoundness}

\begin{lem}[Annotated sort of an expression]\label{lem:annotatedSort}
If
$\annExpTypJud{\xname_1:\rtype_1\inputAnnSB,\,\overline{\xname}:\overline{\rtype}}{\e}{\rtype\pAnnSB}$,
then $\surExpTypJud{\xname:\rtype_1,\,\overline{\xname}:\overline{\rtype}}{\e}{\rtype}$ and 
$\topValOf{\leftKeyOf{\rtype_1}}=\topValOf{\leftKeyOf{\rtype}}$.
\end{lem}

\begin{proof}
Straightforward by induction on the application of the annotated sort checking rules for expressions in Fig.~\ref{fig:AnnotatedTyping}.
\end{proof}

A pure expression $\e$ with free variables $\overline{\xname}$ of sorts $\overline{\rtype}$  represents  the pure function
 that for every $\overline{\anyvalue}\in\semOf{\overline{\rtype}}$ returns the value $\semOf{\applySubstitution{\e}{\substitution{\overline{\xname}}{\overline{\anyvalue}}}}$. In the following we will write $\pureFunOf{\e}$ to denote such a function.

\begin{lem}[Annotation soundness for expressions]\label{lem:annotationSoundness}
If
$\annExpTypJud{,\xname_1:\rtype_1\inputAnnSB,\,\overline{\xname}:\overline{\rtype}}{\e}{\rtype\pAnnSB}$ and $\overline{\anyvalue}\in\semOf{\overline{\rtype}}$,
then 
\begin{enumerate}
    \item
     if $\pAnn=\PpAnn$ then
 \begin{itemize}
\item
 $\anyvalue\keywfLtOf{\rtype_1}\anyvalue'$ and $\opApply{\opFunOf{\pureFunOf{\e}}}{\anyvalue,\overline{\anyvalue}}=\anyvalue''\not\keywfEqOf{\rtype}\topValOf{\rtype}$ imply
    implies
    $\anyvalue''\keywfLtOf{\rtype}\opApply{\opFunOf{\pureFunOf{\e}}}{\anyvalue',\overline{\anyvalue}}$;
\item
for all
        $\anyvalue\in\semOf{\rtype_1}-\{\topValOf{\rtype_1}\}$,
        $\leftKeyOf{\anyvalue}\wfLtOf{\leftKeyOf{\rtype_1}}\leftKeyOf{\opApply{\opFunOf{\pureFunOf{\e}}}{\anyvalue,\overline{\anyvalue}}}$;
\end{itemize}
    \item
     if $\pAnn\in\{\PpAnn,\WpAnn\}$ then
\begin{itemize}
\item
 $\anyvalue\keywfLeOf{\rtype_1}\anyvalue'$
    implies
    $\opApply{\opFunOf{\pureFunOf{\e}}}{\anyvalue,\overline{\anyvalue}}\keywfLeOf{\rtype}\opApply{\opFunOf{\pureFunOf{\e}}}{\anyvalue',\overline{\anyvalue}}$;
\item 
for all
        $\anyvalue\in\semOf{\rtype_1}$,
     $\leftKeyOf{\anyvalue}\wfLeOf{\leftKeyOf{\rtype_1}}\leftKeyOf{\opApply{\opFunOf{\pureFunOf{\e}}}{\anyvalue,\overline{\anyvalue}}}$.
 \end{itemize}
\end{enumerate}
\end{lem}

\begin{proof}
By Lemma~\ref{lem:annotatedSort}
$\semOf{\e}$ has sort $\rtype(\rtype_1\overline{\rtype})$ and either $\rsubwrhof{\leftKeyOf{\rtype_1}}{\leftKeyOf{\rtype}}$ or $\rsubwrhof{\leftKeyOf{\rtype}}{\leftKeyOf{\rtype_1}}$.
Recall that the annotated sort checking rules (in Fig.~\ref{fig:AnnotatedTyping}) and the evaluation rules  (in Fig.~\ref{fig:deviceSemantics} and Fig.~\ref{fig:deviceSemantics-extended}) are syntax directed. 
By induction on the syntax of pure expressions $\e$.
The base of the induction on $\e$ consist of two cases.
\begin{description}
\item[Case $\xname$] 
Immediate by rule \ruleNameSize{[A-VAR]}.
\item[Case $\anyvalue$]
Straightforward by rule \ruleNameSize{[A-GVAL]} and Proposition~\ref{prop:AnnotatedSortsForGroundValues}.
\end{description}
For the inductive step on $\e$, we show only the case for function application (all the other cases are straightforward by induction). 
\begin{description}
\item[Case $\fname(\e_1,\overline{\e})$] 
Then the premises of rule \ruleNameSize{[A-FUN]}:
\begin{itemize}
\item
$ \annExpTypJud{\AnnTypEnv}{\e_1}{\rtype'_1\SB{\pAnn''}}$
\item
$\surExpTypJud{\annErasure{\AnnTypEnv}}{\overline{\e}}{\overline{\rtype}'}$
\item
$\rtype(\rtype''_1\overline{\rtype}'')\SB{\pAnn'}\in\mostSpecificOf{\asignaturesOf{\fname}}{\rtype'_1\overline{\rtype'}}$
\end{itemize}
hold and $\pAnn=\pAnn'(\pAnn'')$. By the last premise, we have that
$\pOfSCRIPT{'}{\fname}{\rtype(\rtype''_1\overline{\rtype}'')}$ holds. Then, the result follows straightforward by induction (using Definition~\ref{def:MuPiProgression}).\qedhere
\end{description}
\end{proof}

\noindent\textbf{Restatement of
Theorem~\ref{the-AnnotationSoundness} (Annotation soundness).} \emph{If
$\annFunTypJud{}{\FUNCTION}{\overline{\rtype(\overline{\rtype})\pAnnSB}}$ holds, then $\pOf{\fname}{\rtype(\overline{\rtype})}$ holds for all 
$\rtype(\overline{\rtype})\pAnnSB\in\overline{\rtype(\overline{\rtype})\pAnnSB}$.}

\begin{proof}
Straighforward from rule \ruleNameSize{[A-DEF]} in  Fig.~\ref{fig:RefinementTyping} using Lemma~\ref{lem:annotationSoundness} and Proposition~\ref{prop:Sounness of annotated subsigning}.
\end{proof}

\end{document}